\documentclass[a4paper]{article}

\usepackage{stmaryrd} 
\usepackage{amsfonts}
\usepackage{amssymb}
\usepackage{amsmath}
\usepackage{amsthm}
\usepackage{scalerel}
\usepackage{xspace}
\usepackage{paralist}
\usepackage{bm}
\usepackage{mfirstuc}
\usepackage{mdframed}
\usepackage{url}
\usepackage{tablefootnote}
\usepackage{scalerel}
\usepackage{relsize}
\usepackage{bussproofs}

\renewcommand{\ln}{\log}

\newtheorem{theorem}{Theorem}

\newtheorem{definition}[theorem]{Definition}
\newtheorem{lemma}[theorem]{Lemma}
\newtheorem{proposition}[theorem]{Proposition}

\newtheorem{remark}[theorem]{Remark}
\newtheorem{example}[theorem]{Example}

\newtheorem{invariant}[theorem]{Invariant}

\newcommand{\interv}[2]{\{ #1,\dots,#2 \}}
\newcommand{\emp}{\mathtt{emp}}
\newcommand{\ar}[1]{\#(#1)}
\newcommand{\fv}[1]{\mathit{fv}\left(#1\right)}
\newcommand{\dom}[1]{\mathit{dom}(#1)}

\newcommand{\img}[1]{\mathit{img}(#1)}
\newcommand{\dunion}{\uplus}
\newcommand{\vars}{{\cal V}}
\newcommand{\size}[1]{\mathit{size}(#1)}
\newcommand{\card}[1]{\mathit{card}(#1)}
\newcommand{\len}[1]{\|#1\|}
\newcommand{\bigO}{\mathcal{O}} 
\newcommand{\id}{{\mathit id}}

\renewcommand{\vec}[1]{\mathbf{#1}}

\newcommand{\exptime}{$\mathsf{EXPTIME}$\xspace}

\newcommand{\repl}[3]{#1\{ #2 \leftarrow #3 \}}
\newcommand{\replall}[4]{#1\{ #2 \leftarrow #3 \mid #4 \}}

\newcommand{\pcSID}{pc-SID\xspace}
\newcommand{\pcSIDs}{pc-SIDs\xspace}

\newcommand{\iseq}{\approx}

\newcommand{\Loc}{{\cal L}}

\newcommand{\astore}{\mathfrak{s}}
\newcommand{\aheap}{\mathfrak{h}}

\newcommand{\swand}{
\mathrel{\mbox{$\hspace*{-0.03em}\mathord{-}\hspace*{-0.4em}
 \mathord{-}\hspace*{-0.36em}\scalebox{0.9}{$\mathord{\bullet}$}$\hspace*{-0.005em}}}
}

\newcommand{\wand}{
\mathrel{\mbox{$\hspace*{-0.03em}\mathord{-}\hspace*{-0.4em}
 \mathord{-}\hspace*{-0.36em}\scalebox{0.9}{$\mathord{\ast}$}$\hspace*{-0.005em}}}
}
\newcommand{\unfoldto}[1]{\Leftarrow_{#1}}
\newcommand{\preds}{{\cal P}_S}

\newcommand{\asid}{{\cal R}}
\newcommand{\rank}{\kappa}
\newcommand{\alloccompatible}{$\allocf$-compatible\xspace}
\newcommand{\alloccompatibility}{$\allocf$-compatibility\xspace}

\newcommand{\locs}[1]{\mathit{loc}(#1)}
\newcommand{\roots}[1]{\mathit{roots}(#1)}
\newcommand{\rootsl}[1]{\mathit{roots}_l(#1)}
\newcommand{\rootsr}[1]{\mathit{roots}_r(#1)}

\newcommand{\bigAnd}{\mathop{\scaleobj{1.5}{*}}}

\newcommand{\modelst}{\models_{\theory}}
\newcommand{\modelsi}{\models_{\theory}^i}
\newcommand{\modelsr}{\models_{\asid}}
\newcommand{\equivr}{\equiv_{\asid}}
\newcommand{\modelssid}[1]{\models_{#1}}

\newcommand{\allocf}{\mathit{alloc}}
\newcommand{\alloc}[1]{\allocf(#1)}
\newcommand{\vdashr}{\vdash_{\asid}}
\newcommand{\vdashsid}[1]{\vdash_{#1}}

\newcommand{\MWf}[2]{#1[#2]}
\newcommand{\MWa}[3]{\Phi_{#1 \swand #2}^{#3}}

\newcommand{\MWl}[4]{\MWf{\MWa{#1}{#2}{#3}}{#4}}
\newcommand{\MWs}[4]{\MWl{#1}{#2}{#3}{#3#4}}

\newcommand{\MW}[2]{#1 \swand #2}

\newcommand{\theory}{{\cal T}}
\newcommand{\tformula}{$\theory$-formula\xspace}
\newcommand{\tatom}{$\theory$-atom\xspace}
\newcommand{\slformula}{$SL$-formula\xspace}

\newcommand{\tswand}{\textsc{pu}}
\newcommand{\Watom}{$\tswand$-atom\xspace}

\newcommand{\tpreds}{{\cal P}_{\cal T}}

\newcommand{\qprenex}{quasi-prenex\xspace}
\newcommand{\widt}[1]{\mathit{width}(#1)}

\newcommand{\mnset}[1]{\{#1\}}

\newcommand{\aform}{\phi}
\newcommand{\aformB}{\psi}
\newcommand{\aformC}{\gamma}

\newcommand{\atform}{\chi}
\newcommand{\atformB}{\xi}
\newcommand{\anatom}{\alpha}
\newcommand{\atail}{\beta}
\newcommand{\aseq}{\Gamma}
\newcommand{\aseqB}{\Delta}

\newcommand{\asetloc}{L}

\newcommand{\InlineRuleCond}[1]{\hfill{\text{#1}}}
\newcommand{\RuleCond}[1]{\begin{center}
		\begin{minipage}[c]{.85\textwidth}
			#1
		\end{minipage}
	\end{center}}

\newcommand{\lvdash}{\; \vdashr \;}

\newcommand{\lse}{\mathtt{mls}}
\newcommand{\ls}{\mathtt{ls}}
\newcommand{\nil}{\mathrm{nil}}
\newcommand{\myfalse}{\mathtt{false}}

\newcommand{\TabRule}[3]{
#1: \begin{tabular}{c}
$#2$ \\
\hline
$#3$
\end{tabular}
}

\newcommand{\BigCondInfRule}[4]{
\TabRule{#1}{#2}{#3}
\RuleCond{#4}
\vspace{4mm}
}

\newcommand{\CondInfRule}[4]{
\TabRule{#1}{#2}{#3}
\InlineRuleCond{#4}
\vspace{4mm}}

\begin{document}

\title{A Proof Procedure For Separation Logic With Inductive Definitions and Theory Reasoning}

\author{Mnacho Echenim and Nicolas Peltier}

\newcommand{\myabstract}
{
A proof procedure, in the spirit of the sequent calculus, 
is proposed to check the validity of entailments between
Separation Logic formulas combining 
inductively defined predicates denoting structures of bounded tree width
and 
theory reasoning. 
The calculus is sound and complete, in the sense that a sequent is valid iff it admits a (possibly infinite) proof tree.
We also show that the procedure terminates in the two following cases:
(i) When the inductive rules that define the predicates occurring on the left-hand side of the entailment terminate, in which case the 
proof tree is always finite.
(ii) When the theory is empty, in which case every valid sequent admits a rational  
proof tree, where the total number of pairwise distinct sequents occurring in the proof tree is doubly exponential w.r.t.\ the size of the end-sequent.
}
\maketitle
\abstract{\myabstract}

\section{Introduction}

\newcommand{\lst}{\mathtt{ls}}
\newcommand{\slst}{\mathtt{ils}}
\newcommand{\alst}{\mathtt{als}}

\newcommand{\twoexptime}{$2$-$\mathsf{EXPTIME}$}

Separation Logic (SL) \cite{IshtiaqOHearn01,Reynolds02},  is a well-established framework for 
reasoning on programs manipulating pointer-based data structures. 
It forms the basis of several industrial-scale static program
analyzers
\cite{DBLP:conf/nfm/CalcagnoDDGHLOP15,DBLP:conf/cav/BerdineCI11,DBLP:conf/cav/DudkaPV11}.
The logic uses specific connectives to assert that formulas are satisfied on disjoint parts of the memory,
which allows for more concise and more natural specifications.
Recursive data structures are specified using \emph{inductively defined predicates}, which provide a specification
mechanism similar to the definition of a recursive data type in an
imperative programming language.
Many verification tasks boil down to 
 checking entailments between formulas built on such atoms.
More precisely, the logic may be used to express pre- or post-conditions describing 
 the shape of the data structures (linked lists, trees, doubly linked lists, etc.) manipulated by the program, 
 and to express structural integrity properties, such as acyclicity of linked lists, absence of dangling pointers, etc.
Investigating the entailment problem for SL formulas is thus of theoretical and practical interest.
In practice, it is essential to offer as much flexibility as possible, and to
handle a wide class of user-defined data structures (e.g., doubly linked lists, trees with child-parent links, trees with chained leaves etc.), possibly involving external theories, such as arithmetic.
In general, the entailment problem is undecidable 
for formulas
containing inductively defined predicates
\cite{DBLP:conf/atva/IosifRV14}, 
and a lot of effort has been devoted to identifying decidable fragments and devising proof procedures, see e.g., 
\cite{berdine-calcagno-ohearn04,CalcagnoYangOHearn01,cook-haase-ouaknine-parkinson-worell11,spen,DBLP:journals/fmsd/EneaLSV17,DemriGalmicheWendlingMery14}.
In particular, a  general class of  decidable  entailment problems is described in
\cite{IosifRogalewiczSimacek13}. It is based on the decidability of the satisfiability
problem for monadic second order logic over graphs of a bounded treewidth, for formulas involving no theory other than equality.
This class is defined by restricting the form of the inductive rules, which must fulfill $3$ conditions, formally defined below: the {\em progress} condition (every rule allocates a single memory location), 
the {\em connectivity} condition (the set of allocated locations has a tree-shaped structure)
and the {\em establishment} condition (every existentially quantified variable is eventually allocated). 
More recently, a
\twoexptime\ algorithm was proposed for such entailments
\cite{PMZ20}. 
In \cite{DBLP:conf/lpar/EchenimIP20}
 we showed that this bound is 
tight and in \cite{EIP21a} we devised a new algorithm, handling more general  classes of inductive definitions.
The algorithms in \cite{PMZ20,DBLP:conf/lpar/EchenimIP20} work by computing some abstract representation of the set of models of SL formulas.
The
abstraction is precise enough to allow checking that all the models of the left-hand side are also models of the right-hand side, 
also general enough to ensure termination of the entailment checking
algorithm. 
Other approaches have been proposed to check entailments in various fragments, see e.g., \cite{cook-haase-ouaknine-parkinson-worell11,DBLP:journals/fmsd/EneaLSV17,DBLP:conf/atva/IosifRV14}.
 In particular, a sound and complete proof procedure is given in \cite{DBLP:conf/aplas/TatsutaNK19} for inductive rules satisfying conditions that are strictly more restrictive than those in \cite{IosifRogalewiczSimacek13}.
In \cite{DBLP:journals/logcom/GalmicheM21} a labeled proof systems is presented for separation logic formulas
handling arbitrary inductive definitions and all connectives (including negation and separated implication).
Due to the expressive power of the considered language, this proof system is of course not terminating or complete in general.

In the present paper,
we extend these results by defining a proof procedure in the style of sequent calculi to 
check the validity of entailments, using top-down decomposition rules.
Induction is performed by considering infinite (rational) proof trees modeling proofs by infinite descent, as in \cite{BrotherstonSimpson11}.
We also tackle the combination of SL reasoning with theory reasoning, relying on external decision procedures
 for checking the validity of formulas in the considered theory. This issue is of uttermost importance 
for applications, as reasoning on data structures without taking into account the properties of the data stored in these structures has a  limited scope, 
and the combination of SL with data constraints has been considered by several authors (see, e.g., \cite{DBLP:conf/cav/PiskacWZ13,DBLP:conf/pldi/Qiu0SM13,DBLP:conf/aplas/PerezR13,DBLP:conf/cade/XuCW17,DBLP:conf/vmcai/Le21}). 
Beside the fact that it is capable of handling theory reasoning, our procedure has several advantages over the model-based bottom-up
algorithms that were previously devised \cite{PMZ20,DBLP:conf/lpar/EchenimIP20}.
It is  goal-oriented: the rules apply backward and reduce the considered entailments to simpler ones, until 
 axioms are reached.
The advantage is that the proof procedure is driven by the form of the current goal.
The procedure is also better-suited for interactive theorem proving (e.g., the user can guide the application of the inference rules, while the procedures in \cite{PMZ20,DBLP:conf/lpar/EchenimIP20} work as ``black boxes'').
If the entailment is valid, then the produced proof tree serves as a certification of the result, which can be checked if needed by the user or another system, 
while the previous procedures \cite{PMZ20,DBLP:conf/lpar/EchenimIP20} produce no certification.
Finally, the correctness proof of the algorithm is also simpler and more modular.
More specifically, we establish several new results in this paper.
\begin{compactenum}
\item{First, we show that the proof procedure is sound (Theorem \ref{theo:sound}) and complete (Theorem \ref{theo:comp}), in the sense that an entailment is valid 
iff it admits a  proof tree. The proof tree may be infinite, hence the result does not entail 
that checking entailments is semi-decidable. However, this shows that the procedure can be used
as a semi-decision procedure for checking non-validity (i.e., an entailment is not valid iff the procedure is stuck eventually in some branch), provided the base theory is decidable.}
\item{If the theory only contains the equality predicate, then we show that 
the entailments can be reduced to entailments in the empty theory (Theorem \ref{theo:elimeq}). The intuition is that  all the equality and disequality constraints
can be encoded in the formulas describing the shape of the data structures.
 This result is also described in a paper \cite{EP22b} accepted for presentation at ASL 2022 (workshop with no formal proceedings).}
\item{By focusing on the  case where the theory is empty, we show (Theorem \ref{theo:comp_term}) that every valid 
entailment admits a  proof tree that is rational (i.e., has a finite number of pairwise distinct subtrees, up to a renaming of variables). Furthermore, the number of sequents occurring in the tree is at most $\bigO(2^{2^n})$, where $n$ is the size of the initial sequent. In combination with the previous result, this theorem allows us to reestablish the \twoexptime\ membership of the entailment problem for inductive systems satisfying the conditions above in the theory of equality \cite{PZ20}.}
\item{We also show that the proof tree is finite if the inductive rules that define the predicates
occurring on the left-hand side of the entailment terminate (Corollary \ref{cor:finite}).}

\end{compactenum}

\section{Preliminaries}

\newcommand{\tmodel}{\modelsr_{\cal T}}
\newcommand{\hpredicate}{spatial predicate\xspace}
\newcommand{\tpredicate}{$\theory$-predicate\xspace}


In this section, we define the syntax and semantics of the fragment of separation logic that is considered in the paper. Our definitions are mostly standard, see, e.g., \cite{DBLP:conf/csl/OHearnRY01,Reynolds02,IosifRogalewiczSimacek13} for more details and explanations on  separation logic as well as on the conditions on the inductively defined predicates that ensure decidability of the entailment problem.

\subsection{Syntax}

Let $\vars$ be a countably infinite set of {\em variables}.
Let $\tpreds$ be a 
set of {\em {\tpredicate}s} (or {\em theory predicates}, denoting relations in an underlying theory of locations) and 
let $\preds$ be a set of {\em {\hpredicate}s}, disjoint from $\tpreds$. Each symbol $p\in \tpreds \cup \preds$ is  associated with a unique arity $\ar{p}$.
We assume that $\tpreds$ contains in particular two binary symbols $\iseq$ and $\not \iseq$ and a nullary symbol $\myfalse$.

\begin{definition}
Let $\rank$ be some fixed\footnote{Note that $\rank$  is not considered as constant for 
the complexity analysis in Section \ref{sect:comp}: it is part of the input.}
natural number.
The set of {\em {\slformula}s} (or simply formulas) $\aform$ is inductively defined as follows:
\[\aform := \emp \; \| \; x \mapsto (y_1,\dots,y_\rank) \; \| \; \aform_1 \vee \aform_2 \; \| \;  \aform_1 * \aform_2 \| \; 
p(x_1,\dots,x_{\ar{p}}) \; \| \; \exists x. ~ \aform_1   \]
where  $\aform_1,\aform_2$ are {\slformula}s, $p\in \tpreds \cup \preds$ and $x,x_1,\dots,x_{\ar{p}}, y_1,\dots,y_{\rank}$ are variables.
\end{definition}

\newcommand{\mapformula}{$\mapsto$-formula\xspace}

Note that the considered fragment (as in many works in separation logic) does not include standard conjunction, negation or universal quantifications.
Indeed, the addition of these constructions, without any further restriction,
makes entailment checking undecidable (see for instance \cite{PMZ20}). The separating implication $\wand$  is  not supported either, although a similar but more restricted connective $\swand$
will be introduced below.
Formulas are taken modulo  associativity and commutativity of $\vee$ and $*$, modulo  commutativity of existential quantifications  and modulo the neutrality of $\emp$ 
for $*$.  A {\em spatial atom} is  a formula that is either of the form $x \mapsto (y_1,\dots,y_\rank)$ (called a {\em points-to atom}) or 
$p(x_1,\dots,x_{\ar{p}})$ with $p\in \preds$  (called a {\em predicate atom}). 
A {\em \tatom} is a formula of the form $p(x_1,\dots,x_{\ar{p}})$ with $p\in \tpreds$.
An {\em atom} is either a spatial atom or a \tatom.
A {\em \tformula} is either $\emp$ or a separating conjunction of {\tatom}s.
A formula of the form $\exists x_1.\dots.\exists x_n.~\aform$ (with $n \geq 0$) is denoted by 
$\exists \vec{x}.~\aform$ with $\vec{x} = (x_1,\dots,x_n)$.
A formula is {\em predicate-free} (resp.\ {\em disjunction-free}, resp. {\em quantifier-free}) if it contains no predicate symbol in $\preds$ (resp.\ no occurrence of $\vee$, resp.\ of $\exists$).
It is in {\em prenex form} if it is of the form $\exists \vec{x}. \aform$, where $\aform$ is quantifier-free and $\vec{x}$ is a possibly empty vector of variables.
A {\em symbolic heap} is a prenex disjunction-free formula, i.e., a formula of the form $\exists \vec{x}. \aform$, where $\aform$ is  a separating conjunction of atoms.

\begin{definition}
A {\em \mapformula} is a formula of the form $\exists \vec{x}. (u \mapsto \vec{v} * \atform)$, where $\atform$ is a \tformula. \label{def:mapform}
\end{definition}

\begin{example}
\label{ex:syntax} 
Let $\slst$ and $\alst$ be two spatial predicates denoting
increasing and acyclic nonempty list segments, respectively.
The symbolic heap: $x_1 \mapsto (x_2) * x_1 \geq 0   * \slst(x_2,x_3) * \slst(x_3,x_4)$ denotes an increasing list of positive numbers
composed by a first element $x_1$, linked to a list composed by the concatenation of two list segments, from $x_2$ to $x_3$ and from $x_3$ to $x_4$, respectively.
The atom $x_1 \mapsto (x_2)$ is a points-to atom, $\slst(x_2,x_3)$ and $\slst(x_3,x_4)$
are predicate atoms and $x_1 \geq 0$ is a \tformula (constructed using a monadic predicate stating that $x_1$ is positive).
The symbolic heap $\exists x_1,x_2. (\alst(x_1,x_2) * x_1 \geq 0 * x_2 \geq 0)$ denotes an acyclic list segment between two positive locations.
\end{example}

We denote by $\fv{\aform}$ the set of variables freely occurring in $\aform$ (i.e., occurring in $\aform$ but not within the scope of any existential quantifier).  
 A {\em substitution} $\sigma$ is a function mapping variables to variables. The {\em domain} $\dom{\sigma}$ of a substitution $\sigma$ is the set of variables $x$ 
such that $\sigma(x) \not = x$, 
and we let $\img{\sigma} = \sigma(\dom{\sigma})$. 
For all substitutions $\sigma$, we assume that $\dom{\sigma}$ is finite and that $\sigma$ is idempotent. For any expression (variable, tuple of variables or formula) $e$, we denote by $e\sigma$ the expression obtained from $e$ by replacing 
every free occurrence of a variable $x$ by $\sigma(x)$ and by $\replall{}{x_i}{y_i}{1 \leq i \leq n}$ (where the $x_1,\dots,x_n$ are pairwise distinct) the substitution 
such that $\sigma(x_i) = y_i$ and $\dom{\sigma} \subseteq \{ x_1,\dots,x_n \}$.
 For all sets $E$, 
 $\card{E}$ 
 is the cardinality of $E$. 
 For all sequences or words $w$, 
 $\len{w}$ denotes the length of $w$.
 We sometimes identify vectors with sets, if the order is unimportant, e.g., we write $\vec{x} \setminus \vec{y}$ to denote the vector formed by the components of $\vec{x}$ that do not occur in $\vec{y}$.
 
\subsection{Size and Width}
 
  We assume that the symbols in $\preds \cup \tpreds \cup \vars$ are words\footnote{Because we will consider transformations introducing an unbounded number of new predicate symbols, we cannot assume that the predicate atoms have a  constant size.} over a finite alphabet of a constant size, strictly greater than $1$. For any expression $e$, we denote by $\size{e}$ 
  the size of $e$, i.e., the number of occurrences of 
symbols\footnote{Each symbol $s$ in $\preds \cup \tpreds \cup \vars$ is counted with a weight equal to its length $\len{s}$, and all the logical symbols have weight $1$.} in $e$.
 We define the {\em width} of a formula as follows: 
{\small
\[
\begin{tabular}{llllll}
$\widt{\aform_1 \vee \aform_2}$ & $ = $ & $\max(\widt{\aform_1},\widt{\aform_2})$ & \qquad \\
   $\widt{\exists x. \aform}$ & $=$ & $\widt{\aform} + \size{\exists x}$ \\
$\widt{\aform_1 * \aform_2}$ & $ = $ & $\widt{\aform_1} + \widt{\aform_2} + 1$ \\
$\widt{\aform}$ & $ = $ & $\size{\aform}$ \quad if $\aform$ is an atom \\
\end{tabular}
\]
}
Note that $\widt{\aform}$ coincides with $\size{\aform}$ if $\aform$ is disjunction-free.

\subsection{Inductive Rules}

\newcommand{\dependson}[1]{\geq_{#1}}

The semantics of the predicates in $\preds$ is provided 
by user-defined inductive rules satisfying some conditions (as defined in \cite{IosifRogalewiczSimacek13}):
\begin{definition}
\label{def:sid}
A (progressing and connected) set of inductive rules (\pcSID) $\asid$
is a finite set of rules of the form
\[ p(x_1,\dots,x_n) \Leftarrow \exists \vec{u}. ~ x_1 \mapsto (y_1,\dots,y_\rank) 
* \aform,\]
	where $\fv{x_1 \mapsto (y_1,\dots,y_\rank) 
		* \aform} \subseteq \{x_1, \ldots, x_n\} \cup \vec{u}$,
		$\aform$ is a possibly empty separating conjunction of predicate atoms and {\tformula}s, and for every predicate atom $q(z_1,\dots,z_{\ar{q}})$ 
		occurring in $\aform$, we have $z_1 \in \{ y_1,\dots,y_\rank\}$.
We let $\size{p(\vec{x}) \Leftarrow \aform} = \size{p(\vec{x})} + \size{\aform}$, 
$\size{\asid} = \Sigma_{\rho\in \asid} \size{\rho}$
and $\widt{\asid} = \max_{\rho \in \asid} \size{\rho}$.
\end{definition}

In the following, $\asid$ always denotes a \pcSID. 
We emphasize that the right-hand side of  every inductive rule contains exactly one 
points-to atom, the left-hand side of which is the first argument $x_1$ of the predicate symbol (this condition is referred to as the {\em progress} condition), and
that this points-to atom contains the first argument of every predicate atom on the right-hand side of the rule (the {\em connectivity} condition).


\begin{example}
\label{ex:rules}
The predicates $\slst$ and $\alst$ of Example \ref{ex:syntax} are defined as follows:
\[
\begin{tabular}{lcr}
$\slst(x,y)$ & $\Leftarrow$ & $x \mapsto (y) * x \leq y$ \\
$\slst(x,y)$  & $\Leftarrow$ & $\exists x'. ~ x \mapsto (x') * \slst(x',y) * x \leq x'$ \\
$\alst(x,y)$ & $\Leftarrow$ & $x \mapsto (y) * x \not\iseq y$ \\
$\alst(x,y)$  & $\Leftarrow$ & $\exists x'. ~ x \mapsto (x') * \alst(x',y) * x \not \iseq y$
\end{tabular}
\]
This set  is progressing and connected.
In contrast, the rule
$\slst(x,y) \Leftarrow x \iseq y$ 
is not progressing, because it contains no points-to atom. 
 A possibly empty list must thus be denoted by a disjunction in our framework:
 $\slst(x,y) \vee (x \iseq y)$.
\end{example}


\begin{definition}
\label{def:unfold}
We write $p(x_1,\dots,x_{\ar{p}}) \unfoldto{\asid} \aform$
if $\asid$ contains a rule (up to $\alpha$-renaming) 
$p(y_1,\dots,y_{\ar{p}}) \Leftarrow \aformB$, where $x_1,\dots,x_{\ar{p}}$ are not bound in $\aformB$, and
$\aform = \replall{\aformB}{y_i}{x_i}{i \in \interv{1}{{\ar{p}}}}$.

The relation $\unfoldto{\asid}$ is extended to all formulas as follows:
$\aform \unfoldto{\asid} \aform'$ if one of the following conditions holds: 
\begin{enumerate}[(i)]
	\item{$\aform = \aform_1 \bullet \aform_2$ (modulo AC, with $\bullet \in \{ *, \vee \}$), $\aform_1 \unfoldto{\asid} \aform_1'$, no free or existential variable  in $\aform_2$ is bound in $\aform_1'$
	and 
	$\aform' = \aform_1' \bullet \aform_2$;}
	\item{$\aform = \exists x.~ \aformB$, $\aformB \unfoldto{\asid} \aformB'$, $x$ is not bound in $\aformB'$
	and 
	$\aform' = \exists x.~ \aformB'$.}
\end{enumerate}
We denote by $\unfoldto{\asid}^+$ the transitive closure of $\unfoldto{\asid}$, and by $\unfoldto{\asid}^*$ its reflexive and transitive closure.
A formula $\aformB$ such that $\aform \unfoldto{\asid}^* \aformB$ is called an {\em $\asid$-unfolding} of 
$\aform$.
We denote by $\dependson{\asid}$ the least transitive and reflexive binary relation on $\preds$
such that $p \dependson{\asid} q$ holds if 
$\asid$ contains a rule of the form $p(y_1,\dots,y_{\ar{p}}) \Leftarrow \aformB$, where $q$ occurs in $\aformB$. If $\aform$ is a formula, we write $\aform \dependson{\asid} q$ if $p \dependson{\asid} q$ for some $p \in \preds$ occurring in $\aform$.
\end{definition}



\begin{example}
	With the rules of Example \ref{ex:rules}, we have:
	\[\begin{array}{rcl}
		\alst(x_1,x_2) * \alst(x_2,x_3) &\unfoldto{\asid}&
		x_1 \mapsto (x_2) * x_1 \not \iseq x_2 * \alst(x_2,x_3) \\
		& \unfoldto{\asid} &\exists x'. ~(
		x_1 \mapsto (x_2) * x_1 \not \iseq x_2 *  x_2 \mapsto (x') *\\
		& & \quad  \alst(x',x_3) * x_2 \not \iseq x_3)\\
		& \unfoldto{\asid} & \exists x', x''. ~(
		x_1 \mapsto (x_2) * x_1 \not \iseq x_2 *  x_2 \mapsto (x') * \\
		& & \quad x' \mapsto (x'') * \alst(x'',x_3) * x' \not \iseq x_3 * x_2 \not \iseq x_3)\\
		& \unfoldto{\asid}& \exists x', x''. ~(
		x_1 \mapsto (x_2) * x_1 \not \iseq x_2 *  x_2 \mapsto (x') *\\
		& & \quad  x' \mapsto (x'') * x'' \mapsto (x_3) * x'' \not \iseq x_3 *\\
		& & \quad  x' \not \iseq x_3 * x_2 \not \iseq x_3)
	\end{array}\]
\end{example}

Note that the number of formulas $\aform'$ such that $\aform \unfoldto{\asid} \aform'$ is finite, up to $\alpha$-renaming. 
Also, if $\aform \unfoldto{\asid}^* \aform'$ then $\fv{\aform'} \subseteq \fv{\aform}$. 

\subsection{Semantics}
\newcommand{\Nh}[1]{N_h(#1)}

\begin{definition}
Let $\Loc$ be a countably infinite set of so-called 
{\em locations}.
An {\em SL-structure} is a pair $(\astore,\aheap)$ where 
$\astore$ is a {\em store}, i.e.\ a total function from $\vars$ to $\Loc$, and 
$\aheap$ is a {\em heap}, i.e.\ a partial finite function from $\Loc$ to $\Loc^\rank$ which is written as a relation: $\aheap(\ell) = (\ell_1,\dots,\ell_\rank)$ iff $(\ell,\ell_1,\dots,\ell_\rank) \in \aheap$.
The {\em size} of a structure $(\astore,\aheap)$ is the cardinality  of $\dom{\aheap}$.
\end{definition}

\begin{definition}	
For every heap $\aheap$, let $\locs{\aheap} = \{ \ell_i \mid (\ell_0,\dots,\ell_\rank) \in \aheap, i = 0,\dots,\rank \}$.
A location $\ell$ (resp.\ a variable $x$) is {\em allocated} in a heap $\aheap$ (resp.\ in a structure ($\astore,\aheap$)) if $\ell \in \dom{\aheap}$ (resp.\ $\astore(x)\in \dom{\aheap}$).
Two heaps $\aheap_1,\aheap_2$ are {\em disjoint} if 
$\dom{\aheap_1} \cap \dom{\aheap_2} = \emptyset$,
in this case $\aheap_1 \dunion \aheap_2$ denotes the 
union of $\aheap_1$ and $\aheap_2$.
\end{definition}


 Let $\modelst$ be a satisfiability relation between stores and {\tformula}s, satisfying the following properties:
  $\astore \modelst x \iseq y$ (resp.\ $\astore \modelst x \not \iseq y$) iff $\astore(x) = \astore(y)$ (resp.\ $\astore(x) \not = \astore(y)$), $\astore \not \modelst \myfalse$ and  
 $\astore \modelst \atform *\atformB$ iff $\astore\modelst \atform$ and $\astore \modelst \atformB$. 
  For all {\tformula}s $\atform, \atformB$, we write 
$\atform \modelst \atformB$ if, for every store $\astore$ such that  
$\astore \modelst \atform$, we have $\astore \modelst \atformB$. We write
$\atform \modelsi \atformB$ if the implication holds for every injective store $\astore$,
 i.e., if $\atform *  \atform_{inj} \modelst \aformB$, 
where $\atform_{inj}$ denotes the separating conjunction 
of all disequations $x \not \iseq x'$, with $x,x' \in \fv{\atform} \cup \fv{\atformB}$, and $x \not = x'$.
We say that $\theory$ is {\em closed under negation} if for every \tformula $\atform$, one can compute 
 a \tformula $\atform'$ (also written $\neg \atform$) such that for every store $\astore$,  $\astore \modelst \atform'$ iff $\astore \not \modelst \atform$.
To make the rules in Section \ref{sect:rules} applicable in practice it is 
necessary to have a procedure to check whether $\atform \modelst \atformB$.
In all examples, $\Loc$ is the set of integers, and the {\tformula}s are arithmetic formulas, interpreted as usual (with predicates $\leq$ and $\iseq$).
  \begin{definition}
 \label{def:semantics}
Given formula $\aform$, a \pcSID $\asid$ and a structure $(\astore,\aheap)$,
we write $(\astore,\aheap) \modelsr \aform$ and say that $(\astore,\aheap)$ is an {\em $\asid$-model} (or simply a model if $\asid$ is clear from the context) of $\aform$ if one of the following conditions holds.
\begin{compactitem}
\item{$\aform = x \mapsto (y_1,\dots,y_\rank)$ and
$\aheap = \{ (\astore(x),\astore(y_1),\dots,\astore(y_\rank)) \}$.}
\item{$\aform$ is a \tformula, $\aheap = \emptyset$ and $\astore \modelst \aform$.}
\item{$\aform = \aform_1 \vee \aform_2$ and 
$(\astore,\aheap) \modelsr \aform_i$, for some $i  = 1,2$.}
\item{$\aform = \aform_1 * \aform_2$ and there exist disjoint heaps
$\aheap_1,\aheap_2$ such that 
$\aheap = \aheap_1 \dunion \aheap_2$ and 
$(\astore,\aheap_i) \modelsr \aform_i$, for all $i  = 1,2$.}

\item{$\aform = \exists x. ~ \aform$ and 
$(\astore',\aheap) \modelsr \aform$, for some store $\astore'$ coinciding with 
$\astore$ on all variables distinct from $x$.}
\item{
$\aform = p(x_1,\dots,x_{\ar{p}})$, $p \in \preds$ and  $(\astore,\aheap) \modelsr \aformB$ for some $\aformB$ such that
$\aform \unfoldto{\asid} \aformB$.}
\end{compactitem}
If $\aseq$ is a sequence of formulas, then we write $(\astore,\aheap) \modelsr \aseq$ if $(\astore,\aheap)$ satisfies at least one formula in $\aseq$.
\end{definition}
We emphasize that a \tformula is satisfied 
only 
in structures with empty heaps. This convention is used to simplify notations, because it avoids 
having to consider both standard and separating conjunctions.
Note that Definition \ref{def:semantics} is well-founded because of the progress condition: the size of $\aheap$ decreases at each recursive call of a predicate atom.
We write $\aform \modelsr \aformB$ if every $\asid$-model of $\aform$ is an $\asid$-model of $\aformB$
and  $\aform \equivr \aformB$ if $\aform \modelsr \aformB$ and $\aformB \modelsr \aform$.
Every formula can be transformed into prenex form using the well-known equivalences: $(\exists x. \aform) \bullet \aformB \equiv \exists x. (\aform \bullet \aformB)$, for all $\bullet \in \{ \vee, * \}$, where $x\not \in\fv{\aformB}$.


\begin{example}
Let $\asid$ be the set of rules in Example \ref{ex:rules}.
Assume that $\Loc$ is the set of integers ${\Bbb Z}$.
The formula $\aform = \slst(x_1,x_2) * \slst(x_2,x_3)$
admits the following $\asid$-model
$(\astore,\aheap)$, where $\astore(x_1) = 1$, 
$\astore(x_2) = 2$,
$\astore(x_3) = 4$,
and $\aheap = \{ (1,2), (2,3), (3,4) \}$.
Indeed, we have: 
{\small
\[
\begin{tabular}{lll}
$\aform$ &  $\unfoldto{\asid}$ &
$x_1 \mapsto (x_2) * x_1 \leq x_2 * \slst(x_2,x_3)$ \\ 
& $\unfoldto{\asid}$ & 
$\exists x'.~ (x_1 \mapsto (x_2) * x_1 \leq x_2 * x_2 \mapsto (x') * x_2 \leq x' * \slst(x',x_3))$ \\ 
& $\unfoldto{\asid}$ & 
$\exists x'.~ (x_1 \mapsto (x_2) * x_1 \leq x_2 * x_2 \mapsto (x') * x_2 \leq x' * x' \mapsto x_3 * x' \leq x_3))$
\end{tabular}
\]}\noindent
and 
$(\astore',\aheap) \models (x_1 \mapsto (x_2) * x_1 \leq x_2 * x_2 \mapsto (x') * x_2 \leq x' * x' \mapsto x_3 * x' \leq x_3))$, with $\astore'(x') = 3$ and
$\astore'(x) = \astore(x)$ if $x\not = x'$.
The formula $\slst(x_1,x_2) * \slst(x_1,x_3)$ admits no 
$\asid$-model. Indeed, it is clear that all the structures that satisfy $\slst(x_1,x_2)$ or $\slst(x_1,x_3)$
must allocate $x_1$, and the same location cannot be allocated in disjoint parts of the heap.
\end{example}

  Note that the progress condition entails the following property:
\begin{proposition}\label{prop:prog-dom}
	If $(\astore,\aheap) \modelsr p(x_1, \ldots,  x_{\ar{p}})$, then $\astore(x_1) \in \dom{\aheap}$.
\end{proposition}
\begin{proof}
By Definition \ref{def:sid}, if $p(x_1,\dots,x_{\ar{p}}) \unfoldto{\asid} \aform$, then $\aform$ contains a points-to atom with left-hand side $x_1$. 
\end{proof}

	\begin{proposition}
		\label{prop:hform}
		Let $\aform$ be a disjunction-free formula containing at least one spatial atom. 
		If $(\astore,\aheap) \modelsr \aform$ then $\aheap$ is nonempty.
	\end{proposition}
	\begin{proof}
		The proof is by induction on the set of formulas.
		\begin{compactitem}
			\item{If $\aform$ is a points-to atom then it is clear that  $\card{\dom{\aheap}} = 1$.}
			\item{If $\aform$ is a predicate atom, then there exists 
				$\aformB$ such that $\aform \unfoldto{\asid} \aformB$ and $(\astore,\aheap) \modelsr \aformB$. 
				By the progress condition, $\aformB$ is of the form $\exists \vec{w}. ~(u \mapsto (v_1,\dots,v_\rank) * \aformB')$, hence there exists a subheap $\aheap'$ of $\aheap$ and a store $\astore'$ such that $(\astore',\aheap') \modelsr u \mapsto (v_1,\dots,v_\rank)$. 
				This entails that $\card{\dom{\aheap'}} = 1$ thus $\card{\dom{\aheap}} \geq 1$.}
			\item{If $\aform = \aform_1 * \aform_2$, then necessarily there exists $i = 1,2$ such that  $\aform_i$ contains at least one spatial atom. Furthermore, there exist disjoint heaps $\aheap_1$, $\aheap_2$ such that
				$(\astore,\aheap_i) \modelsr \aform_i$ and $\aheap = \aheap_1 \dunion \aheap_2$.
				By the induction hypothesis, $\aheap_i$ is non empty, hence $\aheap$ is also non empty.}
			
			\item{If $\aform = \exists x. ~\aformB$ then $\aformB$ contains at least one spatial atom, and  there exists a store $\astore'$ such that $(\astore',\aheap) \modelsr \aformB$. By the induction hypothesis, we deduce that $\aheap$ is non empty.}

		\end{compactitem}
	\end{proof}

\subsection{Establishment}

The notion of establishment \cite{IosifRogalewiczSimacek13} is defined as follows:
\begin{definition}
A \pcSID is {\em established}
if
for every atom $\anatom$,  every predicate-free  formula 
$\exists \vec{x}. \aform$ such that
$\anatom \unfoldto{\asid}^* \exists \vec{x}. \aform$, and  every $x\in \vec{x}$, $\aform$ is of the form $x' \mapsto (y_1,\dots,y_\rank) * \atform * \aformB$, where $\atform$ is a separating conjunction of equations (possibly $\emp$) such that
 $\atform \modelst x \iseq x'$. 
\end{definition}
All the rules in Example \ref{ex:rules} are trivially established. For instance, 
$\alst(x,y) \Leftarrow \exists x'. ~ x \mapsto (x') * \alst(x',y) * x \not \iseq y$
fulfills the condition, since every predicate-free unfolding of  $\alst(x',y)$ contains a formula of the form $x' \mapsto \dots$ and $\emp \modelst x' \iseq x'$.

The following lemma states a key property of an established \pcSID: 
every 
location referred to in the heap is allocated, except possibly those associated with a free variable.

\begin{lemma}
\label{lem:fr}
Let $\asid$ be an established \pcSID and
let $\aform$ be a quantifier-free symbolic heap.
If $(\astore,\aheap) \modelsr \aform$ then
$\locs{\aheap} \setminus \dom{\aheap} \subseteq \astore(\fv{\aform})$.
\end{lemma}
\begin{proof}
By definition, $\aform \unfoldto{\asid}^* \exists \vec{y}. ~ \aformB$, where  
$\aformB$ is a quantifier-free and predicate-free formula, and
$(\astore',\aheap) \modelsr \aformB$ for some store $\astore'$ coinciding with $\astore$
on all the variables not occurring in $\vec{y}$.
Let $\ell \in \locs{\aheap} \setminus \dom{\aheap}$.
Since $(\astore',\aheap) \modelsr \aformB$, necessarily $\aformB$ contains an atom of the form 
$y_0 \mapsto (y_1,\dots,y_\rank)$ with $\astore'(y_i) = \ell$, for some $i = 1,\dots,n$.
If $y_i \in \fv{\aformB} \subseteq \fv{\aform}$, then $\astore(y_i) = \astore'(y_i)$ and the proof is completed.
Otherwise, $y_i \in \vec{y}$, and by the establishment condition 
$\aformB$ contains an atom of the form $y_0' \mapsto (y_1',\dots,y_\rank')$ 
and a \tformula $\atform$ with $\atform \modelst y_i \iseq y_0'$.
Since $(\astore',\aheap) \modelsr \aformB$, we have $\astore' \modelst \atform$, thus $\astore'(y_i) = \astore'(y_0')$ and
$\ell \in \dom{\aheap}$, which contradicts our hypothesis.
\end{proof}

In the remainder of the paper, we assume that every considered \pcSID is established.

\section{Extending the Syntax}

\newcommand{\formalWatom}{$\tswand$-predicate\xspace}



We extend the syntax of formulas by considering 
constructs of the form $\MWl{\atail}{p(\vec{x})}{\vec{u}}{\vec{v}}$, called {\em {\Watom}s}  (standing for {\em partially unfolded atoms}), where $\atail$ is a possibly empty
separating conjunction of predicate atoms and $p \in \preds$. 
 The intuition is that a \Watom is valid in a structure 
 if
there exists a {\em partial} unfolding of $p(\vec{x})$ that is true in the considered structure, and the formula $\atail$ denotes the part that is not unfolded.

\begin{definition}
	A {\em \formalWatom} is an expression of the form 
	$\MWa{\atail}{p(\vec{x})}{\vec{u}}$, 
	where $\atail$ is a possibly empty
	separating conjunction of predicate atoms, $p \in \preds$ and $\vec{u}$ is a vector of pairwise distinct variables containing all the variables in $\fv{\atail} \cup \vec{x}$.
	A {\em \Watom} is an expression of the form 
	$\MWf{\anatom}{\vec{v}}$, where $\anatom$ is a \formalWatom $\MWa{\atail}{p(\vec{x})}{\vec{u}}$ and $\len{\vec{u}} = \len{\vec{v}}$.
	
		For every \Watom $\MWf{\anatom}{\vec{v}}$, we define
	$\MWf{\anatom}{\vec{v}}\sigma = \MWf{\anatom}{\vec{v}\sigma}$, and we let $\size{\MWl{\atail}{p(\vec{x})}{\vec{u}}{\vec{v}}} = \size{\atail} + \size{p(\vec{x})} + \size{\vec{u}} + \size{\vec{v}} + 1$. 
\end{definition}

Note that for a  {\Watom} $\MWf{\anatom}{\vec{v}}$, $\vec{v}$ is necessarily of the form $\vec{u}\theta$, for some substitution $\theta$ with domain $\vec{u}$, since  the variables in $\vec{u}$ are pairwise distinct.
Note also that when applying a substitution to $\MWf{\anatom}{\vec{v}}$ the variables occurring in the \formalWatom $\anatom$ are not instantiated: they may be viewed as bound variables.
Formally, the semantics of these constructs is defined as follows.

\begin{definition}
\label{def:semMW}
	For every \pcSID $\asid$ and for every SL-structure $(\astore,\aheap)$, 
	$(\astore,\aheap) \modelsr \MWs{\atail}{p(\vec{x})}{\vec{u}}{\theta}$ if
	there exists a formula of the form $\exists \vec{y}. (\atail' * \aform)$, a substitution  $\sigma$ with   $\dom{\sigma} \subseteq \vec{y} \cap \fv{\atail'}$, 
	and a store $\astore'$ coinciding with $\astore$ on all variables not occurring in $\vec{y}$ such that:
	\begin{itemize}
		\item $p(\vec{x}) \unfoldto{\asid}^+ \exists \vec{y}. (\atail' * \aform)$, (up to AC and transformation into prefex form),
		\item $\atail = \atail'\sigma$,
		\item $(\astore',\aheap) \modelsr \aform\sigma\theta$. 
	\end{itemize}
\end{definition}

%
%

\newcommand{\myset}[1]{\left\{#1\right\}}
\newcommand{\isdef}{\stackrel{\mbox{\tiny def}}{=}}

\begin{example}
	Consider the \pcSID
 \[\asid= 
	\{ p(x) \Leftarrow \allowbreak  \exists z_1 z_2.~x\mapsto (z_1, z_2) * q(z_1) * q(z_2),\ \allowbreak\ q(x) 
	\allowbreak \Leftarrow\allowbreak x\mapsto \allowbreak (x,x)\},\]
	the heap $\aheap = \myset{(\ell_1, \ell_2, \ell_3),\, (\ell_3, \ell_3, \ell_3)}$ and the store $\astore$ such that $\astore(x) = \ell_1$ and $\astore(y) = \ell_2$. We have $(\astore, \aheap) \modelsr \MWl{q(y')}{p(x')}{x',y'}{x,y}$. Indeed, it is straightforward to verify that
	\(p(x') \unfoldto{\asid}^+ \exists z_1 z_2.~ (x' \mapsto (z_1, z_2) * q(z_1) * z_2 \mapsto (z_2, z_2)).\)
	Hence, by letting 
	$\sigma = \myset{z_1 \leftarrow y'}$ 
	and considering the store $\astore'$ such that $\astore'(z_2) = \ell_3$ and which coincides with $\astore$ otherwise, we have $(\astore', \aheap) \models x \mapsto (y, z_2) * z_2 \mapsto (z_2, z_2) = (x' \mapsto (y', z_2) * z_2 \mapsto (z_2, z_2))\theta$, with $\theta = \{ x' \leftarrow x, y' \leftarrow y \}$. 
\end{example}

\begin{remark}
\label{rem:useless_var}
It is clear that the semantics of an atom 
$\MWs{\atail}{p(\vec{x})}{\vec{u}}{\theta}$ depends only on the variables 
$x\theta$ such that $x$ occurs in $\atail$ or $\vec{x}$ (and the order in which those variables occur in $\vec{u}$ does not matter).
Thus in the following we implicitly assume that the irrelevant variables are dismissed. In particular, in the termination proof of Section \ref{sect:term}, we assume that 
$\size{\MWl{\atail}{p(\vec{x})}{\vec{u}}{\vec{v}}} = \bigO(\size{\atail} + \size{p(\vec{x})})$.
Also, two {\Watom}s $\MWf{\anatom}{\vec{u}}$
and $\MWf{\anatom'}{\vec{u}}$ are equivalent if $\anatom$ and $\anatom'$ are identical up to a renaming.
\end{remark}

A predicate atom $p(\vec{x}\theta)$ where $\vec{x}$ is a vector of pairwise disjoint variables is equivalent to the \Watom $\MWl{\emp}{p(\vec{x})}{\vec{x}}{\vec{x}\theta}$,
thus, in the following, we  sometimes  assume that predicate atoms are written as {\Watom}s.
\begin{definition}
	The relation $\unfoldto{\asid}^*$ is extended to formulas containing {\Watom}s as follows:
	$\MWs{\atail}{p(\vec{x})}{\vec{u}}{\theta} \unfoldto{\asid}^* \exists \vec{x}. \aformB$ iff\footnote{We assume that $\vec{x}'$ contains no variable in $\vec{u}\theta$.} 
	$p(\vec{x}) \unfoldto{\asid}^* \exists \vec{x}'. (\atail' * \aformB')$
	and there exists a substitution $\sigma$ such that $\dom{\sigma} \subseteq \vec{x}'$, 
	$\aformB = \aformB'\sigma\theta$, $\atail = \atail'\sigma$ and 
	$\vec{x} = \vec{x}' \setminus \dom{\sigma}$. 
\end{definition}

\begin{example}
Consider the rules of Example \ref{ex:rules}. Then we have:
$\MWl{\alst(z',y')}{\alst(x',y')}{x',y',z'}{x,y,z} \unfoldto{\asid}
x \mapsto (z) * x \not \iseq y$. Indeed, $\alst(x',y') \unfoldto{\asid} \exists u.~ (x' \mapsto (u) * \alst(u,y')  * x' \not \iseq y')$, hence it suffices to apply the above definition with the substitutions $\sigma = \{ u \leftarrow z' \}$ 
and $\theta = \{ x' \leftarrow x,\, y' \leftarrow y,\, z' \leftarrow z \}$.
\end{example}

\begin{remark}
  The semantics of $\swand$ is similar but slightly different from that  
  of the {\em context predicates} introduced in 
  \cite{EIP21a}.
 The difference is that the semantics uses the syntactic identity  $\atail = \atail'\sigma$ instead
of a semantic equality. 
For instance, with the rules of Example \ref{ex:rules}, the formula $\MWl{\alst(x',z')}{\alst(x',y')}{x',y'}{x,y,z}$ 
is unsatisfiable if $y' \not = z'$, because no atom $\alst(x',z')$  can occur in an unfolding of $\alst(x,',y')$. In contrast, a context predicate $\MW{\alst(x,z)}{\alst(x,y)}$ (as defined in \cite{EIP21a}) possibly holds in some structures $(\astore,\aheap)$ with $\astore(y) = \astore(z)$.
The use of {\formalWatom}s allows one to get rid of all equality constraints, by instantiating the variables occurring in the former. This is essential for forthcoming lemmas (see, e.g., Lemma \ref{lem:outsideheap}).
{\Watom}s are also related to the notion of $\Phi$-trees in \cite{PMZ20}
and to the {\em strong magic wand} introduced in
  \cite{NTKY2018}.
\end{remark}


\newcommand{\swfree}{\tswand-free\xspace}

In the following, unless specified otherwise, all the considered formulas are defined on the extended syntax.
A formula containing no occurrence of the symbol $\swand$  will be called a  {\em \swfree} formula.



\begin{proposition}
	\label{prop:roots}
	If $(\astore,\aheap) \modelsr \MWl{\atail}{p(x_1,\dots,x_{\ar{p}})}{\vec{u}}{\vec{u}\theta}$ then $\astore(x_1\theta) \in \dom{\aheap}$.
\end{proposition}
\begin{proof}
	By definition, there exists a formula $\exists \vec{y}. (\atail' * \aform)$ and a substitution $\sigma$ such that $p(x_1,\dots,x_{\ar{p}}) \unfoldto{\asid}^+ \exists \vec{y}. (\atail' * \aform)$, $\atail = \atail'\sigma$, 
	$\dom{\sigma} \subseteq \vec{y}$ 
	and $(\astore',\aheap) \modelsr \aform\sigma\theta$, where $\astore'$ coincides with $\astore$ on every variable not occurring in $\vec{y}$. 
	Note that $\vec{y} \cap \{ x_1,\dots,x_{\ar{p}} \} = \emptyset$ by Definition \ref{def:unfold}; 
	up to $\alpha$-renaming, we may assume that $\vec{y} \cap \vec{u}\theta = \emptyset$.
	By the progress condition, $\atail' * \aform$ contains an atom of the form $x_1 \mapsto (z_1,\dots,z_\rank)$.
	This atom cannot occur in $\atail'$, because $\atail = \atail'\sigma$ and $\atail$ only contains predicate atoms by definition of {\Watom}s.
	Thus $x_1 \mapsto (z_1,\dots,z_\rank)$ occurs in $\aform$, and  $(\astore'(x_1\sigma\theta),\astore'(z_1\sigma\theta),\dots,\astore'(z_\rank\sigma\theta)) \in \aheap$. 
Since $x_1 \not \in \vec{y}$, we have $x_1\sigma\theta = x_1\theta$, and since $\vec{y} \cap \vec{u}\theta$, necessarily  $x_1\theta \not \in \vec{y}$ and $\astore'(x_1\theta) = \astore(x_1\theta)$. Hence $\astore(x_1\theta) \in \dom{\aheap}$.
\end{proof}

  \begin{proposition}
  \label{prop:inst_deriv}
  If $\aform \unfoldto{\asid}^* \aformB$ then
  $\aform\sigma \unfoldto{\asid}^* \aformB\sigma$, for every susbtitution $\sigma$. If $\aform\sigma \unfoldto{\asid}^* \aformB'$ then $\aform \unfoldto{\asid}\aformB$ for a formula $\aformB$ such that $\aformB\sigma = \aformB'$. 
  \end{proposition}
 \begin{proof}
	The proof is by induction on the derivation.
	We only handle the case where $\aform$ is a predicate atom $p(x_1,\dots,x_n)$ and
	$\aform \unfoldto{\asid} \aformB$ (resp. $\aform\sigma \unfoldto{\asid} \aformB'$), 
	the general case follows by an immediate induction.
	Assume that $\aform \unfoldto{\asid} \aformB$. Then by definition, $\aformB = \replall{\aformC}{y_i}{x_i}{i =1,\dots,n}$,
	for some rule $p(y_1,\dots,y_n) \Leftarrow \aformC$ in $\asid$.
	We deduce that $p(x_1\sigma,\dots,x_n\sigma) \unfoldto{\asid} \replall{\aformC}{y_i}{x_i\sigma}{i =1,\dots,n}$, thus
	$\aform\sigma \unfoldto{\asid} \aformB\sigma$. Now assume that $\aform\sigma \unfoldto{\asid} \aformB'$. Then $\aformB' = \replall{\aformC}{y_i}{x_i\sigma}{i =1,\dots,n}$,
	for some rule $p(y_1,\dots,y_n) \Leftarrow \aformC$ in $\asid$, and w.l.o.g., we may assume up to $\alpha$-renaming that no variable in $\dom{\sigma}$ occurs in the rule, so that in particular, $\aformC\sigma= \aformC$. Let $\aformB = \replall{\aformC}{y_i}{x_i}{i =1,\dots,n}$, then it is clear that $\aform \unfoldto{\asid} \aformB$ and we have
	\[\aformB\sigma\ =\ (\replall{\aformC}{y_i}{x_i}{i =1,\dots,n})\sigma\ =\ \replall{\aformC}{y_i}{x_i\sigma}{i =1,\dots,n}\ =\ \aformB'.\]             
\end{proof}

%
%
%
%
%


\begin{proposition}
\label{prop:subst_sw}
If $(\astore,\aheap) \models \aform$ and $\astore(x) = \astore(x\theta)$ for every variable $x\in \fv{\aform}$, 
then
$(\astore,\aheap) \models \aform\theta$.
\end{proposition}
\begin{proof}
	The proof is by induction on the satisfiability relation. We only detail the proof when $\aform$ is 
	of the form $\MWs{\atail}{p(x_1,\dots,x_n)}{\vec{y}}{\theta'}$, the other cases are straightforward. 
	Since $(\astore,\aheap) \models \aform$, there exist a formula $\exists \vec{u}. (\atail' * \aformB)$,  a store $\astore'$ coinciding with $\astore$ on all variables not in $\vec{u}$ and a substitution $\sigma$ such that 
	$p(x_1,\dots,x_n) \unfoldto{\asid}^+ \exists \vec{u}. (\atail' * \aformB)$, $\dom{\sigma} \subseteq \vec{u} \cap \fv{\atail'}$, $(\astore',\aheap) \models \aformB\sigma\theta'$ and 
	$\atail'\sigma = \atail$.
	We assume (by $\alpha$-renaming) that $\vec{u}$ contains no variable in $\dom{\theta} \cup  \fv{\aform}$. This entails that $\astore'(x) = \astore'(x\theta)$ holds for all variables $x \in \fv{\aformB\sigma\theta'}$. Indeed, if $x\in \vec{u}$ then $x\theta = x$, and otherwise, $x \in \fv{\aform}$, so that
	$\astore'(x) = \astore(x) = \astore(x\theta) = \astore'(x\theta)$. 
	Then by the induction hypothesis we get
	 $(\astore',\aheap) \models \aformB\sigma\theta'\theta$, so that 
	 $(\astore,\aheap) \models \MWs{\atail}{p(x_1,\dots,x_n)}{\vec{y}}{\theta'\theta} = \aform\theta$.
\end{proof}

\section{Sequents}

Our proof procedure handles sequents which are defined as follows.
\newcommand{\empseq}{\Box}

\begin{definition}
\label{def:sequent}
A {\em sequent} is an expression of the form $\aform_0 \vdashr \aform_1,\dots,\aform_n$, where $\asid$ 
is a \pcSID,
$\aform_0$ is a \swfree  formula and $\aform_1,\dots,\aform_n$ are formulas. 
When $n=0$, the right-hand side of a sequent is represented by $\empseq$. 
A sequent is {\em disjunction-free} (resp.\ {\em \swfree}) if $\aform_0,\dots,\aform_n$ are disjunction-free (resp.\ \swfree), and {\em established} if $\asid$ is established.
We define: 
{\small
\begin{eqnarray*}
	\size{\aform_0 \vdashr \aform_1,\dots,\aform_n} &=& \Sigma_{i=0}^n \size{\aform_i} + \size{\asid}, \qquad \fv{\aform_1,\dots,\aform_n} = \bigcup_{i=0}^n \fv{\aform_i},\\
	\widt{\aform_0 \vdashr \aform_1,\dots,\aform_n} &=& \max \{  \widt{\aform_i}, \widt{\asid}, \card{ \bigcup_{i=0}^n \fv{\aform_i}} \mid 0 \leq i \leq n \}.
\end{eqnarray*}}
\end{definition}
Note that {\Watom}s occur only on the right-hand side of a sequent.
 Initially, all the considered sequents will be \swfree, but {\Watom}s will be introduced on the right-hand side by the inference rules defined in Section \ref{sect:rules}.

\newcommand{\countermodel}{countermodel\xspace}


\begin{definition}
\label{def:counter_model}
 A structure $(\astore,\aheap)$ is a {\em \countermodel} of 
a sequent  $\aform \vdashr \aseq$ iff $\astore$ is injective,  $(\astore,\aheap) \modelsr \aform$
 and $(\astore,\aheap) \not \modelsr \aseq$.
A sequent is {\em valid} if it has no \countermodel.
Two sequents are {\em equivalent} if they are both valid or both non-valid\footnote{Hence two non-valid sequents with different {\countermodel}s are equivalent.}.
\end{definition}


\begin{example}
For instance,
 $\slst(x_1,x_2) * x_2 \mapsto (x_3) * x_2 < x_3 \vdashr \alst(x_1,x_3)$ is a sequent, where $\asid$ is the set of rules from Example \ref{ex:rules} and $<$ is a \tpredicate interpreted as the usual strict order on integers. It is easy to check that it is  valid, since by definition of $\slst$, all the locations allocated by 
 $\slst(x_1,x_2)$ must be less or equal to $x_2$, hence cannot be equal to $x_3$. 
On the other hand, the sequent $\slst(x_1,x_2)  \vdashr \alst(x_1,x_2)$ is not valid: it admits the \countermodel $(\astore,\aheap)$, with $\astore(x_1) = 0, \astore(x_2) = 1$
and $\aheap = \{ (0,1), (1,1) \}$.
\end{example}

\begin{remark}	
The restriction to injective {\countermodel}s is for technical convenience only and does not entail any loss of generality, since it is possible to enumerate all the
equivalence relations on the free variables occurring in the sequent and test entailments separately
for each of these relations, by replacing all the variables in the
same class by the same representative. The number of such relations is simply exponential w.r.t.\ the number of free variables, thus the reduction does not affect the overall $2$-\exptime membership result derived from Theorem \ref{theo:comp_term}.
Note that the condition ``$\astore$ is injective'' in Definition \ref{def:counter_model} could be safely replaced by the slightly weaker condition ``$\astore$ is injective on $\fv{\aform} \cup \fv{\aseq}$''. Indeed, since $\Loc$ is infinite, we can always find an injective store $\astore'$ coinciding with $\astore$ on all variables in $\fv{\aform} \cup \fv{\aseq}$.
\end{remark}


Every formula $\aform$ can be reduced to an equivalent disjunction of symbolic heaps, using the well-known equivalences (if $x \not \in \fv{\aformB}$):
\[
\begin{tabular}{lllllll}
$\aform * (\aformB_1 \vee \aformB_2)$ & $\equiv$ & $(\aform * \aformB_1) \vee (\aform * \aformB_2)$ & \qquad &
$(\exists x. \aform) * \aformB$ & $\equiv$ & $\exists x. (\aform * \aformB)$ \\
$\exists x. (\aform_1 \vee \aform_2)$ & $\equiv$ & $(\exists x. \aform_1) \vee (\exists x. \aform_2)$ \\
\end{tabular} 
\]
Consequently any sequent $\aform \vdashr \aseq$ can be reduced to an equivalent sequent of the form $\left(\bigvee_{i=1}^n \aform_i\right) \vdashr \aformB_1,\dots,\aformB_m$, where $\aform_1,\dots,\aform_n,\aformB_1,\dots,\aformB_m$ are disjunction-free.
Moreover, it is clear that the latter sequent is valid iff all the sequents in the set $\{ \aform_i \vdashr \aformB_1,\dots,\aformB_m \mid 1 \leq i \leq n\}$ are valid.
 Thus we will assume in Section \ref{sect:rules} that the considered sequents
are disjunction-free.
Note that for all $i$, we have $\widt{\aform_i \vdashr \aformB_1,\dots,\aformB_m} \leq \widt{\aform \vdashr \aseq}$, hence the reduction (although exponential) preserves the complexity result in Theorem \ref{theo:comp_term}.

\begin{proposition}
\label{prop:size_card}
For every set of expressions (formulas, sequents or rules) $E$, if $\size{e} \leq n$ for all $e \in E$, then
$\card{E} = \bigO(2^{c\cdot n})$, for some constant $c$.
\end{proposition}
\begin{proof}
By definition every element $e\in E$ is a word with $\size{e} \leq n$, on a vocabulary of some fixed cardinality $m$, thus 
$\card{E} \leq m^n = 2^{n\cdot\ln(m)}$.
\end{proof}

\begin{proposition}
\label{prop:size_set}
Let $E$ be a set of expressions (formulas, sequents or rules)  such that $\size{e} \leq n$, for all $e \in E$.
Then $\size{E} \leq \bigO(2^{d\cdot n})$, for some constant $d$.
\end{proposition}
\begin{proof}
By Proposition \ref{prop:size_card}, $\card{E} = \bigO(2^{c\cdot n})$. Thus 
$\size{E} \leq n\cdot \card{E} =  \bigO(2^{d\cdot n})$ with $d = c+1$.
\end{proof}

\section{Allocated Variables and Roots}
 
 
In this section we introduce an additional restriction 
on \pcSIDs, called {\em \alloccompatibility}, and we show that every \pcSID can be reduced to 
an equivalent \alloccompatible set.
This restriction ensures that the set of free variables 
allocated by a predicate atom is the same in every unfolding. This property will be useful 
for defining some of the upcoming inference rules.
Let $\allocf$ be a function mapping each predicate symbol $p$  to a 
subset  of $\interv{1}{\ar{p}}$.
For any disjunction-free and \swfree formula $\aform$, we denote by $\alloc{\aform}$  the set 
of variables $x\in \fv{\aform}$ inductively defined as follows:
{\small 
\[
\begin{tabular}{llll}
$\alloc{\atform}$ & $=$ & $\emptyset$ & if $\atform$ is a \tformula (or $\emp$)\\
$\alloc{x \mapsto (y_1,\dots,y_\rank)}$ & $=$ & $\{ x \}$ \\
$\alloc{p(x_1,\dots,x_n)}$ & $=$ & $\{ x_i \mid i \in \alloc{p} \}$ & if $p \in \preds$  \\
$\alloc{\aform_1 * \aform_2}$ & $=$ & $ \alloc{\aform_1} \cup \alloc{\aform_2}$ \\
$\alloc{\exists x. ~\aform}$ & $=$ & $ \alloc{\aform} \setminus \{ x \}$  \\
\end{tabular}
\]
}

\newcommand{\afun}{{\frak f}}

 \newcommand{\expl}[1]{#1^*}

 \begin{definition}
 An established \pcSID $\asid$ is {\em \alloccompatible} if for all rules $\anatom \Leftarrow \aform$ in $\asid$, we have $\alloc{\anatom} = \alloc{\aform}$ .
A sequent $\aform \vdashr \aseq$ is {\em \alloccompatible} if $\asid$ is \alloccompatible.
 \end{definition}

Intuitively, $\alloc{\aform}$ is meant to contain  the free variables of $\aform$ that are allocated in  
the models of $\aform$.
 The fact that $\asid$ is \alloccompatible ensures that this set does not depend on the considered model of $\aform$.   

\begin{example} 
The set $\asid = \{ p(x,y) \Leftarrow x \mapsto (y), p(x,y) \Leftarrow x \mapsto (y) * p(y,x) \}$ is not \alloccompatible. Indeed, on one hand we have $p(x,y) \unfoldto{\asid}^* x \mapsto (y)$, and on the other hand,
 $p(x,y) \unfoldto{\asid}^* x \mapsto (y) * y \mapsto (x)$. But $\alloc{x \mapsto (y)} = \{ x \} \not = \{ x,y \}= \alloc{x \mapsto (y) * y \mapsto (x)}$.
 
 The set $\asid' = \{ p(x,y) \Leftarrow x \mapsto (y), p(x,y) \Leftarrow \exists z.~ x \mapsto (z) * p(z,x) \}$
 is \alloccompatible, with $\alloc{p} = \{ 1 \}$.
 \end{example}

 \begin{lemma}
 	\label{lem:alloc}
 	Let $\aform$ be a disjunction-free and \swfree  formula, and let $x \in \alloc{\aform}$.
 	If $(\astore,\aheap) \modelsr \aform$ and $\asid$ is \alloccompatible then 
 	$\astore(x) \in \dom{\aheap}$.
 \end{lemma}
 \begin{proof}
 	The proof is by induction on the pair $\big(\card{\dom{\aheap}},\,\size{\aform}\big)$.
 	\begin{compactitem}
 		\item{If $\aform = \emp$ or $\aform$ is a \tformula then $\alloc{\aform} = \emptyset$, which
 			contradicts our hypothesis. Thus this case cannot occur.}
 		\item{If $\aform = x' \mapsto (y_1,\dots,y_\rank)$ then $\alloc{\aform} = \{ x' \}$, hence $x = x'$. By definition we have $\aheap = \{ \astore(x'), \astore(y_1),\dots,\astore(y_\rank) \}$, hence $\astore(x) \in \dom{\aheap}$.}
 		
 		\item{If $\aform = p(x_1,\dots,x_{\#(p)})$ then $\alloc{\aform} = \{ x_i \mid i \in \alloc{p} \}$, hence $x = x_i$ for some $i \in \alloc{p}$. By definition we have $\aform \unfoldto{\asid} \aformB$ and $(\astore,\aheap) \modelsr \aformB$; by the progress condition, $\aformB$ is of the form $\exists \vec{z}.~ (x_1 \mapsto \vec{y} * p_1(u_1^1,\dots,u_{\ar{p_1}}^1) * \dots * p_n(u_1^n,\dots,u_{\ar{p_n}}^n) * \atform)$, where $\atform$ is a \tformula and $\vec{z}$ is a vector of variables not occurring in $\aform$.
 			Thus there exists a store $\astore'$, coinciding with $\astore$ on all variables not occurring in $\vec{z}$, such that
 			$(\astore',\aheap) \modelsr x_1 \mapsto \vec{y} * p_1(u_1^1,\dots,u_{\ar{p_1}}^1) * \dots * p_n(u_1^n,\dots,u_{\ar{p_n}}^n) * \atform$.
 			Since $\asid$ is \alloccompatible, $\alloc{\aformB} = \alloc{\aform}$, hence
 			either $x = x_1$ or there exists $j \in \interv{1}{n}$ and
 			$l \in \alloc{p_j}$ such that $x_i = u_l^j$. 
 			In the first case, it is clear that $\astore(x) = \astore'(x) \in \dom{\aheap}$.
 			In the second case, there is a proper subheap $\aheap_j$ of $\aheap$ such that
 			$(\astore',\aheap_j) \modelsr p_j(u_1^j,\dots,u_{\ar{p_j}}^j)$.
 			We have $u_l^j \in \alloc{p_j(u_1^j,\dots,u_{\ar{p_j}}^j)}$, hence 
 			by the induction hypothesis, 
 			$\astore'(u_l^j) \in \dom{\aheap_j}$, thus
 			$\astore(x) = \astore'(u_l^j) \in \dom{\aheap}$.
 		}  
 		\item{If $\aform = \aform_1 * \aform_2$ then we have $x\in \alloc{\aform_i}$, for some $i = 1, 2$.
 			Furthermore, there exist heaps $\aheap_i$ ($i = 1,2$) such that 
 			$(\astore,\aheap_i) \modelsr \aform_i$ (for all $i = 1,2$) and $\aheap = \aheap_1 \dunion \aheap_2$. 
 			By the induction hypothesis, we deduce that 
 			$\astore(x) \in \dom{\aheap_i}$, for some $i = 1,2$, hence
 			$\astore(x) \in \dom{\aheap}$.}

 		\item{If $\aform = \exists y. ~\aformB$ then we have $x\in \alloc{\aformB}$, and $x \not = y$.
 			Since there exists a store $\astore'$, coinciding with $\astore$ on all variables 
 			distinct from $y$, such that
 			$(\astore',\aheap) \modelsr \aformB$, by the induction hypothesis, we deduce that 
 			$\astore'(x) \in \dom{\aheap}$, hence
 			$\astore(x) \in \dom{\aheap}$.
 		}

 	\end{compactitem}
 \end{proof}

 In the remainder of the paper, we will assume that all the considered sequents are \alloccompatible. This is justified by the following:
 \begin{lemma}
 \label{lem:alloccomp}
 There exists an algorithm which, for every \swfree sequent $\aform \vdashr \aseq$,  
 computes an equivalent 
 \alloccompatible  \swfree sequent $\aform' \vdashsid{\asid'} \aseq'$. Moreover, this algorithm runs in exponential time  
and $\widt{\aform' \vdashsid{\asid'} \aseq'} = \bigO(\widt{\aform \vdashr \aseq}^2)$.
  \end{lemma}
  \begin{proof}
 We associate all pairs $(p,A)$ where $p \in \preds$ and $A \subseteq \interv{1}{\ar{p}}$ with fresh, pairwise distinct predicate symbols $p_A \in \preds$, with the same arity as $p$, and we set $\alloc{p_A} = A$.
 For each disjunction-free  formula $\aform$, we denote by $\expl{\aform}$ the set of formulas obtained from $\aform$ by replacing every predicate atom
 $p(\vec{x})$ by an atom $p_A(\vec{x})$ where $A \subseteq \interv{1}{\ar{p}}$.
 Let $\asid'$ be the set of \alloccompatible rules of the form 
 $p_A(\vec{x}) \Leftarrow \aformB$, where $p(\vec{x}) \Leftarrow \aform$ is a rule in $\asid$
and $\aformB \in \expl{\aform}$. Note that the symbols 
$p_A$ may be encoded by words of length $\bigO(\len{p} + \ar{p})$, thus for every $\aformB \in \expl{\aform}$ we have $\widt{\aformB} = \bigO(\widt{\aform}^2)$, so that $\widt{\asid'} = \bigO(\widt{\asid}^2)$. 
We show by induction on the satisfiability relation that the following equivalence holds for every structure $(\astore,\aheap)$:
$(\astore,\aheap) \modelsr \aform$ iff there exists $\aformB \in \expl{\aform}$ such that
$(\astore,\aheap) \modelssid{\asid'} \aformB$. For the direct implication, we also prove that $\alloc{\aformB} = \{ x \in \fv{\aform} \mid \astore(x) \in \dom{\aheap}\}$.
\begin{compactitem}
\item{The proof is immediate if $\aform$ is a \tformula, since $\expl{\aform} = \{ \aform \}$, and the truth value of $\aform$ does not depend on the considered \pcSID. Also, by definition $\alloc{\aform} = \emptyset$ and all the models of $\aform$ have empty heaps.}
\item{If $\aform$ is of the form $x \mapsto (y_1,\dots,y_n)$, then $\expl{\aform} = \{ \aform \}$ and the truth value of $\aform$ does not depend on the considered \pcSID. Also, $\alloc{\aform} = \{ x \}$ and
we have $\dom{\aheap} = \{ \astore(x) \}$ for every model $(\astore,\aheap)$ of $\aform$.}
\item{Assume that $\aform = p(x_1,\dots,x_{\ar{p}})$. If $(\astore,\aheap) \modelsr \aform$ then
there exists a formula $\aformC$ such that $\aform \unfoldto{\asid} \aformC$ and
$(\astore,\aheap) \modelsr \aformC$. 
By the induction hypothesis, 
there exists $\aformB \in \expl{\aformC}$ such that 
$(\astore,\aheap) \modelssid{\asid'} \aformB$ and $\alloc{\aformB} = \{ x\in \fv{\aformC}  \mid  \astore(x) \in \dom{\aheap} \}$. Let $A = \{ i \in \interv{1}{\ar{p}} \mid \astore(x_i) \in \dom{\aheap} \}$, so that $\alloc{\aformB} =\{ x_i \mid i \in A \}$.
By construction 
$p_A(x_1,\dots,x_n) \Leftarrow \aformB$ is \alloccompatible, and 
therefore $p_A(x_1,\dots,x_n) \unfoldto{\asid'} \aformB$, which entails that $(\astore,\aheap) \modelssid{\asid'} p_A(x_1,\dots,x_n)$.
By definition of $A$, $\alloc{p_A(x_1,\dots,x_n)} = \{ x\in \fv{\aform}  \mid  \astore(x) \in \dom{\aheap} \}$.

Conversely, assume that $(\astore,\aheap) \modelsr \aformB$ for some $\aformB \in \expl{\aform}$. 
Necessarily $\aformB$ is of the form $p_A(x_1,\dots,x_n)$ with $A \subseteq \interv{1}{\ar{p}}$.
We have $p_A(x_1,\dots,x_n) \unfoldto{\asid'} \aformB'$ and $(\astore,\aheap) \modelsr \aformB'$ for some formula $\aformB'$.
By definition of $\asid'$, we deduce that $p(x_1,\dots,x_n) \unfoldto{\asid} \aformC$, for some $\aformC$ such that $\aformB\in \expl{\aformC}$. 
 By the induction hypothesis, 
$(\astore,\aheap) \modelsr \aformC$, thus $(\astore,\aheap) \modelsr p(x_1,\dots,x_{\ar{p}})$. Since $p(x_1,\dots,x_{\ar{p}}) = \aform$, we have the result.
}
\item{Assume that $\aform = \aform_1 * \aform_2$.
If $(\astore,\aheap) \modelsr \aform$ then there exist disjoint heaps $\aheap_1,\aheap_2$ such that $(\astore,\aheap_i) \modelsr \aform_i$, for all $i = 1,2$ and $\aheap = \aheap_1 \dunion \aheap_2$. By the induction hypothesis, this entails that there exist formulas $\aformB_i \in \expl{\aform_i}$ for $i = 1,2$ such that 
$(\astore,\aheap_i) \modelssid{\asid'} \aformB_i$ and  $\alloc{\aformB_i} = \{ x \in \fv{\aform_i} \mid \astore(x) \in \dom{\aheap_i} \}$.  
Let $\aformB = \aformB_1 * \aformB_2$.
It is clear that  $(\astore,\aheap) \modelssid{\asid'} \aformB_1 * \aformB_2$ and  $\alloc{\aformB} = \alloc{\aformB_1 * \aformB_2} = \alloc{\aformB_1} \cup \alloc{\aformB_2} = \{ x \in \fv{\aform_1}\cup \fv{\aform_2} \mid \astore(x) \in \dom{\aheap} \} =
\{ x \in \fv{\aform} \mid \astore(x) \in \dom{\aheap} \}$. Since $\aformB_1 * \aformB_2 \in \expl{\aform}$, we obtain the result.

Conversely, assume that  there exists $\aformB \in \expl{\aform}$ such that
$(\astore,\aheap) \modelssid{\asid'} \aformB$. 
Then $\aformB = \aformB_1 * \aformB_2$ with $\aformB_i \in \expl{\aform_i}$, and we have
$(\astore,\aheap_i) \modelssid{\asid'} \aformB_i$, for $i = 1,2$ with $\aheap = \aheap_1 \dunion \aheap_2$.
Using the induction hypothesis,  we get that $(\astore,\aheap_i) \modelsr \aform_i$, hence
 $(\astore,\aheap) \modelsr \aform$.}

\item{Assume that $\aform = \exists y. \aformC$.
If $(\astore,\aheap) \modelsr \aform$ 
then $(\astore',\aheap) \modelsr \aformC$, for some store $\astore'$ coinciding with $\astore$ on every variable distinct from $y$.
By the induction hypothesis, this entails that  
there exists $\aformB \in \expl{\aformC}$ such that
$(\astore',\aheap) \modelssid{\asid'} \aformB$ and  $\alloc{\aformB} = \{ x \in \fv{\aformC} \mid \astore'(x) \in \dom{\aheap}\}$. Then 
$(\astore,\aheap) \modelssid{\asid'} \exists y. \aformB$, and we have $\exists y.\aformB \in \expl{\aform}$.
Furthermore,  $\alloc{\exists y.\aformB} = \alloc{\aformB} \setminus \{ y\} = \{ x \in \fv{\aformC} \setminus \{ y \} \mid \astore'(x) \in \dom{\aheap}\} = \{ x \in \fv{\aform} \mid \astore(x) \in \dom{\aheap}\}$. 

Conversely, assume that 
$(\astore,\aheap) \modelsr \aformB$, with $\aformB\in \expl{\aform}$.
Then $\aformB$ is of the form $\exists y. \aformB'$, with $\aformB' \in \expl{\aformC}$, thus
there exists a store $\astore'$, coinciding with $\astore$ on all variables other than $y$ 
such that $(\astore',\aheap) \modelsr \aformB'$.
By the induction hypothesis, this entails that
$(\astore',\aheap) \modelsr \aformB$, thus 
$(\astore,\aheap) \modelsr \exists y. \aformC$. Since $\exists y. \aformC = \aform$, we have the result.
} 
\end{compactitem}
 
Let $\aform',\aseq'$ be the sequence of formulas obtained from $\aform,\aseq$ by replacing every atom $\anatom$ by the disjunction of all the formulas in $\expl{\anatom}$.
It is clear that $\widt{\aform' \vdash_{\asid'} \aseq'} \leq \widt{\aform \vdash_{\asid} \aseq}^2$.
By the previous result, $\aform' \vdashsid{\asid'} \aseq'$ is equivalent to $\aform \vdashsid{\asid} \aseq$, hence $\aform' \vdashsid{\asid'} \aseq'$  fulfills all the required properties. Also, since each predicate $p$ 
is associated with $2^{\ar{p}}$  predicates $p_A$, we deduce that $\aform' \vdashsid{\asid'} \aseq'$ can be computed in time $\bigO(2^{\size{\aform \vdashsid{\asid} \aseq}})$. 
\end{proof}
 
%

\newcommand{\purelyspatial}{\constrained{\emptyset}}
\newcommand{\constrained}[1]{$#1$-constrained\xspace}
\newcommand{\nespatial}{\constrained{\{ \not \iseq \}}}
\newcommand{\espatial}{\constrained{\{ \iseq, \not \iseq \}}}
\newcommand{\mroot}{main root\xspace}
\newcommand{\aroot}{auxiliary root\xspace}

We now introduce a few notations to denote variables occurring as the  first argument of a predicate, including the $\mapsto$ predicate:

\begin{definition}
\label{def:roots}
For any disjunction-free formula $\aform$, 
we denote by $\rootsr{\aform}$ the multiset consisting of all variables  $x \in \fv{\aform}$ such that 
$\aform$ contains a subformula of  one of the forms $x \mapsto (y_1,\dots,y_n)$, $p(x,y_1,\dots,y_{\ar{p}-1})$ or $\MWs{\atail}{p(y_1,\dots,y_{\ar{p}})}{\vec{u}}{\theta}$,  with $y_1\theta = x$. 
We denote by $\rootsl{\aform}$ the multiset containing all the variables $x$ such that $\aform$ contains an atom 
$\MWs{q(z_1,\dots,z_{\ar{q}}) * \atail}{p(y_1,\dots,y_{\ar{p}})}{\vec{u}}{\theta}$  with $z_1\theta = x$. 
The variables in $\rootsr{\aform}$ are called the {\em {\mroot}s} 
of $\aform$, those in $\rootsl{\aform}$ are called the {\em auxiliary roots} of $\aform$. We let $\roots{\aform} = \rootsr{\aform} \cup \rootsl{\aform}$.

\end{definition}
 Note that a variable 
 may occur in $\roots{\aform}$ with a multiplicity greater than $1$, since
there may be several subformulas of the above forms, 
for a given $x$.

\begin{proposition}
\label{prop:root_unsat}
Let $\aform$ be a disjunction-free formula.
If $(\astore,\aheap) \modelsr \aform$ and
$x\in \rootsr{\aform}$ then $\astore(x) \in \dom{\aheap}$.
Furthermore, if $x$ occurs twice in $\rootsr{\aform}$ then 
$\aform$ is unsatisfiable.
\end{proposition}
\begin{proof}
The result is an immediate consequence of 
Proposition \ref{prop:roots} and of the definition of the semantics of points-to atoms.
\end{proof}

\newcommand{\rootunsat}{root-unsatisfiable\xspace}
\begin{definition}
A formula $\aform$ 
for which $\rootsr{\aform}$ contains multiple occurrences of the same variable 
is said to be {\em \rootunsat}. 	
\end{definition}
\newcommand{\vart}[1]{\mathit{fv}_{\theory}(#1)}

We also introduce a notation to denote the variables that may occur within  a \tformula (possibly after an unfolding):

\begin{definition}\label{def:vart}
For every \swfree formula $\aform$, we denote by $\vart{\aform}$ the set of variables $x\in \fv{\aform}$ such that
there exists a formula $\aformB$ and a \tformula $\atform$ occurring in $\aformB$ such that 
$\aform \unfoldto{\asid}^* \aformB$ and $x\in \fv{\atform}$.
\end{definition}
For instance, considering the rules of 
Example \ref{ex:rules}, we have 
$\vart{\alst(x_1,x_2) * x_3 \geq 0} = \{ x_1,x_2,x_3 \}$, since
$\alst(x_1,x_2) * x_3 \geq 0 \unfoldto{\asid}^* 
x_1 \mapsto (x_2) * x_1 \not \iseq x_2   * x_3 \geq 0$.

\begin{proposition}
\label{prop:vart}
For every \swfree formula $\aform$, 
the set $\vart{\aform}$ can be computed in polynomial time w.r.t.\ $\size{\aform}\cdot \size{\asid}$, i.e.,  w.r.t.\
$\size{\aform}\cdot \bigO(2^{d\cdot\widt{\asid}})$  for some constant $d$.
\end{proposition}
\begin{proof}
It suffices to associate all predicates $p\in \preds$ with subsets $\vart{p}$ of $\interv{1}{\ar{p}}$, inductively 
defined as the least sets satisfying the following condition:
$i \in \vart{p}$ if there exists a rule $p(x_1,\dots,x_{\ar{p}}) \Leftarrow \aform$ in $\asid$
such that $\aform$ contains either a \tformula $\atform$ with $x_i\in \fv{\atform}$,
or an atom $q(y_1,\dots,y_{\ar{q}})$ with $q\in \preds$ and $x_i = y_j$ for some $j\in \vart{q}$.
It is clear that these sets can be computed in polynomial time in $\size{\asid}$ using a standard fixpoint algorithm.
Furthermore, it is easy to check, by induction on the unfolding relation, 
that $x\in \vart{\aform}$ iff either $x\in \fv{\atform}$, for some \tformula $\atform$ occurring in $\aform$
or $x = x_i$ for some atom $p(x_1,\dots,x_\ar{p})$ occurring in $\aform$ and some $i \in \vart{p}$.

By definition, for every rule $\rho \in \asid$ we have $\size{\rho} \leq \widt{\asid}$, and
by  Proposition \ref{prop:size_set} there exists a constant $d$ such that $\size{\asid} = \bigO(2^{d.\widt{\asid}})$. 
\end{proof}

 The notation $\vart{\aform}$ is extended to formulas $\MWs{\atail}{\anatom}{\vec{x}}{\theta}$ as follows\footnote{This set is an over-approximation 
of the set of variables $x$ such that $x \in \vart{\aform'}$ for some predicate-free formula $\aform'$ with $\aform \unfoldto{\asid}^* \aform'$, but it is sufficient for our purpose.}: $\vart{\MWs{\atail}{\anatom}{\vec{x}}{\theta}} \isdef \vart{\anatom\theta}$.


\section{Eliminating Equations and Disequations}

We show that the equations and disequations can always be eliminated from established sequents, while preserving equivalence.
The intuition is that equations can be discarded by instantiating the inductive rules,
while disequations can be replaced by assertions that the considered variables 
are allocated in disjoint parts of the heap.
We wish to emphasize that the result does not follow from existing translations of SL formulas to graph grammars (see, e.g., \cite{DBLP:journals/eceasst/DoddsP08,DBLP:conf/gg/JansenGN14}), as the rules allow for disequations as well as equations. In particular the result crucially relies on the establishment property: it does not hold for non-established rules.

\begin{definition}
Let $P \subseteq \tpreds$.
A formula $\aform$ is {\em \constrained{P}} if
for every formula $\aformB$ such that $\aform \unfoldto{\asid} \aformB$, and for every symbol  $p \in \tpreds$ occurring in $\aformB$, we have $p \in P$.
A sequent $\aform \vdashr \aseq$ is {\em \constrained{P}} if
all the formulas in $\aform,\aseq$ are \constrained{P}.
\end{definition}

In particular, if $\aform$ is \purelyspatial, then the unfoldings of $\aform$ contain no symbol in $\tpreds$.
 
\begin{theorem}
\label{theo:elimeq}
Let $P \subseteq \tpreds$.
There exists an algorithm that transforms every  \constrained{P} established sequent $\aform \vdashr \aseq$
into an equivalent \constrained{(P \setminus \{ \iseq, \not \iseq \})} established sequent $\aform' \vdashsid{\asid'} \aseq'$.
This algorithm runs in  exponential time
and $\widt{\aform' \vdashsid{\asid'} \aseq'}$ is 
polynomial w.r.t.\  $\widt{\aform \vdashr \aseq}$. 
\end{theorem}

\begin{proof}
We consider a \constrained{P}  established sequent $\aform \vdashsid{\asid} \aseq$.
 This sequent is transformed in several steps, each of which is illustrated in Example \ref{ex:elimeq}.
 
\noindent\textbf{Step 1.}
The first step  consists in transforming all the formulas in $\aform,\aseq$ into disjunctions of symbolic heaps. 
Then for every  symbolic heap $\aformC$ occurring in the obtained sequent, we add all the variables freely occurring
in $\aform$ or $\aseq$ as parameters of every predicate symbol occurring in unfoldings of $\aformC$  (their arities are updated accordingly, and these variables are passed as parameters to each recursive call of a predicate symbol).
We obtain an equivalent sequent $\aform_1 \vdashsid{\asid_1} \aseq_1$, and if 
$v = \card{\fv{\aform} \cup \fv{\aseq}}$
denotes the total number of free variables occurring in $\aform,\aseq$, then using the fact that the size of each of these variables is bounded by $\widt{\aform \vdashsid{\asid} \aseq}$, we have
$\widt{\aform_1 \vdashsid{\asid_1} \aseq_1} \leq v\cdot \widt{\aform \vdashsid{\asid} \aseq}^2$.
By Definition \ref{def:sequent}
we have $v \leq \widt{\aform \vdashsid{\asid} \aseq}$,
thus $\widt{\aform_1 \vdashsid{\asid_1} \aseq_1} = \bigO(\widt{\aform \vdashsid{\asid} \aseq}^3)$.


\noindent\textbf{Step 2.}
All the equations involving an existential variable can be eliminated in a straightforward way 
by replacing each formula of the form $\exists x. (x \iseq y * \aform)$ with 
$\repl{\aform}{x}{y}$.
We then replace every formula $\exists \vec{y}. \aform$ with free variables $x_1,\dots,x_n$ by the disjunction of all the formulas
of the form 
\[\exists\vec{z}. \aform\sigma * \bigAnd_{z\in \vec{z}, z' \in \vec{z} \cup \{ x_1,\dots,x_n \}, z\not = z'} z \not \iseq z',\] 
where $\sigma$ is a substitution such that $\dom{\sigma} \subseteq \vec{y}$, $\vec{z} = \vec{y}  \setminus \dom{\sigma}$ and $\img{\sigma} \subseteq \vec{y} \cup \{ x_1,\dots,x_n\}$. 
Similarly we replace every rule $p(x_1,\dots,x_n)\Leftarrow \exists \vec{y}. \aform$ by the
the set of rules $p(x_1,\dots,x_n)\Leftarrow \exists\vec{z}. \aform\sigma * \bigAnd_{z\in \vec{z}, z' \in \vec{z} \cup \{ x_1,\dots,x_n \}, z\not = z'} z \not \iseq z'$, where $\sigma$ is any substitution satisfying the conditions above.  

Intuitively, this transformation ensures that all existential variables are associated to pairwise distinct locations, also distinct from any location associated to a free variable. The application of the substitution $\sigma$ captures all the rule instances for which this condition does not hold, by mapping  all variables that are associated with the same location to a unique representative.  
 We denote by $\aform_2 \vdashsid{\asid_2} \aseq_2$ the sequent thus obtained.
Let $v'$ be the maximal number of existential variables occurring in a rule in $\asid$. 
We have $v' \leq \widt{\aform \vdashsid{\asid} \aseq}$, since the transformation in Step $1$ adds no existential variable.
Since at most one disequation is added for every pair of variables, and the size of every variable is bounded by $\widt{\aform \vdashsid{\asid} \aseq}$, it is clear that $\widt{\aform_2 \vdashsid{\asid_2} \aseq_2} = \widt{\aform_1 \vdashsid{\asid_1} \aseq_1} + v'\cdot (v+v') \cdot(1+2*\widt{\aform \vdashsid{\asid}   \aseq})
= \bigO(\widt{\aform \vdashsid{\asid} \aseq}^3)$.

\noindent\textbf{Step 3.}
We replace every atom $\anatom = p(x_1,\dots,x_n)$ occurring in $\aform_2, \aseq_2$ or $\asid_2$ with pairwise distinct  variables $x_{i_1},\dots,x_{i_m}$ (with $m \leq n$ and $i_1 = 1$), by an atom $p_{\anatom}(x_{i_1},\dots,x_{i_m})$, where $p_{\anatom}$ is a fresh predicate symbol, associated with rules of the form 
$p_{\anatom}(y_{i_1},\dots,y_{i_m}) \Leftarrow \replall{\aformB}{y_i}{x_i}{i \in \interv{1}{n}}\theta$, where 
$p(y_1,\dots,y_n) \Leftarrow \aformB$ is a rule in $\asid$
and $\theta$ denotes the substitution $\replall{}{x_{i_k}}{y_{i_k}}{i \in \interv{1}{m}}$. 
By construction,  $p_{\anatom}(x_{i_1},\dots,x_{i_m})$ is equivalent to $\anatom$.  We denote by $\aform_3 \vdashsid{\asid_3} \aseq_3$ the resulting sequent.
It is clear that $\aform_3 \vdashsid{\asid_3} \aseq_3$ is equivalent to $\aform \vdashsid{\asid} \aseq$.

By  induction on the derivation, we can show that 
all atoms occurring in an unfolding of the formulas in the sequent $\aform_3 \vdashsid{\asid_3} \aseq_3$ are of the form $q(y_1,\dots,y_{\ar{q}})$, where $y_1,\dots,y_{\ar{q}}$ are pairwise distinct, and that the considered unfolding also contains the disequation $y_i \not \iseq y_j$, for all $i \not = j$ 
such that either $y_i$ or $y_j$ is an existential variable (note that if  $y_i$ and $y_j$ are both free then $y_i \not \iseq y_j$ is valid, since the considered stores are injective). 
This entails that the rules that introduce a trivial equality $u \iseq v$ with $u \not = v$ 
are actually redundant, since unfolding any atom  $q(y_1,\dots,y_{\ar{q}})$ using such a rule yields a formula that is unsatisfiable. Consequently such rules can be eliminated without affecting the status of the sequent. 
All the remaining equations are of form $u \iseq u$ hence can be replaced by $\emp$.
We may thus assume that the sequent $\aform_3 \vdashsid{\asid_3} \aseq_3$ contains no equality. 
Note that by the above transformation, all existential variables must be interpreted as pairwise distinct locations in any interpretation, and also be distinct from all free variables. 
It is easy to see that the fresh predicates $p_{\anatom}$ may be encoded by words
of size at most $\widt{\aform \vdashsid{\asid} \aseq}$, thus
$\widt{\aform_3 \vdashsid{\asid_3} \aseq_3} \leq \widt{\aform \vdashsid{\asid} \aseq} \cdot\widt{\aform_2 \vdashsid{\asid_2} \aseq_2} = \bigO(\widt{\aform \vdashsid{\asid} \aseq}^4)$.
By Lemma \ref{lem:alloccomp}, we may assume that
$\aform_3 \vdashsid{\asid_3} \aseq_3$ is \alloccompatible (note that the transformation given in the proof of Lemma \ref{lem:alloccomp} does not affect the disequations occurring in the rules).

\noindent\textbf{Step 4.}
We now ensure that all the  locations that are referred to are allocated.
Consider a symbolic heap $\aformC$ occurring in $\aform_3,\aseq_3$ and any $\asid_3$-model $(\astore,\aheap)$ of $\aformC$, where $\astore $ is injective.
For the establishment condition to hold, the only unallocated locations in $\aheap$ of $\aformC$ must correspond to locations $\astore(x)$ where $x$ is a free variable. 
We assume the sequent contains a free variable $u$ such that,
 for every tuple $(\ell_0,\dots,\ell_\rank)\in \aheap$, we have $\astore(u) = \ell_\rank$.
This does not entail any loss of generality, since 
we can always add a fresh variable $u$ to the considered problem:
after Step $1$, $u$ is passed as a parameter to all predicate symbols, and we may replace every points-to atom $z_0 \mapsto (z_1,\dots,z_\rank)$ occurring in $\aform_3$, $\aseq_3$ or $\asid_3$, by
$z_0 \mapsto (z_1,\dots,z_\rank,u)$ (note that this increases the value of $\rank$ by $1$).
It is clear that this ensures that $\aheap$ and $u$ satisfy the above property.
We also assume, w.l.o.g., that the sequent contains at least one variable $u'$ distinct from $u$. Note that, since $\astore$ is injective, the tuple $(\astore(u'),\dots,\astore(u'))$ cannot occur in $\aheap$, because its last component is distinct from $\astore(u)$.
We then denote by $\aform_4 \vdashsid{\asid_4} \aseq_4$ the sequent obtained from 
$\aform_3\vdashsid{\asid_3} \aseq_3$ by replacing every symbolic heap 
$\aformC$ in $\aform_3,\aseq_3$ by 
\[\left(\bigAnd_{x \in (\fv{\aform_3} \cup \fv{\aseq_3}) \setminus \alloc{\aformC}} x \mapsto (u',\dots,u')\right) * \aformC.\]
It is straightforward to check that $(\astore,\aheap)\models \aformC$ iff there exists an extension $\aheap'$ of $\aheap$
such that 
$(\astore,\aheap') \models$  
$\left(\bigAnd_{x \in (\fv{\aform_3} \cup \fv{\aseq_3}) \setminus \alloc{\aformC}} x \mapsto (u',\dots,u')\right) * \aformC$, 
with $\locs{\aheap} = \locs{\aheap'} = \dom{\aheap'}$ and
 $\aheap'(\ell) = (\astore(u'),\dots,\astore(u'))$ for all $\ell \in \dom{\aheap'} \setminus \dom{\aheap}$. 
This entails that $\aform_4 \vdashsid{\asid_4} \aseq_4$ is valid if and only if $\aform_3 \vdashsid{\asid_3} \aseq_3$ 
 is valid.

Consider a formula $\aformC$ in  $\aform_4,\aseq_4$  and some satisfiable unfolding $\aformC'$ of $\aformC$.
Thanks to the transformation in this step and the establishment condition, if $\aformC'$  contains a (free or existential) variable $x$ then it also contains an atom $x' \mapsto \vec{y}$ and 
a \tformula $\atform$ such that $\atform \modelst x \approx x'$. 
But by the above transformation, if $x,x'$ occur in the same \tformula $\atform$ then
one of the following holds; $x$ and $x'$ are identical; $x$ and $x'$ are distinct free variables (so that $\astore \modelst x \not \iseq x'$ for all injective store), or $x \not \iseq x'$ occurs in $\aformC'$. The two last cases contradict
 the fact that $\aformC'$ is satisfiable, since $\aformC' \models x \approx x'$, thus $x = x'$. Consequently, if $\aformC'$ contains a disequation $x_1 \not \approx x_2$ with $x_1\not = x_2$, then it also contains atoms $x_1 \mapsto \vec{y}_1$ and $x_2 \mapsto \vec{y}_2$. This entails that the disequation $x_1\not \iseq x_2$ is redundant, since it is a logical consequence of $x_1 \mapsto \vec{y}_1 * x_2 \mapsto \vec{y}_2$. We deduce that the satisfiability status of $\aform_4 \vdashsid{\asid_4} \aseq_4$ is preserved if all disequations are replaced by $\emp$.
\qed
\end{proof}


\begin{example}\label{ex:elimeq}
	We illustrate all of the steps in the proof above. 
	\begin{description}
		\item[Step 1.] Consider the sequent $p(x_1, x_2) \vdashsid{\asid} r(x_1) * r(x_2)$, where $\asid$ is defined as follows: $\asid = \myset{r(x) \Leftarrow x \mapsto (x)}$. After Step 1 we obtain the sequent $p(x_1, x_2) \vdashsid{\asid_1} r'(x_1, x_2) * r'(x_2, x_1)$, where $\asid_1 = \myset{r'(x, y) \Leftarrow x \mapsto (x)}$.
		\item[Step 2.] This step transforms the formula $\exists y_1\exists y_2.\, p(x, y_1) * p(x, y_2)$ into the disjunction: 
		\[\begin{array}{rl}
			\exists y_1,y_2.\, p(x, y_1) * p(x, y_2) * y_1 \not\iseq y_2 * y_1 \not \iseq x * y_2 \not \iseq x & \vee\\
			\exists y_2.\, p(x, x) * p(x,y_2) * y_2 \not \iseq x & \vee\\
			\exists y_1.\, p(x, y_1) * p(x,x) * y_1 \not \iseq x& \vee\\
			p(x,x) * p(x, x)			
		\end{array}\]
		Similarly, the rule 
		$p(x) \leftarrow \exists z \exists u.~ x \mapsto (z) * q(z,u)$ is transformed into the set: 
		\[
		\begin{array}{lll}
		p(x) & \leftarrow & x \mapsto (x) * q(x,x) \\
		p(x) & \leftarrow & \exists z.~ x \mapsto (z) * q(z,x) * z \not \iseq x  \\
		p(x) & \leftarrow & \exists u.~ x \mapsto (x) * q(x,u) * u \not \iseq x  \\
		p(x) & \leftarrow & \exists z \exists u.~ x \mapsto (z) * q(z,u) * z \not \iseq x * u \not \iseq x  * z \not \iseq u \\
\end{array}
\]
		\item[Step 3.] 
		Assume that $\asid$ contains the rules $p(y_1, y_2, y_3) \Leftarrow y_1\mapsto (y_2) * q(y_2, y_3) * y_1 \iseq y_3$ and $p(y_1, y_2, y_3) \Leftarrow y_1\mapsto (y_2) * r(y_2, y_3) * y_1 \iseq y_2$ and consider the sequent $p(x,y,x) \vdashsid{\asid} \emp$. Step 3 generates the sequent $p_\alpha(x,y) \vdashsid{\asid'} \emp$ (with $\alpha = p(x,y,x)$), where  
		$\asid'$ contains the rules $p_\alpha(y_1, y_2) \Leftarrow y_1\mapsto (y_2) * q(y_2, y_1) * y_1 \iseq y_1$ and $p_\alpha(y_1, y_2) \Leftarrow y_1\mapsto (y_2) * r(y_2, y_1) * y_1 \iseq y_2$.
		The second rule is redundant, because $p_\alpha(y_1, y_2)$ is used only in a context where $y_1  \not \iseq y_2$ holds, 
		and equation $y_1\iseq y_1$ can be discarded from the first rule.

		\item[Step 4.] 
		Let $\aformC = p(x,y,z,z') * q(x,y,z,z') * z' \mapsto (z')$, assume $\alloc{\aformC} = \myset{x,z}$, and consider the sequent $\aformC \vdashsid{\asid} \emp$. Then  $\aformC$ is replaced by $p(x,y,z,z',u) * q(x,y,z,z',u) * z' \mapsto (z',u) * u \mapsto (x,x) * y \mapsto (x,x)$ (all non-allocated variables are associated with $(x,x)$, where $x$ plays the r\^ole of the variable $u'$ in Step $4$ above). Also, every points-to atom $z_0 \mapsto (z_1)$ in $\asid$ is replaced by $z_0 \mapsto (z_1,u)$. 
		
	\end{description}
\end{example}


\begin{example}
	Consider the predicate $\ls$ defined by the rules $\{ \ls(x,y) \Leftarrow x \mapsto (y), \ls(x,y) \mapsto \exists z. x \mapsto (z) * \ls(z,y) \}$.
The (non-valid) entailment $\ls(x,y) \vdashr \alst(x,y)$, where $\alst$ is defined in Example \ref{ex:rules},
is transformed into
\[\ls''(x,y,u) \vee (y \mapsto (x,x) * \ls'(x,y,u)) \vdashr y \mapsto (x,x) * \ls'(x,y,u),\] along with the following rules, where for readability useless parameters have been removed, as well as rules with a \rootunsat right-hand side: 
\[
		\begin{array}{lll}
\ls'(x,y,u) & \Leftarrow & x  \mapsto (y,u)\\
\ls'(x,y,u) & \Leftarrow & \exists z. x \mapsto (z,u) * \ls'(z,y,u) \\
\ls''(x,y,u) & \Leftarrow & \exists z. x \mapsto (z,u) * \ls''(z,y,u) \\
\ls''(x,y,u)&  \Leftarrow & x \mapsto (y,u) * \ls'(y,y,u)
\end{array}
\]
The atom $\ls'(x,y,u)$ (resp.\ $\ls''(x,y,u)$)  denotes a list segment from $x$ to $y$ that does not allocate $y$ (resp.\ that allocates $y$). A variable $u$ is added and the value of $\rank$ is increased by $1$ as described at Step $4$.
The atom $\ls(x,y)$ is replaced by the disjunction $\ls'(x,y,u) \vee \ls''(x,y,u)$, by applying the transformation described in the proof of Lemma \ref{lem:alloccomp}. 
The predicate $\alst$ is transformed into $\ls'(x,y,u)$ at Step $4$ (as disequations are removed from the rules). The atom $y \mapsto (x,x)$ is added  since $y$ is not allocated in $\ls'(x,y,u)$. 
\end{example}


\newcommand{\myp}{\mathtt{p}}
\newcommand{\myq}{\mathtt{q}}
\newcommand{\myr}{\mathtt{r}}

\begin{example}
The (valid) entailment $\lse(x,y,x) \vdashr \lse(y,x,y)$, where $\lse(x,y,z)$ denotes a list from $x$ to $y$ containing $z$ and is defined by the rules: 
\[
		\begin{array}{lll}
\lse(x,y,z) & \Leftarrow & x \mapsto (y) * x \iseq z \\
\lse(x,y,z) & \Leftarrow & \exists z'. x \mapsto (z') * \ls(z',y) * x \iseq z \\
\lse(x,y,z) & \Leftarrow & \exists z'. x \mapsto (z') * \lse(z',y,z) 
\end{array}
\]
 is transformed into 
$\myp(x,y,u) \vdashr \myp(y,x,u)$, where (assuming $x,y,z$ are pairwise distinct)
$\myp(x,y,u)$ denotes a list from $x$ to $x$ containing $y$, 
$\myq(x,y,u)$ denotes a list from $x$ to $y$,
and
$\myr(x,y,z,u)$ denotes a list from $x$ to $y$ containing $z$, 
defined by the rules (again for readability redundant parameters and rules have been removed):
\[
		\begin{array}{lll}
\myp(x,y,u) & \Leftarrow & \exists z. x \mapsto (z) * \myr(z,x,y,u) \\
\myp(x,y,u) & \Leftarrow & x \mapsto (y,u) * \myq(y,x,u) \\
\myq(x,y,u) & \Leftarrow & x \mapsto (y,u) \\
\myq(x,y,u) & \Leftarrow & \exists z. x \mapsto (z,u) * \myq(z,y,u) \\
\myr(x,y,z,u) & \Leftarrow & \exists z'. x \mapsto (z',u) * \myr(z',y,z,u) \\
\myr(x,y,z,u) & \Leftarrow & x \mapsto (z,u) * \myq(z,y,u) 
\end{array}
\]
The predicate $\myp(x,y,u)$ is introduced at Step $3$ to replace the atoms $\lse(x,x,y)$ and $\lse(y,y,x)$. The fresh variable $u$ plays no r\^ole here because all variables are allocated.
\end{example}

\section{Heap Splitting}

\label{sect:split}

\newcommand{\splitv}[2]{\mathrm{split}_{#2}(#1)}

In this section we introduce a so-called \emph{heap splitting} operation that will 
ensure that a given variable $x$ occurs in the {\mroot}s of a given formula.
This operation preserves equivalence on the injective structures in which the considered variable $x$
is allocated (Lemma \ref{lem:split}).
It is implicitly dependent on some \pcSID $\asid$, which will always be clear from the context.
The intuition is that if an atom $p(x_1,\dots,x_n)$ (with $x_1 \not = x$) allocates $x$, then it must eventually call a predicate atom $q(x,\vec{y})$, where $\vec{y}$ may contain variables in $\{ x_1,\dots,x_n\}$ as well as fresh variables, introduced during the unfolding. Thus $p(x_1,\dots,x_n)$ can be written as a disjunction of formulas of the form:
$\exists \vec{z}. ((\MWl{q(x,\vec{y})}{p(x_1,\dots,x_n)}{\vec{u}}{\vec{u}}) * q(x,\vec{y}))$, where $\vec{z}$ denotes the vector of fresh variables mentioned above. The transformation may then be inductively extended to any formula, using the fact that a \tformula allocates no location and that $x$ is allocated by a separating conjunction $\aform_1 * \aform_2$ if it is allocated by $\aform_1$ or $\aform_2$.


\begin{definition}
\label{def:splitv}
Let $\anatom = \MWs{\atail}{p(x_1,\dots,x_n)}{\vec{u}}{\theta}$ and $x$ be a variable such that $x\not = x_1\theta$. 
We assume that $\vec{u}$ contains a variable $y$ such that $y\theta = x$\footnote{This condition is not restrictive since a fresh variable $y$ can always be added both to $\vec{u}$ and $\dom{\theta}$, and by letting $y\theta = x$.}. We denote by 
$\splitv{\anatom}{x}$ the set of 
all formulas
of the form: 
\begin{align*}
	\exists \vec{z}. ~ \left(\MWs{(\atail_1 * q(y,\vec{y}))}{p(x_1,\dots,x_n)}{(\vec{u},\vec{z})}{\theta}\right)  * \left(\MWs{\atail_2}{q(y,\vec{y})}{(\vec{u},\vec{z})}{\theta}\right),
\end{align*}
where  $\atail = \atail_1 * \atail_2$; $q \in \preds$; $p \dependson{\asid} q$;  $y,\vec{y}$ is a vector of variables of length $\ar{q}$ such that 
$\vec{z}$ denotes the vector of variables occurring in $\vec{y}$ 
but not in $\fv{\atail} \cup \{ x_1,\dots,x_n\}$. 
The function is extended to  disjunction-free formulas as follows (modulo prenex form and deletion of \rootunsat formulas): 
{\small
\[
\begin{tabular}{llll}
$\splitv{\aform}{x}$ & $=$ & $\emptyset$ \quad \text{if $\aform$ is a \tformula (possibly $\emp$)} \\
$\splitv{x \mapsto \vec{y}}{x}$ & $=$ & $\{  x \mapsto \vec{y}  \}$ \\
$\splitv{x' \mapsto \vec{y}}{x}$ & $=$ & $\emptyset$ \quad if $x\not = x'$ \\
$\splitv{\anatom}{x}$ & $=$ & $\{ \anatom \}$  if $\anatom$ is a \Watom and $\{ x\} = \rootsr{\anatom}$ \\
$\splitv{\aform_1 * \aform_2}{x}$ & $=$ &
$\{ \aformB_1 * \aform_2 \mid \aformB_1 \in \splitv{\aform_1}{x} \} \cup 
\{ \aform_1 * \aformB_2 \mid \aformB_2 \in \splitv{\aform_2}{x} \}$ \\
$\splitv{\exists y. ~\aform}{x}$ & $=$ & $\{ \exists y. \aformB \mid \aformB \in \splitv{\aform}{x} \} \cup \splitv{\repl{\aform}{y}{x}}{x}$ \\
\end{tabular}
\] }
\end{definition}
We emphasize that the vector $y,\vec{y}$ may contain variables from $\fv{\atail} \cup \{ x_1,\dots,x_n \}$ or fresh variables. Also, by construction, no variable in $\vec{z}$ can occur in $\dom{\theta}$, we can thus write, e.g.,  $\MWs{\atail_2}{q(y,\vec{y})}{\vec{u}\theta,\vec{z}}{}$ instead of $\MWs{\atail_2}{q(y,\vec{y})}{(\vec{u},\vec{z})}{\theta}$.
As usual, the equality  $\atail = \atail_1 * \atail_2$ in the definition above is to be understood modulo neutrality of $\emp$ for separating conjunctions.
 Note that the set $\splitv{\aform}{x}$ is finite, up to $\alpha$-renaming. 
In practice, it is sufficient to consider the vectors $y,\vec{y}$ such that $q(y,\vec{y})$ occurs in some unfolding 
of $p(x_1,\dots,x_n)$, up to $\alpha$-renaming (it is clear that the other formulas are redundant).

\begin{example} 
\label{ex:ls}
	Let  $\aform = \exists y. \ls(y,\nil)$ and $\asid = \{ \ls(u,v) \Leftarrow u \mapsto (v), \ls(u,v) \mapsto \exists w. u \mapsto (w) * \ls(w,v) \}$. Then we have: 
\begin{eqnarray*}
	\splitv{\ls(y,\nil)}{x}& =& \splitv{\MWl{\emp}{\ls(u,\nil)}{u}{y}}{x}\\
	& =& \{\MWl{(\emp * \ls(z,\nil))}{\ls(u, \nil)}{u,z}{y,x})  \\
	&  & \quad * (\MWl{\emp}{\ls(z,\nil)}{z}{x}) \}\\
	& = & \myset{\MWl{(\ls(z,\nil))}{\ls(u, \nil)}{u,z}{y,x}) * {\ls(x,\nil)}};\\
	\splitv{\ls(x,\nil)}{x}& =& \myset{\ls(x,\nil)}.
\end{eqnarray*}
	Note that, in this example, it is not useful to 
	consider atoms other than $\ls(z,\nil)$ for defining $\splitv{\ls(y,\nil)}{x}$, since  the only predicate atoms occurring in the unfoldings of $\ls(y,\nil)$ are renamings of 
	$\ls(z,\nil)$. Also, the variable $y$ is discarded in  $\MWl{\emp}{\ls(z,\nil)}{z}{x}$ since it does not occur in $\emp$ or $\ls(z,\nil)$ (see Remark \ref{rem:useless_var}).
	Thus  $\splitv{\aform}{x} = 
		\{ \ls(x,\nil), \exists y. ((\MWl{\ls(z,\nil)}{\ls(u,\nil)}{u,z}{y,x}) * \ls(x,\nil)) \}$. 
\end{example}


\begin{proposition}
\label{prop:allocsplit}
Let $\aform$ be a formula,  $x$ be a variable and assume $\splitv{\aform}{x} = \{ \aformC_1,\dots,\aformC_n \}$.
Then we have $\rootsr{\aformC_i} = \{ x \} \cup \rootsr{\aform}$, for every $i = 1,\dots,n$.
Thus, if $(\astore,\aheap) \modelsr \aformC_i$ for some $i \in \interv{1}{n}$, then 
$\astore(x) \in \dom{\aheap}$.
\end{proposition}
\begin{proof}
By an immediate induction on $\aform$, inspecting the different cases in the definition of $\splitv{\aform}{x}$  
and by Proposition \ref{prop:root_unsat}.
\end{proof}

Note that if $x\in \rootsr{\aform}$ 
then it is straightforward to check using Proposition \ref{prop:allocsplit} 
that $\splitv{\aform}{x} = \{ \aform \}$, up to the deletion of \rootunsat formulas. 


\begin{lemma}
\label{lem:split}
Let $\aform$ be a formula, $x$ be a variable and assume that
  $\splitv{\aform}{x} = \{ \aformC_1,\dots,\aformC_n \}$. For all structures $(\astore,\aheap)$, if $(\astore,\aheap) \modelsr \aformC_1 \vee \dots \vee \aformC_n$, then $(\astore,\aheap) \modelsr \aform$. 
Moreover, if $\astore(x)\in \dom{\aheap}$, $\astore(x') \not = \astore(x)$ for every $x' \not = x$ and $(\astore,\aheap) \modelsr \aform$, then $(\astore,\aheap) \modelsr \aformC_1 \vee \dots \vee \aformC_n$.
\end{lemma}
\newcommand{\deriv}{\mathcal{D}}
\begin{proof}
 We proceed by induction on $\aform$.
\begin{compactitem}
\item{Assume that $\aform = x' \mapsto (y_1,\dots,y_\rank)$ with $x' \not = x$.
Then $n = 0$ and the first implication trivially holds.
If $(\astore,\aheap) \modelsr \aform$, $\astore(x)  \in \dom{\aheap}$ and $\astore(x) \not = \astore(x')$, then since 
$\dom{\aheap} = \{ \astore(x') \}$, we have
$\astore(x) \not \in \dom{\aheap}$, yielding a contradiction.
}

\item{
	If $\aform = x \mapsto (y_1,\dots,y_\rank)$ then we have  
$n = 1$ and $\aformC_1 = \aform$, thus the result is immediate.}

\item{
\newcommand{\thetap}{\theta}
Assume that $\aform = \MWs{\atail}{p(x_1,\dots,x_n)}{\vec{u}}{\theta}$, where $x_1\theta \not = x$.
We prove the two results separately: 
\begin{itemize}
\item{If
	 $(\astore,\aheap) \modelsr \aformC_1 \vee \dots \vee \aformC_n$, then there exist $\atail_1,\atail_2$ and $q(y,\vec{y})$ such that
	$(\astore,\aheap) \modelsr \exists \vec{z}. ~ \MWs{\atail_1 * q(y,\vec{y})}{p(x_1,\dots,x_n)}{(\vec{u},\vec{z})}{\thetap} * \MWs{\atail_2}{q(y,\vec{y})}{(\vec{u},\vec{z})}{\thetap}$, where $\atail = \atail_1 * \atail_2$ and $\thetap$, $y$, $\vec{y}$ and $\vec{z}$ fulfill the conditions of Definition \ref{def:splitv}.   
	This entails that there exists a store $\astore'$, coinciding with $\astore$ on all variables  not occurring in $\vec{z}$, and  disjoint heaps $\aheap_1$, $\aheap_2$ such that $\aheap = \aheap_1\dunion \aheap_2$ and
\[
	\begin{tabular}{lll}
	$(\astore',\aheap_1)$ &  $\modelsr$ & $\MWs{(\atail_1 * q(y,\vec{y}))}{p(x_1,\dots,x_n)}{(\vec{u},\vec{z})}{\thetap}$, \\
	$(\astore',\aheap_2)$ & $\modelsr$ &  $\MWs{\atail_2}{q(y,\vec{y})}{(\vec{u},\vec{z})}{\thetap}$.
	\end{tabular}
	\]
By definition, there exists a formula $\aformB$ of the form $\exists \vec{v}_1. (\aformB_1 * \atail_1' * q(y',\vec{y'}))$ such that $p(x_1, \ldots, x_n) \unfoldto{\asid}^+ \aformB$,  a substitution $\sigma_1$ with $\dom{\sigma_1} \subseteq \vec{v}_1 \cap (\fv{\atail_1'} \cup \{ y' \} \cup \vec{y}')$ and $(\atail_1' * q(y',\vec{y'}))\sigma_1 = \atail_1 * q(y,\vec{y})$, and a store $\astore_1'$ coinciding with $\astore'$ on all the variables not occurring in $\vec{v}_1$ such that  $(\astore_1', \aheap_1) \modelsr \aformB_1\sigma_1\thetap$. Similarly, there exists a formula $\aformB'$ of the form $\exists \vec{v}_2. ~ (\aformB_2 * \atail_2')$ such that $q(y,\vec{y}) \unfoldto{\asid}^+ \aformB'$, a substitution $\sigma_2$ and a store $\astore_2'$ coinciding with $\astore'$ on all variables not occurring in $\vec{v}_2$ such that $\dom{\sigma_2} \subseteq \vec{v}_2 \cap \fv{\atail_2'}$, $\atail_2'\sigma_2 = \atail_2$ and  $(\astore_2', \aheap_2) \modelsr \aformB_2\sigma_2\thetap$. 
 By $\alpha$-renaming, we assume that the following condition $(\dagger)$ holds:


\begin{quote}
	For $i = 1,2$, $\vec{v}_i$ contains no variable in $\vec{v}_{3-i}  \cup \vec{z} \cup \img{\sigma_{3-i}} \cup \fv{\aformB\thetap} \cup \fv{\aformB'\thetap} \cup \{ x \} \cup \dom{\thetap}$.
\end{quote} 
	Since $q(y',\vec{y'})\sigma_1 = q(y,\vec{y})$ and  $q(y,\vec{y}) \unfoldto{\asid}^+ \exists \vec{v}_2. ~ (\aformB_2 * \atail_2')$, by Proposition \ref{prop:inst_deriv}, there exists a formula of the form $\aformB_2' * \atail_2''$ such that  $q(y',\vec{y'}) \unfoldto{\asid}^+ \exists \vec{v}_2. ~ (\aformB_2' * \atail_2'')$ and $(\aformB_2' * \atail_2'')\sigma_1 = \aformB_2 * \atail_2'$.
	We deduce that
	\[\begin{array}{rcl}
		p(x_1,\dots,x_n) &\unfoldto{\asid}^+ &\exists \vec{v}_1. (\aformB_1 * \atail_1' * q(y',\vec{y'}))\\
		& \unfoldto{\asid}^+ & \exists \vec{v}_1, \vec{v}_2. (\aformB_1 * \atail_1' * \aformB_2' * \atail_2'').
	\end{array}\]
Let $\sigma = \sigma_2 \circ \sigma_1$, note that by construction  $\dom{\sigma} \subseteq (\vec{v_1} \cup \vec{v_2})$. 	
By $(\dagger)$ we have $(\aformB_1 * \atail_1')\sigma = (\aformB_1\sigma_1 * \atail_1)\sigma_2 = \aformB_1\sigma_1 * \atail_1$ and 
	$(\aformB_2' * \atail_2'')\sigma = (\aformB_2 * \atail_2')\sigma_2 = \aformB_2\sigma_2 * \atail_2$. 
Let $\astore''$ be a store 
 such that 
	$\astore''(y') = \astore''(x)$, $\astore''(\vec{y}') = \astore''(\vec{y}\thetap)$, 
and	otherwise that coincides with 
	$\astore_i'$ on $\vec{v}_i$ and with $\astore'$ elsewhere. By construction $\astore''$ coincides with $\astore_1'$ on all the variables occurring in $\aformB_1\sigma_1\thetap$, and since $(\astore_1', \aheap_1) \modelsr \aformB_1\sigma_1\thetap$, we deduce that 
	$(\astore'', \aheap_1) \modelsr \aformB_1\sigma_1\thetap$. 
	A similar reasoning shows that $(\astore'', \aheap_2) \modelsr \aformB_2\sigma_2\thetap$ and therefore,  $(\astore'', \aheap) \modelsr \aformB_1\sigma_1\thetap * \aformB_2\sigma_2\thetap = (\aformB_1 * \aformB_2)\sigma\thetap$. 
By definition of $\astore''$,
this entails that $(\astore'', \aheap) \modelsr  (\aformB_1 * \aformB_2)\sigma'\thetap$, where $\sigma'$ is the restriction of $\sigma$ to $\fv{\atail_1' * \atail_2''}$. Indeed, all variables $y''$ such that  $y''\sigma' \not = y''\sigma$ must occur in $y',\vec{y}'$ because $\dom{\sigma_1} \subseteq \fv{\atail_1'} \cup \{ y' \} \cup \vec{y}'$ and $\dom{\sigma_2} \subseteq \fv{\atail_2'}$; therefore $\astore''(y''\sigma') = \astore''(y''\sigma\thetap)$. 
	Since  $(\atail_1' * \atail_2'')\sigma' = \atail_1'\sigma * \atail_2''\sigma = \atail_1*\atail_2 = \atail$ and $\astore'$ and $\astore''$ coincide on variables that do not occur in $\vec{v_1},\vec{v_2}$, we deduce that $(\astore',\aheap) \modelsr 
	\MWs{\atail}{p(x_1,\dots,x_n)}{\vec{u}}{\thetap}$. But no variable from $\vec{z}$ can occur in $\MWs{\atail}{p(x_1,\dots,x_n)}{\vec{u}}{\theta}$, hence $\astore$ and $\astore'$ coincide on all variables occurring in this formula, so that $(\astore,\aheap) \modelsr \MWs{\atail}{p(x_1,\dots,x_n)}{\vec{u}}{\theta}$.

 }
 \item{
Now assume that $(\astore,\aheap) \modelsr \MWl{\atail}{p(x_1,\dots,x_n)}{\vec{u}}{\vec{u}\theta}$, where $x_1\theta\neq x$,
$\astore(x)\in \dom{\aheap}$ and $\astore(x') \not = \astore(x)$ for every $x' \not = x$ ($\ddagger$). 
Then there exists a derivation
$\deriv$ such that $p(x_1,\dots,x_n) \unfoldto{\asid}^+ \exists \vec{w}. ~(\aformB * \atail')$, 
 a store $\astore'$ coinciding with $\astore$ on all variables not occurring in $\vec{w}$ and a substitution $\sigma$ with $\dom{\sigma} \subseteq \vec{w} \cap \fv{\atail'}$, such that
$(\astore',\aheap) \modelsr \aformB\sigma\theta$ and 
$\atail = \atail'\sigma$. W.l.o.g., we  assume that $\aformB$ contains no predicate atom 
and that $\vec{w}$ contains no variables in $\fv{\atail} \cup \{ x,x_1,\dots,x_n\} \cup \dom{\theta} \cup \img{\theta}$.

Since $\astore(x) \in \dom{\aheap}$, necessarily 
$\aformB\sigma$ 
contains a points-to atom of the form $x' \mapsto (u_1, \ldots, u_\rank)$,  
for some variable $x'$ such that $\astore(x) = \astore'(x) =\astore'(x'\theta)$.
Note that, by  hypothesis ($\ddagger$), we must have either $x = x'\theta$ and 
$x' \in \{ x_2,\dots,x_n\}$ or $x'\in \vec{w}$. 
In the latter case we define $\sigma' = \sigma \{ x' \leftarrow y \}$, where $y$ is a variable in $\vec{u}$ such that $y\theta = x$ (recall that $y$ necessarily exists by the assumption of Definition \ref{def:splitv}); 
otherwise we let $y = x'$ and $\sigma' = \sigma$. Since $(\astore', \aheap) \modelsr \aformB\sigma\theta$ and $\astore'(y\theta) = \astore'(x'\sigma'\theta)$, we have $(\astore', \aheap) \modelsr \aformB\sigma'\theta$.
Due to the progress condition, the derivation $\deriv$ is necessarily of the form
\[p(x_1,\dots,x_n)\ \unfoldto{\asid}^+\ \exists \vec{v}. (\aformC * q(x'',\vec{y}'))\ \unfoldto{\asid}^+\ \exists \vec{w}. ~(\aformB * \atail'),\] 
where $x' = x''\sigma$, $\vec{v}$ is a subvector of $\vec{w}$ and $q \in \preds$. 
Thus, by removing from $\deriv$ all the unfolding steps applied to $q(x'',\vec{y}')$ or its descendants, we obtain a derivation of the form
\[p(x_1,\dots,x_n)  \unfoldto{\asid}^* \exists \vec{w}_1. (\aformB_1 * q(x'',\vec{y}') * \atail_1'),\] 
for some variables $\vec{w}_1$ occurring in $\vec{w}$, and formulas $\atail_1'$ and $\aformB_1$ that are subformulas of $\atail'$ and $\aformB$ respectively.
Note that since $x_1\theta \not = x$ by hypothesis, $q(x'',\vec{y}')$ cannot coincide with $p(x_1,\dots,x_n)$, hence 
the length of this derivation is at least $1$. Similarly, by keeping in $\deriv$ only the unfolding steps applied to $q(x'',\vec{y}')$ or its descendants, 
we obtain a derivation  of the form
\[q(x'',\vec{y}')\ \unfoldto{\asid}^*\ \exists \vec{w}_2. ~(\aformB_2 * \atail_2'),\] 
where
$\atail' = \atail_1' * \atail_2'$,
$\aformB = \aformB_1 * \aformB_2$, $\vec{w} = \vec{w}_1 \cup \vec{w}_2$ and $\vec{w}_1 \cap \vec{w}_2 = \emptyset$. Note that $q(x'',\vec{y}')$ must be unfolded at least once, for the atom $x' \mapsto (u_1, \ldots, u_\rank)$ to be generated. Since $(\astore', \aheap) \modelsr \aformB\sigma'\theta$ and $\aformB = \aformB_1 * \aformB_2$, we deduce that there exist heaps $\aheap_1$ and $\aheap_2$ such that
$\aheap = \aheap_1 \dunion \aheap_2$ and
$(\astore',\aheap_i) \modelsr \aformB_i\sigma'\theta$ .

For $i= 1,2$, let $\sigma_i$ be the restriction of $\sigma'$ to the variables occurring in $\vec{w}_i$. 
By instantiating the derivation above (using Proposition \ref{prop:inst_deriv}), we get:
\[q(x',\vec{y}'\sigma') = q(x'',\vec{y}')\sigma' = q(x'',\vec{y}')\sigma_1\ \unfoldto{\asid}^*\ \exists \vec{w}_2. ~(\aformB_2\sigma_1 * \atail_2'\sigma_1),\] 
where $\aformB_2\sigma_1\sigma_2 = \aformB_2\sigma'$ and
$(\atail_2'\sigma_1)\sigma_2 = \atail_2'\sigma' = \atail_2$.
We have 
$\fv{\aformB_1 * \atail'_1} \subseteq \vec{w_1} \cup \{x_1,\dots,x_n\}$ 
and
$\fv{\aformB_2\sigma_1 * \atail'_2\sigma_1} \subseteq \vec{w_2} \cup \vec{y}'\sigma' \cup \{ x' \}$ because no unfolding derivations can introduce new free variables to a formula.
In particular,
$\atail_1'\sigma_1 = \atail_1'\sigma' = \atail_1$.

 Let $\vec{y} = \vec{y'}\sigma_1$, so that $q(x'', \vec{y'})\sigma' = q(y, \vec{y})$. 
Since  $(\astore',\aheap_i) \modelsr \aformB_i\sigma'\theta$ and $\astore'(x'\theta) = \astore'(x)$, we deduce that $(\astore', \aheap_1)\modelsr \aformB_1\sigma_1\thetap$
and $(\astore', \aheap_2)\modelsr (\aformB_2\sigma_1)\sigma_2\thetap$.
Let $\vec{z}$ be the vector of all the variables  in $\vec{y} \setminus (\fv{\atail} \cup \{ x,x_1,\dots,x_n\})$, with no repetition. 
Then since $p(x_1,\dots,x_n)  \unfoldto{\asid}^+ \exists \vec{w}_1. (\aformB_1 * q(x'',\vec{y}') * \atail_1')$ and $(\astore', \aheap_1)\modelsr \aformB_1\sigma_1\thetap$, we have  $(\astore',\aheap_1) \modelsr \MWl{(\atail_1 * q(y,\vec{y}))}{p(x_1,\dots,x_n)}{(\vec{u},\vec{z})}{(\vec{u},\vec{z})\thetap}$.
Similarly, since $q(y,\vec{y}) \unfoldto{\asid}^*\ \exists \vec{w}_2. ~(\aformB_2\sigma_1 * \atail_2'\sigma_1)$, and $(\astore', \aheap_2) \modelsr (\aformB_2\sigma_1)\sigma_2\thetap$, we deduce that $(\astore',\aheap_2 )\modelsr \MWl{\atail_2}{q(y,\vec{y})}{(\vec{u},\vec{z})}{(\vec{u},\vec{z})\thetap}$. This entails that  $(\astore',\aheap) \modelsr \MWl{\atail_1 * q(y,\vec{y})}{p(x_1,\dots,x_n)}{(\vec{u},\vec{z})}{(\vec{u},\vec{z})\thetap} * \MWl{\atail_2}{q(y,\vec{y}}{(\vec{u},\vec{z})}{(\vec{u},\vec{z})\thetap}$.

The stores $\astore'$ and $\astore$ coincide on all the variables not occurring in $\vec{z}$ and occurring in 
$\atail_1$, $\atail_2$, $p(x_1,\dots,x_n)$ or $q(x,\vec{y})$.
Indeed, if $\astore'(u) \not = \astore(u)$, for some variable $u$, then necessarily $u\in \vec{w}$ (by definition of $\astore'$), hence $u \not \in \fv{\atail} \cup \{ x_1,\dots,x_n,x \}$.
Thus, if $u\in \vec{y}$, then necessarily $u \in \vec{z}$.
We deduce that 
$(\astore,\aheap) \modelsr \exists \vec{z}.~((\MWs{\atail_1 * q(x,\vec{y})}{p(x_1,\dots,x_n)}{(\vec{u},\vec{z})}{\thetap}) * (\MWs{\atail_2}{q(x,\vec{y})}{(\vec{u},\vec{z})}{\thetap})))$; therefore,
$(\astore,\aheap) \modelsr \aformC_1 \vee \dots \vee \aformC_n$.
}
\end{itemize}

 }

\item{If $\aform = \emp$ or $\aform$ is a \tformula, then $n = 0$ hence $(\astore,\aheap) \not \modelsr \aformC_1 \vee \dots \vee \aformC_n$.
Furthermore, if $\astore(x) \in \dom{\aheap}$ then $\aheap$ is not empty hence $(\astore,\aheap) \not \modelsr \aform$.}
\item{If $\aform$ is a \Watom and $\rootsr{\aform} = \{ x \}$ then by definition 
$n = 1$ and $\aformC_1 = \aform$, thus the result is immediate.}
\item{Assume that $\aform = \aform_1 * \aform_2$ and let $\{ \aformC^i_1,\dots,\aformC^i_{n_i} \} = \splitv{\aform_i}{x}$. 
In this case, 
$\{ \aformC_1,\dots,\aformC_n  \} = \{ \aformB_{3-i} * \aformC^i_j \mid i \in \{ 1,2 \}, j \in \interv{1}{n_i} \}$.

If $(\astore,\aheap) \modelsr \aformC_1 \vee \dots \vee \aformC_n$, then there exists $i \in \{ 1,2 \}$ and $j \in \interv{1}{n_i}$ such that
$(\astore,\aheap) \modelsr \aform_{3-i} * \aformC^i_j$.
Thus 
there exist disjoint heaps $\aheap_{3-i},\aheap_i$ such that $\astore = \aheap_{3-i}\cup \aheap_i$,
$(\astore,\aheap_{3-i}) \modelsr \aform_{3-i}$
and
$(\astore,\aheap_i) \modelsr \aformC^i_j$, so that 
$(\astore,\aheap_i) \modelsr \aformC^i_1 \vee \dots \vee \aformC^i_{m_i}$. 
By the induction hypothesis, we deduce that
$(\astore,\aheap_i) \modelsr \aform_i$.
Thus 
$(\astore,\aheap_{3-i} \cup \aheap_{i}) \modelsr \aform_{3-i} * \aform_i$, i.e.,
$(\astore,\aheap) \modelsr \aform$.

If $(\astore,\aheap) \modelsr \aform$, $\astore(x)\in \dom{\aheap}$ and $\astore(x') \not = \astore(x)$ for every $x' \not = x$, then 
there exist disjoint heaps $\aheap_1,\aheap_2$ such that
$(\astore,\aheap_i) \modelsr \aform_i$ for $i=1,2$, and $\aheap = \aheap_1 \dunion \aheap_2$.
Since $\astore(x) \in \dom{\aheap}$, we must have 
$\astore(x) \in \dom{\aheap_i}$, for some $i = 1,2$.
Then, by the induction hypothesis, we deduce that 
$(\astore,\aheap_i) \modelsr \aformC^i_1 \vee \dots \vee \aformC^i_{m_i}$.
Consequently, 
$(\astore,\aheap) \modelsr \aform_{3-i} * (\aformC^i_1 \vee \dots \vee \aformC^i_{m_i}) \equiv
\bigvee_{j=1}^{n_i} (\aform_{3-i} * \aformC^i_j)$, thus
$(\astore,\aheap) \modelsr \aformC_1 \vee \dots \vee \aformC_n$.

}
\item{
Assume that $\aform = \exists y. ~ \aform'$, where $y \not = x$.
Let $\{ \aformC'_1,\dots,\aformC'_m \} = \splitv{\aform'}{x}$ 
and $\{ \aformC''_1,\dots,\aformC''_l \} = \splitv{\repl{\aform'}{y}{x}}{x}$, so that we have
$\{ \aformC_1,\dots,\aformC_n \} = \{ \exists y.~ \aformC_1',\dots,\exists y.~ \aformC_m', \aformC''_1,\dots,\aformC''_l \}$.

If $(\astore,\aheap) \modelsr \aformC_i$ for some $i = 1,\dots,n$, then we have either 
$(\astore,\aheap) \modelsr \exists y.~ \aformC_j'$ for some $j = 1,\dots,m$ or 
$(\astore,\aheap) \modelsr \aformC_j''$, for some $j = 1,\dots,l$.
In the former case, we get
$(\astore',\aheap) \modelsr \aformC_j'$, for some store $\astore'$ coinciding with $\astore$ on all variables distinct from $y$, thus $(\astore',\aheap) \modelsr \aform'$ by the induction hypothesis, and
therefore $(\astore,\aheap) \modelsr \exists y. \aform'$.
In the latter case, we get
$(\astore,\aheap) \modelsr \repl{\aform'}{y}{x}$ by the induction hypothesis, 
thus
$(\astore,\aheap) \modelsr \exists y. \aform'$.

Conversely, if $(\astore,\aheap) \modelsr \aform$, $\astore(x)\in \dom{\aheap}$ and $\astore(x') \not = \astore(x)$, for every $x' \not = x$, then 
either $(\astore,\aheap) \modelsr \repl{\aform}{y}{x}$ or 
 $(\astore',\aheap) \modelsr \aform'$, for some store $\astore'$ coinciding with   $\astore$ on all variables distinct from $y$, with $\astore'(y) \not = \astore(x)$.
By the induction hypothesis, this entails that either
$(\astore,\aheap) \modelsr \aformC''_1 \vee \dots \vee \aformC''_{l}$ 
or
$(\astore',\aheap) \modelsr \aformC'_1 \vee \dots \vee \aformC'_{m}$, so that
$(\astore,\aheap) \modelsr \exists y. \aformC'_1 \vee \dots \vee \exists y. \aformC'_{m}$.
We deduce that 
$(\astore,\aheap) \modelsr \aformC_1 \vee \dots \vee \aformC_n$.
}
\end{compactitem} 
\end{proof}

\section{The Proof Procedure}


\label{sect:rules}
\newcommand{\myrulefont}[1]{\textcolor{blue}{\rulefont{#1}}}

\newcommand{\ExistDec}{Existential Decomposition\xspace}
\newcommand{\ED}{\myrulefont{ED}}

\newcommand{\rulefont}[1]{{\tt #1}}
\newcommand{\Skolemization}{Skolemisation\xspace}
\newcommand{\Sk}{\myrulefont{Sk}}

\newcommand{\HeapFunc}{Heap Functionality\xspace}
\newcommand{\HF}{\myrulefont{HF}}

\newcommand{\Unfold}{Left Unfolding\xspace}
\newcommand{\Unf}{\myrulefont{UL}}
\newcommand{\UnitHeap}{Unit Heap axiom\xspace}
\newcommand{\UH}{\myrulefont{UHA}}
\newcommand{\PointsTo}{Points-To axiom\xspace}
\newcommand{\PTA}{\myrulefont{PTA}}
\newcommand{\RUnfold}{Right Unfolding\xspace}
\newcommand{\RU}{\myrulefont{UR}}
\newcommand{\Reflexivity}{Reflexivity axiom\xspace}
\newcommand{\Refl}{\myrulefont{R}}
\newcommand{\HeapDecomp}{Separating Conjunction Decomposition\xspace}
\newcommand{\hdec}{\myrulefont{SC}}



\newcommand{\Weakening}{Weakening\xspace}
\newcommand{\Wk}{\myrulefont{W}}

\newcommand{\HeapSep}{Heap Decomposition\xspace}
\newcommand{\hsep}{\myrulefont{HD}}

\newcommand{\TSimp}{${\cal T}$-Simplification\xspace}
\newcommand{\TS}{\myrulefont{TS}}


\newcommand{\smallerf}{\prec}

\newcommand{\TDecomp}{${\cal T}$-Decomposition\xspace}
\newcommand{\TD}{\myrulefont{TD}}

\begin{figure}
{\small
	\begin{mdframed}
		\BigCondInfRule{\Sk}{\repl{\aform}{x}{x_1} \lvdash \aseq \; \dots \; \repl{\aform}{x}{x_n} \lvdash \aseq \; \quad \repl{\aform}{x}{x'} \lvdash \aseq}
		{\exists x. ~ \aform \lvdash \aseq}
		{if $\{ x_1,\dots,x_n\} = \fv{\aform} \cup \fv{\aseq}$ and $x'$ is a fresh variable not occurring in $\aform$ or $\aseq$.
		}
		
		\BigCondInfRule{\HF}{x \mapsto (y_1,\dots,y_{\rank}) * \aform \lvdash \exists \vec{y}'. (x \mapsto (z_1,\dots,z_\rank) * \aformB)\sigma, \aseq}
		{x \mapsto (y_1,\dots,y_{\rank}) * \aform \lvdash \exists \vec{y}. (x \mapsto (z_1,\dots,z_\rank) * \aformB), \aseq}
		{if $\vec{y} \cap \{ x,y_1,\dots,y_\rank\} = \emptyset$, $\dom{\sigma} \subseteq \vec{y} \cap \{ z_1,\dots,z_\rank\}$,  $\dom{\sigma}\not = \emptyset$, 
			$\forall i \in \interv{1}{\rank}\, z_i\sigma = y_i$  and
			$\vec{y}'$ is the vector of variables occurring in $\vec{y}$ but not in $\dom{\sigma}$.
}

		\CondInfRule{\Unf}{\aform_1 * \aform \lvdash \aseq \quad \dots \quad \aform_n * \aform \lvdash \aseq}
		{p(\vec{x})  * \aform \lvdash \aseq}{
			\begin{tabular}{rrr}
				if 
				$p \in \preds$, $\{ \aform_1,\dots,\aform_n \}$ is the set \\
				of formulas such that $p(\vec{x}) \unfoldto{\asid} \aform_i$.
		\end{tabular}}

		\BigCondInfRule{\RU}{\aform \lvdash \exists \vec{x}.(\aformB_1 * \aformB),\dots,\exists \vec{x}.(\aformB_n * \aformB),  \aseq}
		{\aform \lvdash \exists \vec{x}. (\anatom * \aformB), \aseq}
		{if $\anatom$ is a \Watom, $\rootsr{\anatom} \subseteq \fv{\aform}$, $\{\aformB_1,\dots,\aformB_n\}$ is the set of {\mapformula}s (see Definition \ref{def:mapform}) such that
			$\anatom \unfoldto{\asid}\aformB_i$ (with possibly $n = 0$).}

		
			\TabRule{\Wk}{\aform \lvdash \aseqB}
			{\aform \lvdash  \aformB, \aseqB}
\vspace{4mm}

		\CondInfRule{\hsep}{\aform \lvdash \aformB_1,\dots,\aformB_n, \aseqB}
		{\aform \lvdash  \aformB, \aseqB}
		{if $x\in \alloc{\aform}\setminus  \rootsr{\aformB}$,
				$\{ \aformB_1,\dots,\aformB_n \} = \splitv{\aformB}{x}$.}
		
		\BigCondInfRule{\hdec}{\aform \lvdash \aseq_1 \quad \dots \quad \aform \lvdash \aseq_m  \qquad
			\aform' \lvdash \aseq_1' \quad \dots \quad \aform' \lvdash \aseq_l'}
		{\aform * \aform' \lvdash \aformB_1 * \aformB_1',\dots,\aformB_n * \aformB_n'}
		{if: 
			\begin{inparaenum}[(i)]
				\item{
					$\alloc{\aform} \not = \emptyset$, $\alloc{\aform'} \not = \emptyset$;}
				\item{$I_1,\dots,I_m,J_1,\dots,J_l \subseteq \interv{1}{n}$,
					for every $X \subseteq \interv{1}{n}$, either $X \supseteq I_i$ for some $i \in \interv{1}{m}$, or $\interv{1}{n} \setminus X \supseteq   J_j$ for some $j \in \interv{1}{l}$, 
					and for every $i,j \in \interv{1}{m}$ (resp.\ $i,j \in \interv{1}{l}$), with $i \not = j$
					we have $I_i \not \subseteq I_j$ (resp.\ $J_i \not \subseteq J_j$);}
				\item{
					$\aseq_i$ ($1 \leq i \leq m$) is the sequence of formulas $\aformB_j$ 
					for $j \in I_i$ and
					$\aseq_j'$ (for $1 \leq i \leq l$) is the sequence of formulas $\aformB_j'$ for $j \in J_i$.}
			\end{inparaenum}
		}
	
		\BigCondInfRule{\ED}{\aform * \aform' \lvdash \aformC_1,\dots, \aformC_m,\ \exists \vec{y}. \repl{\aformC}{x}{x_1},\ \dots,\ \exists \vec{y}. \repl{\aformC}{x}{x_n},\ \aseq} 
		{\aform * \aform' \lvdash \exists \vec{y}.\exists x. \aformC, \aseq}
		{if $\{ x_1,\dots, x_n\} = \fv{\aform} \cap \fv{\aform'}$, $x'$ is a fresh variable (not occurring in the conclusion) 
			and $\myset{\aformC_1,\dots,\aformC_m}$ is a set of formulas
			of the form $\exists \vec{y}. ((\exists x. \aformB) * \repl{\aformB'}{x}{x'} * \atformB)$, where
			$\aformC = \aformB * \aformB' * \atform$, $\aformB'\not = \emp$, both 
			$\atform$ and $\atformB$ are {\tformula}s, $\atformB \modelst \atform$  
			and $x \not \in \vart{\aformB'} \cup \fv{\atformB}$.
		}

		\CondInfRule{\TS}{\aform * \atform \lvdash \aseq }{\aform * \atform' \lvdash \aseq}
		{
			if $\atform'$ is a \tformula, $\atform \smallerf \atform'$
			and $\atform' \modelsi \atform$.
		}
	

		\CondInfRule{\TD}{\aform * \atform \lvdash \aform', \aseq \qquad 
			\aform * \atform' \lvdash \aseq}
		{\aform  \vdashr \atform * \aform', \aseq}
		{\quad if $\atform$ is a {\tformula} and $\atform \vee \atform'$ is valid.}
	\end{mdframed}
	}
	\caption{Inference rules\label{fig:rules}}
\end{figure}

The inference rules, except for the axioms, are depicted in Figure \ref{fig:rules}.
The rules are intended to be applied  
bottom-up: a rule is {\em applicable} on a sequent $\aform \vdashr \aseq$ if there exists an instance of the rule
the conclusion of which is $\aform \vdashr \aseq$. 
We 
assume from now on that all the considered sequents are disjunction-free 
(as explained above, this is without loss of generality since every formula is equivalent to a disjunction of symbolic heaps) and that the formula on the left-hand side is in prenex form.
We do not assume that the formulas occurring on the right-hand side of the sequents are in prenex form, because  one of the
inference rules (namely the \ExistDec rule) will actually 
shift existential quantifiers inside separating conjunctions. 
 Note that none of the rules rename the existential variables occurring in the conclusion (renaming is only used on the existential variables occurring in the rules). This feature will be used in the termination analysis.
  We now provide some explanations on these rules. We refer to  Figure \ref{fig:rules} for the notations.
%
%
%
%
%
%
%

The \Skolemization rule (\Sk)  gets rid of existential quantifiers on the left-hand side of the sequent. 
The rule replaces an existential variable $x$ by a new free variable $x'$. 
Note that the case where $x = x_i$ for some $i = 1,\dots,n$ must be considered apart because  {\countermodel}s must be injective.


The \HeapFunc rule  (\HF)  
exploits the fact that every location refers to at most one tuple to instantiate some existential variables occurring on the right-hand side of a sequent. 
Note that the vector $\vec{y}'$ may be empty, in which case there is  no existential quantification.

The  \Unfold rule (\Unf) unfolds a predicate atom on the left-hand side. 
 Note that the considered set of formulas is finite, up to $\alpha$-renaming, since $\asid$ is finite.
All the formulas $\aform_i * \aform$ are implicitly transformed into prenex form.

The  \RUnfold rule (\RU) unfolds a predicate atom on the right-hand side, but only when the unfolding yields a single points-to spatial atom.
Note that 
this rule always applies on  formulas $\exists \vec{x}. (\anatom * \aformB)$: in the worst case, the set $\{\aformB_1,\dots,\aformB_n\}$ is empty ($n = 0$), 
in which case the rule simply removes the considered formula from the right-hand side.
\begin{example}
With the rules of Example \ref{ex:ls}, \RU\ applies on the 
sequent $x \mapsto (y) \vdashr \exists z. \MWl{\ls(z,y)}{\ls(x,y)}{x,y,z}{x,y,z}$, yielding
$x \mapsto (y) \vdashr \exists z. x \mapsto (z)$. Indeed, $\rootsr{\MWl{\ls(z,y)}{\ls(x,y)}{x,y,z}{x,y,z}} = \mnset{x}$, $\ls(x,y) \unfoldto{\asid} \exists u. (x \mapsto (u) * \ls(u,y))$ and $\repl{\ls(u,y)}{u}{z} = \ls(z,y)$; thus $\MWl{\ls(z,y)}{\ls(x,y)}{x,y,z}{x,y,z} \unfoldto{\asid} x \mapsto (z)$.

Note that we also have $\ls(x,y) \unfoldto{\asid} x \mapsto (y) = x \mapsto (y) * \emp$, but there is no substitution $\sigma$ such that $\emp\sigma = \ls(z,y)$.
The rule also applies on $x \mapsto (y) \vdashr \MWl{\ls(y,x)}{\ls(x,y)}{x,y}{x,y}$, yielding
$x \mapsto (y) \vdashr \empseq$, since there is no substitution $\sigma$ with domain $\{ u \}$ such that 
 $\ls(u,y)\sigma = \ls(y,x)$.
\end{example}

The \Weakening  rule (\Wk) allows one to remove formulas from the right-hand side.
%
%
%
%
%

\HeapSep (\hsep) makes use of the heap splitting operation introduced in Section \ref{sect:split};  the soundness of
this rule is 
a consequence of Lemma \ref{lem:split}.

The \HeapDecomp rule (\hdec) permits the decomposition of separating conjunctions on the left-hand side, by relating
 an entailment of the form $\aform * \aform' \vdashr \aseq$ to entailments 
of the form $\aform \vdashr \aseqB$ and $\aform' \vdashr \aseqB'$. 
This is possible only if all the formulas in $\aseq$ are separating conjunctions. As we shall see, \hdec\ is always sound, 
but it is invertible\footnote{We recall that a rule is {\em invertible} if the validity of its conclusion implies the validity of  each of its premises.} only if the heap decomposition corresponding to the left-hand side coincides with that of the formulas on the right-hand side (see Definition \ref{def:svalid}).
It plays a similar  r\^ole to the rule $(*)$ defined in \cite{DBLP:conf/aplas/TatsutaNK19}\footnote{The key difference is that in our rule the premises are directly written into disjunctive normal form, rather than using universal quantifications over sets of indices and disjunction.}. 
For instance,  assume that the conclusion is $p(x) * q(y) \vdashr p_1(x) * q_1(y), p_2(x) * q_2(y)$.
Then the rule can be applied with one of the following  premises: 
\[
\begin{array}{l}
 p(x) \vdashr p_1(x) \quad p(x) \vdashr p_2(x) \quad q(y) \vdashr q_1(y), q_2(y), \\
 p(x) \vdashr p_1(x), p_2(x) \quad q(y) \vdashr q_1(y) \quad q(y) \vdashr q_2(y),\\ 
	p(x) \vdashr \empseq, \\
	q(y) \vdashr \empseq.
\end{array}
\]
Other applications, such as the one with premises
\[\begin{array}{l}
	p(x) \vdashr p_1(x) \quad p(x) \vdashr p_2(x) \quad q(y) \vdashr q_1(y) \quad q(y) \vdashr q_2(y)\\
\end{array}\]
are redundant: 
if it is provable, then the first sequence above is also provable.
The intuition of the rule is that its conclusion holds if for every model $(\astore,\aheap \dunion \aheap')$ of 
$\aform * \aform'$ with 
$(\astore,\aheap) \modelsr \aform$ and 
$(\astore,\aheap') \modelsr \aform'$, there exists $i = 1,\dots,n$ such that
$(\astore,\aheap) \modelsr \aformB_i$ and
$(\astore,\aheap') \modelsr \aformB_i'$  (note that the converse does not hold in general).
Intuitively, the premises are obtained
by putting the disjunction $\bigvee_{i=1}^n (\astore,\aheap) \modelsr \aformB_i \wedge (\astore,\aheap') \modelsr \aformB_i'$  into conjunctive normal form, by distributing the implication over conjunctions
and by replacing the entailments of the form 
$(\astore,\aheap) \modelsr \aform \wedge (\astore,\aheap') \modelsr \aform' \implies
 (\astore,\aheap) \modelsr  \Gamma \vee  (\astore,\aheap') \modelsr  \Gamma'$
 by the logically stronger conjunction of the two entailments
$(\astore,\aheap) \modelsr \aform \implies (\astore,\aheap) \modelsr \Gamma$
and
$(\astore,\aheap') \modelsr \aform' \implies
 (\astore,\aheap') \modelsr  \Gamma'$. Also, the application conditions of the rules ensure that redundant 
 sequents are discarded, such as $p(x) \vdashr p_1(x), p_2(x)$ w.r.t.\ $p(x) \vdashr p_1(x)$. 
%
%

The \ExistDec rule (\ED) allows one to shift existential quantifiers on the right-hand side
inside separating conjunctions. This rule is useful to allow for further applications of Rule \hdec.
 The condition $x \not \in \vart{\aformB'}$ can be replaced by the stronger condition 
``$\aformB'$ is \constrained{\emptyset}'', which is easier to check (one does not have to compute the set $\vart{\aformB'}$).
All the completeness results in Section \ref{sect:comp} also hold
with this stronger condition 
(and with $\atformB = \atform = \emp$). 
The intuition behind the rule is that $\aformB$ denotes the part of $\aformC$ the interpretation of which depends on the value of $x$. In most cases, this formula is unique and $m = 1$, but there are cases where several decompositions of $\aformC$ must be considered, depending on the  unfolding.
\begin{example} 
Consider a sequent
 $p(x,z) * p(y,y) \vdashr \exists z. (q(x,z) * r(x,z))$, with the rules $\{ p(u,v) \Leftarrow u \mapsto (v), q(u,v) \Leftarrow u \mapsto (v), r(u,v) \Leftarrow u \mapsto (u) \}$.
Note that the interpretation of $r(u,v)$ does not depend on $v$.
One of the premises the rule \ED\ yields is
$p(x,z) * p(y,y) \vdashr (\exists z. q(x,z)) * r(x,x')$.
Afterwards, the rule \hdec\ can be applied, yielding for instance the premises
$p(x,z) \vdashr \exists z. q(x,z)$
and
$p(y,y) \vdashr r(x,x')$.
\end{example}

The \TSimp rule (\TS) allows one to simplify {\tformula}s, depending on some external procedure, and
$\smallerf$ denotes a {\em fixed} well-founded order on {\tformula}s.
We assume that $\atform \smallerf \atform * \atformB$ for every formula $\atformB \not = \emp$, 
and
that $\atform \smallerf \atform' \implies \atformB * \atform \smallerf \atformB * \atform'$, for every formula $\atformB$.
Note that \TS\ is not necessary for  the completeness proofs in Section \ref{sect:comp}.

The \TDecomp rule (\TD) shifts {\tformula}s from the right-hand side to the left-hand side of a sequent.
In particular, the rule applies with $\atform' = \neg \atform$ if the theory is closed under negation.
The completeness results in Section \ref{sect:comp} hold under this requirement.

\newcommand{\Disj}{Disjointness Rule\xspace}
\newcommand{\DR}{\rulefont{D}$_l$}

\newcommand{\HClash}{Disjointness axiom\xspace}
\newcommand{\HC}{\myrulefont{D}}
\newcommand{\TClash}{${\cal T}$-Clash axiom\xspace}
\newcommand{\TC}{\myrulefont{TC}}
\newcommand{\EmptyHeap}{Empty Heap axiom\xspace}
\newcommand{\EH}{\myrulefont{EH}}

\paragraph*{Axioms.}

Axioms are represented in Figure \ref{fig:axioms}.
The \Reflexivity (\Refl) gets rid of trivial entailments, which can be proven simply by instantiating existential variables on the right-hand side. 
For the completeness proofs in Section \ref{sect:comp}, the case where $\sigma = \id$ is actually sufficient.
The \HClash (\HC) handles the case where the same location is allocated in two disjoint parts of the heap.
Note that (by Definition \ref{def:sequent}) the left-hand side of the sequent is \swfree, hence $\alloc{\aform}$ and 
$\alloc{\aform'}$ are well-defined.
The \TClash (\TC) handles the case where the left-hand side is unsatisfiable modulo $\theory$, while the 
\EmptyHeap (\EH)  applies when the left-hand side is a \tformula.


%
%
%
%
%
%
%
%
\begin{figure}

{\small
	\begin{mdframed}
		
		\CondInfRule{\Refl}{}
		{\aform * \atform \lvdash \exists \vec{z}. \aformB,\, \aseq}
		{\qquad\begin{tabular}{ll}
				\ if $\atform$ is a \tformula and there is a substitution $\sigma$ \\
				such that
				$\dom{\sigma} = \vec{z}$ and $\aform = \aformB\sigma$.
		\end{tabular}}

		\CondInfRule{\HC}{}
		{\aform * \aform' \lvdash \aseq}
		{\qquad if $\alloc{\aform} \cap \alloc{\aform'} \neq \emptyset$}

		%
		

		\CondInfRule{\TC}{}
		{\aform * \atform \lvdash \aseq}
		{\qquad if $\atform$ is a \tformula and $\atform \modelsi \myfalse$.}

		

		\CondInfRule{\EH}{}
		{\atform \lvdash \exists \vec{x}_1. \atformB_1, \dots, \exists \vec{x}_n. \atformB_n, \aseq}
		{\qquad\begin{tabular}{ll}
				if $\atform,\atformB_1,\dots,\atformB_n$ are {\tformula}s and \\
				$\atform \modelsi \exists \vec{x}_1. \atformB_1 \vee \dots \vee \exists \vec{x}_n. \atformB_n$.
		\end{tabular}}

	\end{mdframed}}
	\caption{Axioms\label{fig:axioms}}
\end{figure}

\subsection*{Proof Trees}

\begin{definition}
\label{def:ptree}
A {\em proof tree} is a possibly infinite tree, in which each node is labeled by a sequent, 
and if a node is labeled by some sequent $\aform \vdashr \aseq$, then its successors
are labeled by $\aform_i \vdashr \aseq_i$ with $i = 1,\dots,n$,
for some rule application 
\begin{tabular}{ccc} 
$\aform_1 \vdashr \aseq_1 \cdots \aform_n \vdashr \aseq_n$ \\
\hline 
$\aform \vdashr \aseq$
\end{tabular}.
A proof tree is {\em rational} if it contains a finite number of subtrees, up to a renaming of variables.
The {\em end-sequent} of a proof tree is the sequent labeling the root of the tree.
\end{definition}
In practice one is of course interested in constructing rational proof trees. Such rational proof trees can be infinite, but they can be represented finitely.
The cycles in a rational proof tree may be seen as applications of the induction principle.
We provide a simple example showing applications of the rules.

 \begin{example}
	\label{ex:infinite}
	Consider the \pcSID consisting of the following rules:
	\[\begin{array}{rcl}
		p(x) &\Leftarrow& \exists y,z.\, x \mapsto (y,z) * p(y) * p(z)\\
		p(x) &\Leftarrow& x \mapsto (x,x)\\
		q(x,u) &\Leftarrow& \exists y,z.\, x \mapsto (y,z) * p(y) * q(z,u)\\
		q(x,u) &\Leftarrow& x \mapsto (u,u)
	\end{array}\]
	The proof tree $\tau(x)$ below admits the end-sequent $p(x) \vdashr \exists u. q(x,u)$:
	{\small
		\begin{prooftree}
			\AxiomC{}
			\RightLabel{\Refl}
			\UnaryInfC{$x \mapsto (x,x) \vdashr \exists u. x \mapsto (u,u)$}
			\RightLabel{\RU}
			\UnaryInfC{$x \mapsto (x,x) \vdashr \exists u.  q(x,u)$}
			\AxiomC{$\pi(x)$}
			\UnaryInfC{$\exists y,z.\, x \mapsto (y,z) * p(y) * p(z) \vdashr \exists u.  q(x,u)$}
			\RightLabel{\Unf}
			\BinaryInfC{$p(x) \vdashr \exists u. q(x,u)$}
	\end{prooftree}} \noindent
	where the proof tree  
	$\pi(x)$ with end-sequent $\exists y,z.\, x \mapsto (y,z) * p(y) * p(z) \vdashr \exists u.  q(x,u)$ is defined as follows (using $\aform$ to denote the \Watom $\MWl{p(y) * q(z,u)}{q(x,u)}{x,y,z,u}{x,y,z,u}$):
	{\small
		\begin{prooftree}
			\AxiomC{}
			\RightLabel{\Refl}
			\UnaryInfC{$x \mapsto (y,z) \vdashr x \mapsto (y,z)$}
			\RightLabel{\RU}
			\UnaryInfC{$x \mapsto (y,z) \vdashr \aform$}
			\AxiomC{}
			\RightLabel{\Refl}
			\UnaryInfC{$p(y) \vdashr p(y)$}
			\AxiomC{$\tau(z)$}
			\UnaryInfC{$p(z) \vdashr \exists u. q(z,u)$}
			\RightLabel{\hdec}
			\BinaryInfC{$p(y) * p(z) \vdashr p(y) * \exists u. q(z,u)$}
			\RightLabel{\ED}
			\UnaryInfC{$p(y) * p(z) \vdashr \exists u. (p(y) *  q(z,u))$}
			\RightLabel{\hdec}
			\BinaryInfC{$x \mapsto (y,z) * p(y) * p(z) \vdashr \aform * \exists u. (p(y) * q(z,u))$}
			\RightLabel{\ED}
			\UnaryInfC{$x \mapsto (y,z) * p(y) * p(z) \vdashr \exists u. (\aform * p(y) * q(z,u))$}
			\RightLabel{\hsep}
			\UnaryInfC{$x \mapsto (y,z) * p(y) * p(z) \vdashr \exists u. (\MWl{p(y)}{q(x,y,u)}{x,y,u}{x,y,u} * p(y))$}
			\RightLabel{\hsep}
			\UnaryInfC{$x \mapsto (y,z) * p(y) * p(z) \vdashr \exists u. q(x,u)$}
			\RightLabel{\Sk}
			\UnaryInfC{$\exists z. x \mapsto (y,z) * p(y) * p(z) \vdashr \exists u. q(x,u)$}
			\RightLabel{\Sk}
			\UnaryInfC{$\exists y,z.\, x \mapsto (y,z) * p(y) * p(z) \vdashr \exists u. q(x,u)$}
		\end{prooftree}
	}
	For the sake of readability, \rootunsat formulas are removed. For example, the rule \Sk\ applied above also adds sequents with formulas on the left-hand side such as 
	$x \mapsto (y,z) * p(x) * p(z)$. 
Some weakening steps are also silently applied to dismiss irrelevant formulas from the right-hand sides of the sequents. 
	
	Note that the formula $q(x,u)$ on the right-hand side of a sequent is not unfolded unless this unfolding yields a single points-to atom. Unfolding $q(x,u)$ would be possible for this set of rules 
	because $p(x)$ and $q(x,u)$ share the same root, but in general such a strategy would not terminate. Instead the rule \hsep\ is used to perform a partial  unfolding of $q(x,u)$ with a split on variable $y$ for the first application. 
	The sequent $p(z) \vdashr \exists u. q(z,u)$ is identical to the root sequent, up to a renaming of variables. The generated proof tree is thus infinite but rational, up to a renaming of variables.
\end{example}

\section{Soundness}

We prove that the calculus is sound, in the sense that the end-sequent 
of every (possibly infinite, even irrational) proof tree is valid.
In the entire section, we assume that $\asid$ is \alloccompatible.


\begin{lemma}
\label{lem:ax}
The conclusions of the rules \Refl, \HC, \TC\ and \EH\ are all valid.
\end{lemma}
\begin{proof}
We consider each axiom separately.
\begin{compactitem}

\item[\Refl]{Assume that
$(\astore,\aheap) \modelsr \aform * \atform$.
Since $\atform$ is a \tformula, we also have $(\astore,\aheap) \modelsr \aform$.
Let $\astore'$ be a store mapping each variable $z$ to 
$\astore(z\sigma)$. Since $\dom{\sigma} = \vec{z}$, $\astore$ and $\astore'$
coincide on all variables not occurring in $\vec{z}$.
Since $\aform = \aformB\sigma$, we have 
$(\astore,\aheap) \modelsr \aformB\sigma$, 
and by definition of $\astore'$,
Thus $(\astore,\aheap) \modelsr \exists \vec{z}. \aformB$.
}

\item[\HC]{
Assume that $(\astore,\aheap) \modelsr \aform * \aform'$, where $x\in \alloc{\phi} \cap \alloc{\aform'}$.
Then
there exist disjoint heaps $\aheap_1$ and $\aheap_2$ with $\aheap = \aheap_1 \dunion \aheap_2$,
$(\astore,\aheap_1) \modelsr \aform$ and
$(\astore,\aheap_2) \modelsr \aform'$.
By Lemma \ref{lem:alloc}, since $\asid$ is \alloccompatible and $x \in \alloc{\aform} \cap \alloc{\aform'}$,
we have $\astore(x) \in \dom{\aheap_1} \cap \dom{\aheap_2}$, contradicting the fact that $\aheap_1$ and $\aheap_2$ are disjoint.}

\item[\TC]{
Assume that $(\astore,\aheap) \modelsr \aform * \atform$, where $\atform$ is a \tformula. 
Then by definition of the semantics
we must have $(\astore,\emptyset) \modelsr \atform$, hence 
$\atform$   cannot be unsatisfiable.
}
\item[\EH]{
If $(\astore,\aheap) \modelsr \atform$ then since $\atform$ is a \tformula necessarily $\aheap = \emptyset$,
hence $(\astore,\aheap) \modelsr \exists \vec{x}_i. \atformB_i$, for some $i =1,\dots,n$, by the application condition of the rule.
}

\end{compactitem}
\end{proof}

\begin{lemma}
\label{lem:sound}
The rules \HF, \Unf, \RU, \Wk, \hsep, \TS, \TD\ are sound. More precisely, if $(\astore,\aheap)$ is a \countermodel of the conclusion of the rule, then it is also a \countermodel of at least one of the premises.
\end{lemma}
\begin{proof}
We consider each rule separately:
\begin{compactitem}
\item[\HF]{The proof is immediate since 
it is clear that 
$\exists \vec{y}'. (x \mapsto (z_1,\dots,z_\rank) * \aformB)\sigma \modelsr \exists \vec{y}. (x \mapsto (z_1,\dots,z_\rank) * \aformB)$.}

\item[\Unf]{
Let $\anatom = p(\vec{x})$, and assume $(\astore,\aheap) \modelsr \anatom * \aform$
and $(\astore,\aheap) \not \modelsr \aseq$. Then $\aheap = \aheap_1 \dunion \aheap_2$, with
$(\astore,\aheap_1) \modelsr \anatom$
and
$(\astore,\aheap_2) \modelsr \aform$.
By definition of the semantics of the predicate atoms, necessarily $(\astore,\aheap_1) \modelsr \aformB$, for some $\aformB$ such that
$\anatom \unfoldto{\asid} \aformB$.
By definition of $\{ \aform_1,\dots,\aform_n \}$, there exists $i\in \interv{1}{n}$ such that $\aformB = \aform_i$ (modulo $\alpha$-renaming). We deduce that $(\astore,\aheap_1) \modelsr \aform_i$, and that $(\astore,\aheap) \modelsr \aform_i * \aform$. Therefore,
$(\astore,\aheap)$
 is a \countermodel of 
 $\aform_i * \aform\vdashr \aseq$.
 }
 \item[\RU]{
 Assume that 
 $(\astore,\aheap) \modelsr \aform$,
  $(\astore,\aheap) \modelsr \exists \vec{x}.(\aformB_1 * \aformB),\dots,\exists \vec{x}.(\aformB_n * \aformB), \aseq$ and $(\astore, \aheap) \not\modelsr \exists \vec{x}. (\anatom * \aformB), \aseq$. Then necessarily, $(\astore,\aheap) \modelsr \exists \vec{x}.(\aformB_i * \aformB)$  holds for some $i = 1,\dots,n$. Thus there exist a store $\astore'$ coinciding with $\astore$ on all variables not occurring in $\vec{x}$ and disjoint heaps $\aheap_i$ (for $i = 1,2$) such that $\aheap = \aheap_1 \dunion \aheap_2$, $(\astore',\aheap_1) \models \aformB_i$
  and
  $(\astore',\aheap_2) \models \aformB$.
Since by hypothesis $\anatom \unfoldto{\asid} \aformB_i$, we deduce that
$(\astore',\aheap_1) \models \anatom$, so that
$(\astore',\aheap_1 \dunion \aheap_2) \models \anatom * \aformB$, hence
$(\astore,\aheap) \models \exists \vec{x}. (\anatom * \aformB)$,
which contradicts our assumption.}
\item[\Wk]{The proof is immediate, since by definition every \countermodel of $\aform \vdashr \aformB,\aseq$ is 
also a  \countermodel of $\aform \vdashr \aseq$.} 
\item[\hsep]{Assume that $(\astore,\aheap) \modelsr \aform$ and that $(\astore,\aheap) \not \modelsr \aformB$. 
By Lemma \ref{lem:split}, since $\splitv{\aformB}{x} = \{ \aformB_1,\dots,\aformB_n\}$, we deduce that $(\astore,\aheap) \not \modelsr \aformB_1,\dots,\aformB_n$.
 }

 \item[\TS]{Assume that $(\astore,\aheap) \modelsr \aform * \atform'$, $\astore$ is injective 
and $(\astore,\aheap) \not \modelsr \aseq$.
Since $\atform$ is a \tformula, this entails that $(\astore,\aheap) \modelsr \aform$
and $\astore \modelst \atform'$, thus 
$\astore \modelst \atform'$.
By the application condition of the rule we have $\atform' \modelsi \atform$, hence
$\astore \modelst \atform$ because $\astore$ is injective.
We deduce that $(\astore,\aheap) \modelsr \aform * \atform$, and that
$(\astore,\aheap)$ is a \countermodel of $\aform * \atform \vdashr \aseq$.}

 \item[\TD]{Assume that $(\astore,\aheap) \modelsr \aform$
 and $(\astore,\aheap) \not \modelsr \atform * \aform', \aseq$.
 We distinguish two cases. If $(\astore,\emptyset) \modelsr \atform$ then
 $(\astore,\aheap) \modelsr \aform * \atform$, since $\aheap = \aheap \dunion \emptyset$.
  Furthermore, since $(\astore,\aheap) \not \modelsr \atform * \aform'$, necessarily,
$(\astore,\aheap) \not \modelsr \aform'$ and
 $(\astore,\aheap)$ is a \countermodel of $\aform * \atform \vdashr \aform', \aseq$.
Otherwise,
 $(\astore,\emptyset) \not \modelsr \atform$, hence
 $(\astore,\emptyset) \modelsr \atform'$ because $\atform \vee \atform'$ is valid, and
  $(\astore,\aheap) \modelsr \aform * \atform'$. Therefore,
 $(\astore,\aheap)$ is a \countermodel of $\aform * \atform' \vdashr \aseq$.
 }


\end{compactitem}
\end{proof}

\begin{lemma}
\label{lem:HD}
The rule \hdec\ is sound. More precisely, if $(\astore,\aheap)$ is a \countermodel of the rule conclusion, then at least one of the premises admits a
\countermodel $(\astore,\aheap')$, where $\aheap'$ is a proper subheap of $\aheap$.
\end{lemma}
\begin{proof}
Let $(\astore,\aheap)$ be a \countermodel of 
$\aform * \aform' \vdashr \aformB_1 * \aformB_1',\dots, \aformB_n * \aformB_n'$, and assume that for every proper subheap 
$\aheap' \subset \aheap$, 
$(\astore,\aheap')$ satisfies 
all the premises.
Necessarily, $(\astore,\aheap) \modelsr \aform * \aform'$, hence
there exist heaps $\aheap_1$ and $\aheap_2$ such that
$\aheap = \aheap_1 \dunion \aheap_2$,
$(\astore,\aheap_1) \modelsr \aform$
and
$(\astore,\aheap_2) \modelsr \aform'$.
Furthermore, since $\aform$ and $\aform'$ both contain at least one spatial atom, both $\aheap_1$ and
$\aheap_2$ must be nonempty by Proposition \ref{prop:hform}, and are therefore both proper subheaps of $\aheap$. 
Since $(\astore,\aheap) \not \modelsr \aformB_x * \aformB_x'$, for every $x \in \interv{1}{n}$, 
we have either
$(\astore,\aheap_1) \not \modelsr \aformB_x$
or
$(\astore,\aheap_2) \not \modelsr \aformB_x'$.
By gathering all the indices $x$ satisfying the first assertion, we obtain a set $X \subseteq \interv{1}{n}$ such that
$x \in X \Rightarrow (\astore,\aheap_1) \not \modelsr \aformB_x$ ($\dagger$)
and
$x \in \interv{1}{n} \setminus X \Rightarrow (\astore,\aheap_2) \not \modelsr \aformB_x'$ ($\ddagger$).

By the application condition of the rule, 
we have either $I_i \subseteq X$ for some $i \in \interv{1}{m}$, or
$J_j \subseteq \interv{1}{n} \setminus X$ for some $j \in \interv{1}{l}$. 
First assume that $I_i \subseteq X$.
Then since $(\astore,\aheap_1) \modelsr \aform$ and
$\aheap_1$ is a proper subheap of $\aheap$, we deduce that
$(\astore,\aheap_1) \modelsr \bigvee_{x\in I_i} \aformB_x$, which contradicts ($\dagger$).
Similarly if $J_j \subseteq \interv{1}{n} \setminus X$, then, since $(\astore,\aheap_2) \modelsr \aform'$ and
$\aheap_2$ is a proper subheap of $\aheap$, we have
$(\astore,\aheap_2) \modelsr \bigvee_{x\in J_j} \aformB_x'$, 
which  contradicts ($\ddagger$).
\end{proof}

\begin{lemma}
\label{lem:outsideheap}
Let $\aform$ be a formula, $x$ be a variable not occurring in $\vart{\aform}$ 
and let $(\astore,\aheap)$ be a model of $\aform$ such that $\astore(x) \not \in \locs{\aheap}$.
Then, for every store $\astore'$ coinciding with $\astore$ on all variables distinct from $x$, we have $(\astore',\aheap) \modelsr \aform$. 
\end{lemma}
 \begin{proof}
 The proof is by induction on the satisfiability relation. 
 We distinguish several cases.
 \begin{compactitem}
 \item{If $\aform$ is a \tformula then $\fv{\aform} = \vart{\aform}$, hence by hypothesis $x\not \in \fv{\aform}$ and $(\astore',\aheap) \modelsr \aform$.}
 \item{If $\aform = x_0 \mapsto (x_1,\dots,x_\rank)$ then 
since $(\astore,\aheap) \modelsr \aform$, $\aheap = \myset{(\astore(x_0),\dots,\astore(x_\rank))}$, and
$\locs{\aheap} = \{ \astore(x_i) \mid 0 \leq i \leq \rank \}$.
Since $\astore(x) \not \in \locs{\aheap}$ we deduce that $\astore(x) \not = \astore(x_i)$ for all $i = 0,\dots,\rank$, 
thus $x \not = x_i$ and $\astore'$ coincides with $\astore$ on $x_0,\dots,x_\rank$.
We conclude that $(\astore',\aheap) \modelsr \aform$.}
 \item{Assume that $\aform$ is of the form $\MWs{\atail}{\anatom}{\vec{u}}{\theta}$.
 Then there exist a formula $\exists \vec{y}.  ~ (\atail' * \aformB)$, a substitution $\sigma$ with $\dom{\sigma} \subseteq \vec{y}$ and a store $\hat{\astore}$ coinciding with 
 $\astore$ on all variables not occurring in $\vec{y}$ such that
 $\anatom \unfoldto{\asid}^+ \exists \vec{y}.  ~ (\atail' * \aformB)$, 
 $(\hat{\astore},\aheap) \modelsr \aformB\sigma\theta$ and $\atail = \atail'\sigma$.
 We assume by $\alpha$-renaming that $x\not \in \vec{y}$.
 Let $\hat{\astore}'$ be a store coinciding with $\hat{\astore}$ on all variables distinct from $x$, and such that
 $\hat{\astore}'(x) = \astore'(x)$. By construction,
 $\hat{\astore}'$ coincides with $\astore'$ on all variables not occurring in $\vec{y}$.
By definition of the extension of $\unfoldto{\asid}^*$ to formulas containing {\Watom}s,
we have $\aform \unfoldto{\asid}^* \exists \vec{y}'. \aformB\sigma\theta$, where $\vec{y}' = \vec{y} \setminus \dom{\sigma}$.
Assume that $x\in \vart{\aformB\sigma\theta}$. As  $x\not \in \vec{y}$, this entails that
$x\in \vart{\exists \vec{y}'.\aformB\sigma\theta}$, and 
 by Definition \ref{def:vart} $x\in \vart{\aform}$, which contradicts the hypothesis of the lemma.
  Thus  $x \not \in \vart{\aformB\sigma}$ and
  by the induction hypothesis we have $(\hat{\astore}',\aheap) \modelsr \aformB\sigma$. We deduce that $(\astore',\aheap) \modelsr \aform$.}
  \item{If $\aform = \aform_1 * \aform_2$ or $\aform = \aform_1 \vee \aform_2$ then the result is an immediate consequence of the induction hypothesis.}
  \item{If $\aform = \exists y. ~\aformB$ then there exists a store $\hat{\astore}$ coinciding with $\astore$ on all variables
  distinct from $y$ such that $(\hat{\astore},\aheap) \modelsr \aformB$.
Let $\hat{\astore}'$ be the store coinciding with $\astore'$ on $x$ and with $\hat{\astore}$ on all other variables.
  By the induction hypothesis we get  $(\hat{\astore}',\aheap) \modelsr \aformB$
  thus $(\astore',\aheap) \modelsr \aformB$.
  }
  \end{compactitem}
   \end{proof}

\begin{lemma}
\label{lem:soundbis}
The rules \Sk\ and \ED\ are sound. More precisely,
if $(\astore,\aheap)$ is  a \countermodel of the conclusion of the rule, then 
there exists a store  $\astore'$ such that
$(\astore',\aheap)$ is a \countermodel of the premise of the rule.
\end{lemma}
\begin{proof}
We consider each rule separately. 
\begin{compactitem}
\item[\Sk]{
Let $(\astore,\aheap)$ be a \countermodel of 
$\exists x.~\aform \vdashr \aseq$. Then  $\astore$ is injective, $(\astore,\aheap) \modelsr \exists x.~\aform$ and
$(\astore,\aheap) \not \modelsr \aseq$.
This entails that there exists a  store $\astore''$, 
coinciding with $\astore$ on all variables distinct from $x$, such that $(\astore'',\aheap) \modelsr \aform$.
Let $\{ x_1,\dots,x_n \}  =\fv{\exists x. \aform} \cup \fv{\aseq}$; note that 
by the application condition of the rule, 
$x \not \in \myset{x_1,\dots,x_n}$.
Assume that there exists $i \in \interv{1}{n}$ such that 
$\astore''(x) = \astore''(x_i) = \astore(x_i)$.
Then $(\astore'',\aheap) \modelsr \repl{\aform}{x}{x_i}$, and since  $\astore''$ and $\astore$ coincide on all variables freely occurring in $\repl{\aform}{x}{x_i}$, we have
$(\astore,\aheap) \modelsr \repl{\aform}{x}{x_i}$.
This entails that $(\astore,\aheap)$ is a \countermodel of
$\repl{\aform}{x}{x_i} \vdashr \aseq$, and the proof is completed, with $\astore' = \astore$.
Otherwise, consider any injective store $\hat{\astore}$ coinciding with $\astore''$ on all variables in $\fv{\exists x. \aform} \cup \fv{\aseq}$ and such that $\hat{\astore}(x') = \astore''(x)$.
Since $(\astore'',\aheap)  \models \aform$ and 
$(\astore,\aheap) \not \models \aseq$ we have
$(\hat{\astore},\aheap) \models \repl{\aform}{x}{x'}$ 
and
$(\hat{\astore},\aheap) \not \models \aseq$. The proof is thus completed with $\astore' = \hat{\astore}$.
 }

  \item[\ED]{Let $(\astore,\aheap)$ be a \countermodel 
of $\aform * \aform' \vdashr \exists \vec{y}.\exists x. \aformC, \aseq$.
Then $(\astore,\aheap) \modelsr \aform * \aform'$ and
$(\astore,\aheap) \not \modelsr \exists \vec{y}.\exists x. \aformC, \aseq$.
Let $\astore'$ be an 
injective store coinciding with $\astore$ on all variables distinct from 
$x'$ and such that $\astore'(x')$ is a location not occurring in $\locs{\aheap}$.
Since $x'$ does not occur free in the considered sequent, we have
$(\astore',\aheap) \modelsr \aform * \aform'$, and
$(\astore',\aheap) \not \modelsr \exists \vec{y}.\exists x. \aformC, \aseq$, so that 
$(\astore',\aheap) \not \modelsr \exists \vec{y}.\repl{\aformC}{x}{x_i}$, for every $i = 1,\dots,n$.

 Assume that $(\astore',\aheap) \modelsr \exists \vec{y}.(\exists x. \aformB) * \repl{\aformB'}{x}{x'} * \atformB$, for some formulas $\aformB,\aformB'$ and $\atform$ such that
 $\aformC = (\aformB * \aformB' * \atform)$,  $\atform$ and $\atformB$ are {\tformula}s with 
 $\atformB \models \atform$, and  $x \not \in \vart{\aformB'} \cup \fv{\atform}$.
  Then there exists a store $\astore''$ coinciding with $\astore'$ on all variables not occurring in $\vec{y}$, and disjoint heaps $\aheap_1,\aheap_2$ such that
  $\aheap = \aheap_1 \dunion \aheap_2$,
 $(\astore'',\aheap_1) \modelsr \exists x. \aformB$, 
 $(\astore'',\aheap_2) \modelsr \repl{\aformB'}{x}{x'}$
 and
 $(\astore'',\emptyset) \models \atformB$.
 This entails that there exists a store $\hat{\astore}$ coinciding with 
 $\astore''$ on all variables distinct from $x$ such that
 $(\hat{\astore},\aheap_1) \modelsr \aformB$.
 Since $x \not \in \vart{\aformB'}$, necessarily, $x' \not \in \vart{\repl{\aformB'}{x}{x'}}$.
   By Lemma \ref{lem:outsideheap}, since $\astore'(x') \not \in \locs{\aheap_2}$,
    this entails that 
  $(\hat{\astore},\aheap_2) \modelsr \repl{\aformB'}{x}{x'}$. 
  Now $(\astore'',\emptyset) \models \atformB$ and by the application condition of the rule we have $x\not \in \fv{\atformB}$, so that $(\hat{\astore},\emptyset) \modelsr \atformB$. This entails that $(\hat{\astore},\emptyset) \modelsr \atform$ because $\atformB \modelst \atform$, again by the application condition of the rule.
Hence we get  $(\hat{\astore},\aheap) \modelsr \aformB * \aformB' * \atform$,
  and
  $(\astore',\aheap) \modelsr \exists \vec{y}.\exists x. (\aformB * \aformB' * \atform) = \exists \vec{y}.\exists x. \aformC$, which contradicts our assumption.
We deduce that $(\astore',\aheap)$ is a \countermodel
of the premise.
 }
 \end{compactitem}
\end{proof}
\newcommand{\sizep}[1]{\size{#1}}

\newcommand{\mes}{\mu}
\newcommand{\Ne}[1]{N_{\exists}(#1)}
\newcommand{\wid}[1]{w(#1)}
\newcommand{\widp}[1]{wp(#1)}

 
 \newcommand{\spatialpart}[1]{{#1}^{h}}
 \newcommand{\theorypart}[1]{{#1}^{\theory}}

\newcommand{\mesp}{\tau}

We  introduce a measure on sequents.  
For every formula $\aform$ 
we denote by 
$\spatialpart{\aform}$ (resp.\ $\theorypart{\aform}$) 
the formula obtained from 
$\aform$ by replacing every \tformula (resp.\ every spatial atom) 
by $\emp$. We denote by $\Ne{\aform}$ the number of existential quantifications in the prefix 
 of $\aform$.
For all sequents $S = \aform \vdashr \aformB_1,\dots,\aformB_n$, we denote by $\mes(S)$ the tuple
$\left(\sizep{\spatialpart{\aform}},\{ \mes'(\aformB_i)  \mid i \in \interv{1}{n} \}, \theorypart{\aform} \right)$,
where 
\[\mes'(\aformB_i) = \left(\card{\alloc{\aform} \setminus \rootsr{\aformB_i}},  \sizep{\spatialpart{\aformB_i}}, \Ne{\aformB_i}, \theorypart{\aformB_i} \right).\]
The measure $\mes$ 
 is ordered by the lexicographic and multiset extensions of the natural ordering on natural numbers and of the order
 $\smallerf$ on {\tformula}s.
We assume that all variables have the same size and that the weights of the predicates in $\preds$ are chosen in such a way that the size of every predicate atom is strictly greater than that of all points-to atoms.

We say that a rule with conclusion $\aform \vdashr \aseq$ {\em decreases $\mes$} 
if the inequality $\mes(\aform \vdashr \aseq) > \mes(\aformB \vdashr \aseqB)$ holds for all the premises $\aformB \vdashr \aseqB$ of the rule. 

\begin{lemma}
\label{lem:mes}
All the rules, except \Unf, decrease $\mes$.
\end{lemma}
\begin{proof}
 By an  inspection of the rules (we refer to the definitions of the rules and of $\mes$ for notations). The result is straightforward for the rules \Wk, \Refl, {\HC} and \EH.
The rule \Sk\ decreases $\sizep{\spatialpart{\aform}}$, since it removes an existential quantifier.
The rule {\HF} has no influence on the left-hand side of the conclusion.
This rule does not remove variables from $\rootsr{\aformB_i}$, because it only instantiates existential variables and by definition, no existential variable may occur in $\rootsr{\aformB_i}$.
It eliminates at least one existential quantifier from the right-hand side of the sequence since $\dom{\sigma} \not = \emptyset$; hence there is a $\sizep{\spatialpart{\aformB_i}}$ that decreases strictly. Rule {\RU} has no influence on the left-hand side of the conclusion and replaces a \Watom 
 on the right-hand side by the conjunction of a points-to atom and a {\tformula}. By the above assumption on the weight of the predicate symbols, one of $\sizep{\spatialpart{\aformB_i}}$ decreases 
(note that by the progress condition the roots of $\aformB_i$ are unchanged, hence $\card{\alloc{\aform} \setminus \rootsr{\aformB_i}}$ cannot increase).
Rule \hsep\ decreases one of $\card{\alloc{\aform} \setminus \rootsr{\aformB_i}}$, 
since by Proposition \ref{prop:allocsplit}, we have $\rootsr{\psi_i} = \{ x \} \cup \rootsr{\phi}$ for every 
$\psi_i \in \splitv{\phi}{x}$ (furthermore, by the application condition of the rule, there exists $i = 1,\dots,n$ such that $x \in \alloc{\aform} \setminus \rootsr{\psi_i}$). 
For rule \hdec, we have  $\sizep{\spatialpart{\aform}} < \sizep{\spatialpart{(\aform* \aform')}}$ 
and $\sizep{\spatialpart{\aform'}} < \sizep{\spatialpart{(\aform* \aform')}}$ 
 since both $\aform$ and $\aform'$ contain a spatial atom.
Rule \ED\ does not affect $\alloc{\aform} \setminus \rootsr{\aformB_i}$ or $\sizep{\spatialpart{\aformB_i}}$, and the rule reduces one of the $\Ne{\aformB_i}$, since an existential quantifier 
is shifted into the scope of a separating conjunction,
by the application condition of the rule $\aformB' \not = \emp$. 
The rule \TS\ does not affect $\spatialpart{\aform}$ or the right-hand side and strictly decreases
$\theorypart{\aform}$ by definition of the rule.
The rule \TD\ does not affect $\spatialpart{\aform}$  and decreases one of the $\theorypart{\aformB_i}$ (note that, since the formula $\atform$ in the rule is a \tformula, no new variables may be added in $\alloc{\aformB}$). 
\end{proof}

Let
$\widp{\aform} = 2*m + l$, where $m$ (resp.\ $l$) denotes the number of occurrences of 
points-to atoms (resp.\ of predicate atoms) in $\aform$.
Let $\mesp(\aform \vdashr \aseq)$ be the measure defined as follows:
\[\mesp(\aform \vdashr \aseq) = (\Nh{\aform \vdashr \aseq}, -\widp{\aform}, \mes(\aform \vdashr \aseq)), \]
where $\Nh{\aform \vdashr \aseq}$ denotes the least size of a 
\countermodel  of $\aform \vdashr \aseq$, or $\infty$ if $\aform \vdashr \aseq$ is valid.
$\mesp(\aform)$ is ordered using the lexicographic extension of the usual ordering on integers and of the ordering on $\mes(\aform \vdashr \aseq)$. Note that 
this order is well-founded on non-valid sequents, because 
the number of points-to atoms in $\aform$ cannot be greater than $\Nh{\aform \vdashr \aseq}$, since any \countermodel
must satisfy the left-hand side of the sequent (hence 
$\widp{\aform} \leq 2.\Nh{\aform \vdashr \aseq}$).

\begin{lemma}
\label{lem:nonvaliddec}
Let
\begin{tabular}{c} 
$\aform_1 \vdashr \aseq_1\ \dots\ \aform_n \vdashr \aseq_n$ \\
\hline 
$\aform \vdashr \aseq$
\end{tabular} be a rule application.
If $\aform \vdashr \aseq$ admits a \countermodel, then 
there exists $i \in \interv{1}{n}$ such that
$\mesp(\aform_i \vdashr \aseq_i) < \mesp(\aform \vdashr \aseq)$.
\end{lemma}
\begin{proof}
By Lemma \ref{lem:ax} the conclusions of all axioms are valid and do not admit any \countermodel, thus the considered rule must admit at least one premise.
If the rule is $\hdec$, then the result follows immediately from  
Lemma \ref{lem:HD}.
Rule \Unf\ does not affect the least-size of {\countermodel}s 
by Lemma \ref{lem:sound}
and, since the rules are progressing, we have $\widp{p(\vec{x})*\aform} < \widp{\aform_i*\aform}$. 
By Lemma \ref{lem:mes}, all the other rules decrease $\mes$, and it is straightforward to check, by an inspection of the rules, that $\widp{\aform}$ cannot decrease (except for \hdec).
By Lemmas \ref{lem:sound} and 
\ref{lem:soundbis}, there exists $i \in \interv{1}{n}$ such that
$\Nh{\aform_i \vdashr \aseq_i} = \Nh{\aform \vdashr \aseq}$, which entails that
$\mesp(\aform_i \vdashr \aseq_i) < \mesp(\aform \vdashr \aseq)$.
\end{proof}

\begin{theorem}(Soundness)
\label{theo:sound}
Let $\asid$ be an \alloccompatible \pcSID.
If $\aform \vdashr \aseq$ is the end-sequent of a (possibly infinite) proof tree, then 
$\aform \vdashr \aseq$ is valid.
\end{theorem}

\begin{proof}
 \label{ap:soundness}
  Let \begin{tabular}{c} 
$\aform_1 \vdashr \aseq_1\ \dots\ \aform_n \vdashr \aseq_n$ \\
\hline 
$\aform \vdashr \aseq$
\end{tabular} be a rule application in the proof tree such that $\aform \vdashr \aseq$ is not valid. 
We assume, w.l.o.g., that  
$\mesp(\aform \vdashr \aseq)$ is minimal.
By Lemma \ref{lem:nonvaliddec}, there exists
$i \in \interv{1}{n}$ such that
$\mesp(\aform_i \vdashr \aseq_i) < \mesp(\aform \vdashr \aseq)$.
Note that this implies that $\aform_i \vdashr \aseq_i$ is not valid, as otherwise we would have 
$\Nh{\aform_i \vdashr \aseq_i} = \infty >  \Nh{\aform \vdashr \aseq}$, contradicting the fact that $\mesp(\aform_i \vdashr \aseq_i) < \mesp(\aform \vdashr \aseq)$.
Hence we get a contradiction about the minimality of $\aform \vdashr \aseq$.
\end{proof}
 
\begin{remark}
In general, with infinite proof trees, soundness requires some 
 well-foundedness condition, to rule out infinite paths of invalid statements (see, e.g., \cite{BrotherstonSimpson11}).
 In our case, no further condition is needed because every rule application either reduces the considered sequent into simpler ones or unfolds a predicate atom.  
\end{remark}

\section{Completeness And Termination Results}

\label{sect:comp}
\label{sect:term}

\subsection{The General Case}


We first show that every valid sequent admits a (possibly infinite) proof tree. 
Together with Theorem \ref{theo:sound}, this result shows that the calculus can be used as a semi-decision procedure for detecting non-validity: the procedure
will be ``stuck'' eventually in some branch, in the sense that one obtains a sequent on which no rule is applicable, iff the initial sequent is non-valid. 

We shall assume that all the formulas in the considered root sequent are in prenex form.
However, the sequents occurring within the proof tree will fulfill a slightly less restrictive condition, stated below:
\begin{definition}
A sequent $\aform \vdashr \aform_1,\dots,\aform_n$ is {\em \qprenex}
if $\aform$ is in prenex form and 
all the formulas $\aform_i$ that are not in prenex form are separating conjunctions of two prenex formulas.
\end{definition}

It is easy to check that the premises of a rule with a \qprenex conclusion are always \qprenex.

\begin{lemma}
\label{lem:inv}
The rules \Sk, \HF, \Unf, 
\hsep\ are invertible. More precisely, if $(\astore,\aheap)$ is a \countermodel of one of the premises then it is also a \countermodel of the 
conclusion.
\end{lemma}
\begin{proof}
Each rule is considered separately:
\begin{compactitem}
\item[\Sk]{
If $(\astore,\aheap) \modelsr \bigvee_{i=1}^n \repl{\aform}{x}{x_i} \vee \repl{\aform}{x}{x'}$
and
$(\astore,\aheap) \not \modelsr \aseq$,
then it is clear that 
$(\astore,\aheap) \modelsr \exists x. ~ \aform$, thus
$(\astore,\aheap)$ is a \countermodel of $\exists x. ~ \aform \vdashr \aseq$.
}

\item[\HF]{ 
Assume that $(\astore,\aheap) \modelsr x \mapsto (y_1,\dots,y_{\rank}) * \aform$
and
$(\astore,\aheap) \not  \modelsr \exists \vec{y}'. (x \mapsto (z_1,\dots,z_\rank) * \aformB)\sigma, \aseq$.
If
$(\astore,\aheap) \models \exists \vec{y}. (x \mapsto (z_1,\dots,z_\rank) * \aformB)$,
then there exist two disjoint heaps $\aheap_1,\aheap_2$ and a store $\astore'$ coinciding with 
$\astore$ on all variables not occurring in $\vec{y}$ such that
$\aheap = \aheap_1 \dunion \aheap_2$,
$(\astore',\aheap_1) \models x \mapsto (z_1,\dots,z_\rank)$,
and $(\astore',\aheap_2) \models \aformB$.
By the application condition of the rule, we have $\vec{y} \cap \{ x,y_1,\dots,y_\rank\} = \emptyset$, hence $\astore(x) = \astore'(x)$.
Since $(\astore,\aheap) \modelsr x \mapsto (y_1,\dots,y_{\rank}) * \aform$, we have $\aheap(\astore(x)) = (\astore(y_1),\dots,\astore(y_\rank))$, thus
$\astore'(z_i) = \astore(y_i)$, for all $i = 1,\dots,\rank$. 
If $z_i\in \dom{\sigma}$ then by definition we have $z_i\sigma = y_i$ so that $\astore'(z_i\sigma) = \astore'(y_i) = \astore(y_i) = \astore'(z_i)$. We deduce that for all variables $y$,
$\astore'(y\sigma) = \astore'(y)$  and by Proposition 
\ref{prop:subst_sw},
$(\astore',\aheap_1) \models x \mapsto (z_1,\dots,z_\rank)\sigma$
and
$(\astore',\aheap_2) \models \aformB\sigma$. 
We now show that no variable occurring in $(x \mapsto (z_1,\dots,z_\rank) * \aformB)\sigma$ but not in $\vec{y}'$ can occur in $\vec{y}$. Consider a variable $z'$ that occurs in $(x \mapsto (z_1,\dots,z_\rank) * \aformB)\sigma$ but not in $\vec{y}'$. Then $z'$ is of the form $z\sigma$ and we distinguish two cases. If $z\in \dom{\sigma}$ then since $\vec{y} \cap \{ x,y_1,\dots,y_\rank\} = \emptyset$, we have the result. Otherwise, since $\vec{y}'$ is the vector of variables occurring in $\vec{y}$ but not in $\dom{\sigma}$, we deduce that $z'$ cannot occur in $\vec{y}$ either. We deduce that $\astore$ and $\astore'$ coincide on all variables occurring in $(x \mapsto (z_1,\dots,z_\rank) * \aformB)\sigma$ but not in $\vec{y}'$. This entails that $(\astore,\aheap) \modelsr \exists \vec{y}'. (x \mapsto (z_1,\dots,z_\rank) * \aformB)\sigma$, which contradicts our assumption.
}
\item[\Unf]{
Assume that $(\astore,\aheap)$ is a \countermodel of $\aform_i * \aform \modelsr \aseq$, for some $i \in \interv{1}{n}$, say, $i = 1$.
Then $(\astore,\aheap) \modelsr \aform_1 * \aform$ and
 $(\astore,\aheap) \not \modelsr \aseq$.
 We deduce that there exist disjoint heaps $\aheap_1,\aheap_2$ such that
 $\aheap = \aheap_1 \dunion \aheap_2$,
 $(\astore,\aheap_1) \modelsr \aform_1$
 and
 $(\astore,\aheap_2) \modelsr \aform$.
By definition of the semantics of the predicate atom, since $\anatom \unfoldto{\asid} \aform_1$, 
we have
 $(\astore,\aheap_1) \modelsr \anatom$, thus
 $(\astore,\aheap) \modelsr \anatom * \aform$. Therefore, $(\astore,\aheap)$ is a \countermodel of $\anatom * \aform \vdashr \aseq$. 
 
}

 \item[\hsep]{Assume that $(\astore,\aheap) \modelsr \aform$ and that $(\astore,\aheap) \not \modelsr \aformB_1,\dots,\aformB_n$. By the application condition of the rule, we have $x \in \alloc{\aform}$, thus, by Lemma \ref{lem:alloc}, necessarily $\astore(x) \in \dom{\aheap}$. 
 Since $\astore$ is injective by hypothesis, we deduce by Lemma \ref{lem:split} that $(\astore,\aheap) \not \modelsr \aformB$.
 }
 
%
 

\end{compactitem}
\end{proof}

\begin{theorem}
\label{theo:comp}
If $\theory$ is closed under negation, then 
for all valid disjunction-free, prenex and \alloccompatible sequents $\aform \vdashr \aseq$ there exists a (possibly infinite and irrational)
proof tree with end-sequent $\aform \vdashr \aseq$.
\end{theorem}
\begin{proof}
It suffices to show that for every valid, disjunction-free, \qprenex sequent $\aform \vdashr \aseq$ there exists 
a rule application with conclusion $\aform \vdashr \aseq$ such that 
all the premises are valid; this ensures that an infinite proof tree can be constructed in an iterative way. 
Note that in particular this property holds if an invertible rule is applicable on $\aform \vdashr \aseq$, since 
all premises are valid in this case.
 The sequence $\aseq$ can be written as $\exists \vec{x}_1.\aformB_1,\dots,\exists \vec{x}_n.\aformB_n$ where no $\aformB_i$ is an existential 
formula.
If $\aform$ contains an existential variable or a predicate atom then one of the rules {\Sk} or \Unf\ applies and both rules are
invertible, by Lemma \ref{lem:inv}.
Thus we may assume that $\aform$ is of the form 
$x_1 \mapsto \vec{y}_1 * \dots * x_m \mapsto \vec{y}_m * \atform$, where $\atform$ is a \tformula.
Since by Lemma \ref{lem:inv} \HF\ and   
\hsep\ are invertible, we may assume that $\aform\vdashr \aseq$ is irreducible w.r.t.\  these rules (and also w.r.t.\ the axioms \HC, \TC, \Refl\ and \EH).
Note that for every model $(\astore,\aheap)$ of $\aform$, $\card{\dom{\aheap}} =  m$.
By irreducibility w.r.t.\ \HC\ and \TC, we assume that the variables $x_1,\dots,x_m$ are pairwise distinct, and that $\atform$
is satisfiable (on injective structures), so that $\aform$ admits an injective model.

If $m = 0$, then $\aform$ is a \tformula and every model of $\aform$ must be of the form $(\astore,\emptyset)$. Let $I$ be the set of indices $i \in \interv{1}{n}$ such that $\aformB_i$ is a \tformula.
Since $\aform \vdashr \aseq$ is valid, for every model 
$(\astore,\emptyset)$ of $\aform$, there exists an $i \in \interv{1}{n}$ such that
$(\astore,\emptyset) \models  \exists \vec{x}_i.\aformB_i$.
By Proposition \ref{prop:hform}, $\aformB_i$ cannot contain any spatial atom, and therefore $i \in I$.
Therefore, $\aform \models  \bigvee_{i\in I} \exists \vec{x}_i.\aformB_i$ and  rule \EH\ applies, which contradicts our assumption.

We now assume that $m > 0$.
Since \hsep\ is not applicable on $\aform \vdashr \aseq$, 
for every $i = 1,\dots,n$, and $j = 1,\dots,m$, 
$\aformB_i$ 
is such that 
$x_j \in \rootsr{\aformB_i}$. 
Assume that there exists $i \in \interv{1}{n}$ such that
$\aformB_i$ is of the form $\anatom  * \aform_i'$, where  
$\rootsr{\anatom} = \{ x_1 \}$ and $\anatom$ is a \Watom.
Then the rule \RU\ applies on $\anatom * \aform_i'$, yielding a 
premise of the form 
$\aform \vdashr \aformC_1,\dots,\aformC_p, \aseq'$, where $\aseq'$ is the sequence of formulas
$\exists \vec{x}_j. \aformB_j$ with $j \not = i$.
We show that this premise is valid.
Assume for the sake of contradiction that it admits a \countermodel $(\astore,\aheap)$.
Then $(\astore,\aheap) \models \aform$, and since
$\aform \vdashr \aseq$ is valid, necessarily, $(\astore,\aheap) \models \aseq$.
But $(\astore,\aheap) \not \models \aseq'$, hence 
$(\astore,\aheap) \models \exists \vec{x}_i. \aformB_i$ and 
there exists a store $\astore'$ coinciding with $\astore$ on all variables not occurring in $\vec{x}_i$
and heaps $\aheap_1,\aheap_2$ such that $\aheap = \aheap_1 \dunion \aheap_2$,
$(\astore',\aheap_1) \models  \anatom$
and
$(\astore',\aheap_2) \models \aform_i'$.
This entails that 
$\anatom \unfoldto{\asid}^+ \exists \vec{u}. \aformB$, 
 and there exists 
a store $\astore''$ coinciding with $\astore'$ on all variables not occurring in $\vec{u}$, 
and a substitution $\sigma$ such that $\dom{\sigma} \subseteq \vec{u}$
and 
$(\astore'',\aheap_1) \models \aformB\sigma$.
W.l.o.g., we assume that $\aformB$ contains no predicate symbol. 
The formula $\aform_i'$ contains at least $m-1$ spatial atoms 
(one atom for each variable $x_2,\dots,x_m$), 
hence $\card{\dom{\aheap_2}} \geq m-1$, since by Proposition \ref{prop:roots}, every such atom allocates at least one variable. We deduce that $\card{\dom{\aheap_1}} \leq \card{\dom{\aheap}}- (m-1) = 1$. This entails 
by Proposition \ref{prop:roots} that 
$\aformB$ contains at most one spatial atom, and since $\aformB$ contains no predicate atom, 
this atom must be a points-to atom. Because of the progress condition, each unfolding of a predicate atom 
 introduces exactly one points-to atom, 
 thus the derivation from $\anatom$ to $\exists \vec{u}. \aformB$ is of length $1$, i.e.,
 we have $\anatom \unfoldto{\asid} \exists \vec{u}. \aformB$. 
By definition of the rule \RU, the formula $\exists \vec{x}_i.\exists \vec{u}'. \aformB\sigma$, occurs in $\aformC_1,\dots,\aformC_p$ (where $\vec{u}'$ is defined in the rule \RU),  thus
$(\astore,\aheap) \models \aformC_1,\dots,\aformC_p$ 
and the premise is therefore valid.

We now assume that 
every formula $\aformB_i$ is of the form 
$x_1 \mapsto \vec{z}_i * \aformB_i'$.
First suppose that $\vec{z}_i \cap \vec{x}_i \not = \emptyset$, for some $i =1,\dots,n$. 
If for all models $(\hat{\astore},\hat{\aheap})$ of $\aform$ such that $\hat{\astore}$ is injective, we have $(\hat{\astore},\hat{\aheap}) \not \models \exists \vec{x}_i. \aformB_i$, then we may apply the rule \Wk, to remove the formula $\exists \vec{x}_i. \aformB_i$ and obtain a premise that is valid. Thus we assume that 
$(\hat{\astore},\hat{\aheap}) \models \exists \vec{x}_i. \aformB_i$ for at least one model $(\hat{\astore},\hat{\aheap})$  of $\aform$ such that $\hat{\astore}$ is injective.
Then we have $\hat{\aheap}(\hat{\astore}(x_1)) = \hat{\astore}(\vec{y}_1)$, and
there exists a store $\hat{\astore}'$ coinciding with $\hat{\astore}$ on all variables not occurring in $\vec{x}_i$ such that $\hat{\aheap}(\hat{\astore}(x_1)) = \hat{\astore}'(\vec{z}_i)$, hence $\hat{\astore}'(\vec{z}_i) = \hat{\astore}(\vec{y}_1)$.
Since $\hat{\astore}$ is injective, this entails that there exists a substitution $\sigma_i$ with domain 
$\vec{z}_i \cap \vec{x}_i$ such that $\sigma_i(\vec{z}_i) = \vec{y}_1$.
Since  $\vec{z}_i \cap \vec{x}_i \not = \emptyset$ we have $\dom{\sigma_i} \not = \emptyset$, thus the rule \HF\ applies, and the proof is completed.

We now assume that $\vec{z}_i \cap \vec{x}_i = \emptyset$, for all $i = 1,\dots,n$.
If $\vec{x}_i$ is not empty then \ED\  
applies. 
Indeed, by letting $\aformC = \aformB_i$, $m=1$, $\atform = \atformB = \emp$, $\aformB = \aformB_i'$ and $\aformB' = (x_1 \mapsto \vec{z}_i)$, the application conditions of the rule are fulfilled because  $\emp \modelst \emp$ and 
$\vart{x_1 \mapsto \vec{z}_i} \cup \fv{\atformB} = \emptyset$. The application of {\ED} shifts all variables $\vec{x}_i$ behind the formula $\aformB_i'$.
Note that this application of the rule preserves the validity of the sequent, because
$x_1 \mapsto \vec{z}_i$ contains
no variable in $\vec{x}_i$. 

We finally  assume that {every} $\vec{x}_i$ is empty, and therefore that
$\aseq$ is of the form $x_1 \mapsto \vec{z}_1 * \aformB_1', \dots, x_1 \mapsto \vec{z}_n * \aformB_n'$. 
We distinguish two cases.
\begin{compactitem}
\item{$m = 1$.
If there exists $i\in \interv{1}{n}$ such that $\aformB_i'$ is a \tformula (other than $\emp$) 
then the rule \TD\ applies, since $\theory$ is closed under negation 
(by letting $\atform = \aformB_i'$ and $\atform' = \neg \aformB_i'$).
Moreover, it is clear that the obtained sequent is valid.
Now let $i \in \interv{1}{n}$ and suppose $\aformB_i'$ is not a \tformula. 
Then all the models of $\aformB_i$ 
are of size at least $2$, and thus no model of $\aform$ can satisfy
$\aformB_i$, since 
$\aform$ admits only models of cardinality $m = 1$.
Similarly, if $\vec{z}_i \not = \vec{y}_1$, then for all injective models $(\astore,\aheap)$ 
of $\aformB_i$ we have $\aheap(\astore(x_1)) = \astore(\vec{z}_i) \not = \astore(\vec{y}_1)$, thus
$(\astore,\aheap) \not \models \aform$. 
Since  $\aform \vdashr \aseq$ is valid and $\aform$ is satisfiable, this entails that 
there exists 
$i = 1,\dots,n$ such that $\aformB_i' = \emp$ and $\vec{z}_i = \vec{y}_1$.
Then rule \Refl\ applies.
}
\item{$m \geq 2$.
Let $I$ be the set of indices $i$ such that $\vec{z}_i = \vec{y}_1$ and $\aformB_i'$ contains at least 
one spatial atom.
We apply rule \hdec, with the decompositions $\aform = (x_1 \mapsto \vec{y}_1) * (x_2 \mapsto \vec{y}_2 * \dots * x_m \mapsto \vec{y}_m * \atform)$
and 
$\aformB_i = x_1 \mapsto \vec{z}_i * \aformB_i'$
and with
 the sets $\{ \{ i \} \mid i \in I \}$ and $I$.  
It is clear that these sets satisfy the application condition of the rule: indeed, for every $X \subseteq \interv{1}{n}$, 
either $X \cap I \not = \emptyset$ and then $\{ i \} \subseteq X$ for some $i \in I$, 
or $\interv{1}{n} \setminus X \supseteq I$. 
This yields the two premises:
 $x_1 \mapsto \vec{y}_1 \vdashr x_1 \mapsto \vec{y}_1$ (more precisely this premise is obtained $\card{I}$ times)
 and
 $\aform' \vdashr \aseq''$, 
 where $\aform' = x_2 \mapsto \vec{y}_2 * \dots * x_m \mapsto \vec{y}_m * \atform$
 and $\aseq''$ is the sequence of formulas $\aformB_i'$ for $i \in I$.
 We prove that these premises are all valid. This is straightforward for the former one.
 Let $(\astore,\aheap)$ be an injective model of $\aform'$.
 Since the $x_1,\dots,x_m$ are distinct, the heaps $\aheap$ and $\aheap' = \{ (\astore(x_1),\astore(\vec{y}_1)) \}$
 are disjoint, and we have 
 $(\astore,\aheap \dunion \aheap') \models \aform$,  thus 
 $(\astore,\aheap \dunion \aheap') \models x_1 \mapsto \vec{z}_j * \aformB_j'$, for some $j \in \interv{1}{m}$.
This entails that $\astore(\vec{z}_j) = \astore(\vec{y}_1)$ (hence $\vec{y}_1 = \vec{z}_j$ since $\astore$ is injective) 
and that
 $(\astore,\aheap) \models \aformB_j'$. Since $m \geq 2$ necessarily $\aheap\not = \emptyset$, thus
 $\aformB_j'$ contains a spatial atom. Since $\vec{y}_1 = \vec{z}_j$, this entails that $j \in I$, hence the proof is completed.}
 \end{compactitem}
\end{proof}

The calculus is 
a decision procedure for 
sequents in which the rules defining the left-hand side terminate, if
$\theory$ is closed under negation and some decision procedure exists for checking entailments between {\tformula}s:

\newcommand{\leftterminating}{left-terminating\xspace}


\begin{definition}
A sequent $\aform \vdashr \aseq$ is {\em \leftterminating}
iff for every predicate $p \in \preds$ such that $\aform \dependson{\asid} p$, we have $p \not \dependson{\asid} p$.
\end{definition}


\begin{lemma}
\label{lem:finite}
Every proof tree with a \leftterminating end-sequent is finite.
\end{lemma}
\begin{proof}
Let $\tau$ be a proof tree with a \leftterminating end-sequent $\aform \vdashr \aseq$.
By hypothesis, $\dependson{\asid}$ is an order on the set of
predicates $p$ such that $\aform \dependson{\asid} p$. 
Thus we may assume in the definition of the measure $\mes$ that the weight of any such predicate $p$ is strictly greater than the size of every formula $\aformB$ such that $p(\vec{x}) \Leftarrow \aformB$ is a rule in $\asid$.
Then every application of a rule \Unf\ on an atom of the form $p(\vec{y})$ with $\aform \dependson{\asid} p$ strictly decreases $\mes$.
Furthermore,  since the rule \Unf\ is the only rule that can add new predicate atoms on the left-hand side of the sequent,  it is easy to check 
that $\aform \dependson{\asid} q$ holds for every atom $q(\vec{y})$ occurring on the left-hand side of a sequent in $\tau$. 
Consequently, we deduce by Lemma \ref{lem:mes} that all the rules decrease $\mes$, 
which entails that $\tau$ is finite, since $\mes$ is well-founded.
\end{proof}

\begin{theorem}
\label{cor:finite}
If $\theory$ is closed under negation, then for every valid disjunction-free, \swfree, \alloccompatible  and \leftterminating sequent $\aform \vdashr \aseq$ there exists  a finite proof tree of end-sequent $\aform \vdashr \aseq$.
\end{theorem}
\begin{proof}
This is an immediate consequence of 
Theorem \ref{theo:comp}
and Lemma \ref{lem:finite}.
\end{proof}

\subsection{Entailments Without Theories}


In this section, we show that every valid, \constrained{\emptyset} sequent admits  a rational proof tree. This shows that the calculus is a decision procedure for \constrained{\emptyset} sequents, and also, using the reduction in Theorem \ref{theo:elimeq},
for \constrained{\{\iseq, \not \iseq \}} sequents.
We first introduce a notion of strong validity:


\newcommand{\svalid}{strongly valid\xspace}
\newcommand{\svalidity}{strong validity\xspace}
\newcommand{\amap}{\eta}

\begin{definition}
\label{def:svalid}
A sequent $\aform \vdashr \aseq$ is {\em \svalid} relatively to a decomposition $\aform = \aform_1 * \aform_2$ (modulo AC)
if, for every structure $(\astore,\aheap_1 \dunion \aheap_2)$ such that:
\begin{compactitem}
	\item $\astore$ is injective,
	\item $(\astore,\aheap_i) \modelsr \aform_i$  and
	\item $x\in \alloc{\aform_i}$ whenever $\astore(x) \in \dom{\aheap_i}$ for $i = 1,2$,
\end{compactitem} 
there exists a formula of the form $\exists \vec{x}. (\aformB_1 * \aformB_2)$ in $\aseq$
and a store $\astore'$ coinciding with $\astore$ on every variable not occurring in $\vec{x}$
such that 
$(\astore',\aheap_i) \modelsr \aformB_i$ for $i = 1,2$.
\end{definition}

Using Lemma \ref{lem:alloc_counter_model} below, it is possible to show that 
 if a \purelyspatial sequent is \svalid then it is also valid, but the converse does not hold.
The intuition is that in a \svalid sequent, the decomposition of the heap associated with the separating conjunction on the  
left-hand side
corresponds to a syntactic decomposition of a formula on the right-hand side. Note that the 
notion of \svalidity is relative to a decomposition $\aform = \aform_1 * \aform_2$ which will always be clear from the context.
For instance, the sequent
$\ls(x,y) * \ls(y,z) \vdashr \ls(x,z)$ is valid (where $\ls$ is defined as in Example \ref{ex:ls}), but not \svalid, because the decomposition of the left-hand side
does not correspond to any decomposition of the right-hand side.
On the other hand, $\ls(x,y) * \ls(y,z) \vdashr \exists u.~(\ls(x,u) * \ls(u,z))$ is \svalid.

\newcommand{\astoreB}{\hat{\astore}}
\begin{definition}
Let $\aheap$ be a heap.
For any mapping $\amap: \Loc \rightarrow \Loc$ that is injective on $\dom{\aheap}$, we denote by $\amap(\aheap)$ the heap
$\{ (\amap(\ell_0),\dots,\amap(\ell_\rank)) \mid (\ell_0,\dots,\ell_n) \in \aheap \}$.	
\end{definition}

\begin{proposition}
\label{prop:heap_morphism}
Let $(\astore,\aheap)$ be a structure and let $\amap: \Loc \rightarrow \Loc$ be a mapping that is 
injective on $\dom{\aheap}$.
If $(\astore,\aheap) \modelsr \aformB$ and $\aformB$ is \purelyspatial then
$(\amap \circ \astore,\amap(\aheap)) \modelsr \aformB$.
\end{proposition}
\begin{proof}
We prove the result by induction on the satisfiability relation $\modelsr$.
\begin{compactitem}
\item{
	Assume that $\aformB = \MWs{\atail}{\anatom}{\vec{u}}{\theta}$,  $\anatom \unfoldto{\asid}^+ \exists \vec{x}. (\aformC * \atail')$, 
and there exists a store $\astore'$ coinciding with $\astore$ on all variables not occurring in $\vec{x}$ and a substitution $\sigma$  such that
$\dom{\sigma} \subseteq \vec{x}\cap\fv{\atail'}$, 
 $\atail = \atail'\sigma$ and $(\astore',\aheap) \modelsr \aformC\sigma\theta$.
By the induction hypothesis we get
$(\amap \circ \astore',\amap(\aheap)) \modelsr \aformC\sigma\theta$.
But $\amap \circ \astore'$ coincides with $\amap \circ \astore$ on all variables not occurring in $\vec{x}$, 
thus
$(\amap \circ \astore,\amap(\aheap)) \modelsr \MWs{\atail}{\anatom}{\vec{u}}{\theta} = \aformB$.}
\item{If $\aformB = \emp$ then $\aheap = \emptyset$,  hence $\amap(\aheap) = \emptyset$
 and $(\amap \circ \astore,\amap(\aheap)) \modelsr \aformB$.}
 \item{If $\aformB = (x \mapsto (y_1,\dots,y_\rank)$ and $\aheap = \{ (\astore(x),\astore(y_1),\dots,\astore(y_\rank))\}$,
 then  $\amap(\aheap) = \{ (\amap(\astore(x)),\amap(\astore(y_1)),\dots,\amap(\astore(y_\rank)))\}$, so that
 $(\amap \circ \astore,\amap(\aheap)) \modelsr \aform$.}
 \item{Assume that $\aformB = \aform_1 * \aform_2$ and that there exist disjoint heaps $\aheap_1,\aheap_2$ such that
 $\aheap = \aheap_1 \dunion \aheap_2$ and $(\astore,\aheap_i) \modelsr \aform_i$.
 By the induction hypothesis, we deduce that $(\amap\circ\astore,\amap(\aheap_i)) \modelsr \aform_i$.
 It is clear that $\amap(\aheap) = \amap(\aheap_1) \dunion \amap(\aheap_2)$, hence
 $(\amap\circ\astore,\amap(\aheap)) \modelsr \aform_1 * \aform_2$.}
 \item{The proof is similar if $\aformB = \aform_1 \vee \aform_2$ and $(\astore,\aheap) \modelsr \aform_i$, for some $i = 1,2$.}
 \item{Assume that $\aformB = \exists x. \aformC$ and that there exists a store $\astore'$, coinciding with $\astore$ on all variables distinct from $x$, such that 
 $(\astore',\aheap) \modelsr \aformC$. By
the induction hypothesis, we have $(\amap \circ \astore',\amap(\aheap)) \modelsr \aformC$. 
Now, $\amap(\astore')$ coincides with $\amap(\astore)$ on all variables distinct from $x$, therefore
$(\amap \circ \astore,\amap(\aheap)) \modelsr \aformB$.}
\end{compactitem}
\end{proof}

\begin{lemma}
\label{lem:alloc_counter_model}
Let $\aform \vdashr \aseq$ be a \purelyspatial 
sequent, where $\aform$ is disjunction-free. 
If $\aform \vdashr \aseq$ admits a \countermodel $(\astore',\aheap')$, then there exists a structure $(\astore,\aheap)$ satisfying the following properties:
\begin{compactitem}
\item{$(\astore,\aheap)$ is a \countermodel of $\aform \vdashr \aseq$;}
\item{for all variables $x$: 
$x \not \in \alloc{\aform} \implies \astore(x) \not \in \dom{\aheap}$;}
\item{there exists an injective mapping $\amap$ such that $\astore' = \amap  \circ \astore$ and
 $\aheap' = \amap(\aheap)$.}
 \end{compactitem}
\end{lemma}
\begin{proof}
Since $\aform$ is \purelyspatial, we have 
$\aform \unfoldto{\asid}^* \exists \vec{x}. x_1 \mapsto \vec{y}_1 * \dots * x_n \mapsto \vec{y}_n$
and $\aheap' = \{ (\astore''(x_i),\astore''(\vec{y}_i)) \mid 1 \leq i \leq n \}$, for some store $\astore''$ coinciding with $\astore'$ on all variables not occurring in $\vec{x}$.
We assume, w.l.o.g., that $\Loc \setminus \img{\astore'}$ is infinite.  
Modulo $\alpha$-renaming, we may also assume that $\fv{\aform} \cap \vec{x}  = \emptyset$.
Since $\asid$ is \alloccompatible, we have 
\[\begin{array}{rcl}
	\alloc{\aform} & = & \alloc{\exists \vec{x}. x_1 \mapsto \vec{y}_1 * \dots * x_n \mapsto \vec{y}_n}\\
	& = & \{ x_1,\dots,x_n\} \cap \fv{\aform}\\
	& = & \{ x_1,\dots,x_n\} \setminus \vec{x}.
\end{array}\] 
Let $\astore$ be 
 a store 
mapping all variables 
to pairwise distinct locations not occurring in $\locs{\aheap} \cup \img{\astore''}$. Note that $\astore$ is injective by construction. 
Let $\hat{\astore}$ be the store mapping each variable $x\in \vec{x}$ to 
$\astore''(x)$ and coinciding with $\astore$ on all other variables. 
It is clear that $\hat{\astore}$ is injective, thus the heap $\aheap = \{ (\hat{\astore}(x_i),\hat{\astore}(\vec{y}_i)) \mid 1 \leq i \leq n \}$ is well-defined, and by construction, $(\hat{\astore},\aheap) \modelsr x_1 \mapsto \vec{y}_1 * \dots * x_n \mapsto \vec{y}_n$,
 so that $({\astore},\aheap) \modelsr \aform$.
   Consider the function $\amap$ that maps each location $\astore(x)$ to $\astore'(x)$ 
 and leaves all other locations unchanged.
 Note that $\amap$ is injective because both $\astore$ and $\astore'$ are injective, and
by definition, $\astore' = \amap \circ \astore$.
Let $x$ be a variable. If $x\in \vec{x}$ then $\hat{\astore}(x) = \astore''(x)$, hence
$(\amap \circ \hat{\astore})(x) = \amap(\astore''(x))$. By definition of $\astore$ we have $\astore''(x) \not \in \img{\astore}$, hence 
$\amap(\astore''(x)) = \astore''(x)$ and we deduce that $(\amap \circ \hat{\astore})(x) = \astore''(x)$.
If $x\not \in \vec{x}$ then $\hat{\astore}(x) = \astore(x)$ and $\astore''(x) = \astore'(x)$, thus again
$(\amap \circ \hat{\astore})(x) = \astore''(x)$.
We deduce that $\aheap' = \amap(\aheap)$.
If $(\astore,\aheap) \modelsr \aseq$ then  by Proposition \ref{prop:heap_morphism} we have
$(\astore',\aheap') \modelsr \aseq$, which contradicts the definition of $(\astore',\aheap')$.
Thus $(\astore,\aheap)$ is a \countermodel of $\aform \vdashr \aseq$, and by construction, for all variables $x$, 
$x \not \in \alloc{\aform} \implies \astore(x) \not \in \dom{\aheap}$.
\end{proof}

\begin{lemma}
\label{lem:rename_counter_example}
Let $\aform \vdashr \aseq$ be a \purelyspatial and non-valid sequent, where $\aform$ is disjunction-free.
Let $\astore$ be an injective store and  $\asetloc$ be an infinite set of locations such that $\asetloc \cap \img{\astore} = \emptyset$. 
Then $\aform \vdashr \aseq$ admits a \countermodel $(\astore,\aheap)$ with $\dom{\aheap} \subseteq \asetloc \cup \astore(\alloc{\aform})$.
\end{lemma}
\begin{proof}
By hypothesis, there exists an injective store $\astore'$ and a heap $\aheap'$ such that
$(\astore',\aheap') \modelsr \aform$ 
and
$(\astore',\aheap') \not \modelsr \aseq$. 
By Lemma \ref{lem:alloc_counter_model}, we may assume that 
for all variables $x\in \fv{\aform} \cup \fv{\aseq}$, if
$x \not \in \alloc{\aform}$ then $\astore'(x) \not \in \dom{\aheap'}$ ($\dagger$).
Let $\amap: \Loc \rightarrow \asetloc \cup \astore(\fv{\aform}\cup \fv{\aseq})$ be a bijective mapping 
such that $\amap(\astore'(x)) = \astore(x)$, for all variables $x\in \fv{\aform} \cup \fv{\aseq}$. 
Such a mapping necessarily exists. Indeed, since $\asetloc$ is infinite, there exists a bijection between $\Loc\setminus \astore'(\fv{\aform} \cup \fv{\aseq})$ and $\asetloc$; and since both $\astore$ and $\astore'$ are injective, there exists a bijection between $\astore'(\fv{\aform} \cup \fv{\aseq})$ and $\astore(\fv{\aform} \cup \fv{\aseq})$.
Let $\aheap = \amap(\aheap')$.
By construction, for every location $\ell \in \dom{\aheap}$, we have 
$\amap^{-1}(\ell) \in \dom{\aheap'}$, and if $\ell \not \in \asetloc$, then necessarily $\ell = \astore(x)$ and
$\amap^{-1}(\ell)  = \astore'(x)$, for some $x \in \fv{\aform}\cup \fv{\aseq}$.
Since $\astore'(x) \in \dom{\aheap'}$, by ($\dagger$) we have $x \in \alloc{\aform}$, and therefore $\dom{\aheap} \subseteq \asetloc \cup \astore(\alloc{\aform})$.

We now show that $(\astore,\aheap)$ is a \countermodel of $\aform \vdashr \aseq$.
Since 
$(\astore',\aheap') \modelsr \aform$, we deduce by Proposition \ref{prop:heap_morphism}
that $(\amap\circ\astore',\aheap) \modelsr \aform$. Moreover, since $\amap\circ\astore'$ and $\astore$ agree on all variables in $\fv{\aform}$ by construction, this entails that
$(\astore,\aheap) \modelsr \aform$.
Assume that  $(\astore,\aheap) \modelsr \aseq$.
Applying Proposition \ref{prop:heap_morphism} with $\amap^{-1}$, we get
$(\amap^{-1}\circ \astore,\amap^{-1}(\aheap)) \modelsr \aseq$.
But we have $\amap^{-1}(\aheap) = \amap^{-1}(\amap(\aheap')) = \aheap'$,
thus
$(\amap^{-1} \circ \astore,\aheap') \modelsr \aseq$.
Since $\amap^{-1} \circ \astore$ and $\astore'$ agree on all variables in $\fv{\aseq}$ by construction, we deduce that
$(\astore',\aheap') \modelsr \aseq$, contradicting our initial assumption.
Thus $(\astore,\aheap) \not \modelsr \aseq$
and 
$(\astore,\aheap)$ is a \countermodel of $\aform \vdashr \aseq$.
\end{proof}


\begin{definition}
A {\em path} from $\ell$ to $\ell'$ in  a heap $\aheap$
is a nonempty sequence of locations  $(\ell_1,\dots,\ell_n)$ such that 
$\ell_1 = \ell$, $\ell_n = \ell'$ and for every $i \in \interv{1}{n-1}$, $\ell_{i+1} \in \aheap(\ell_i)$.
\end{definition}

The following proposition states an important property of \pcSIDs.

\begin{proposition}
\label{prop:path}
Let $p(x_1,\dots,x_n)$ be a predicate atom and let 
$(\astore,\aheap)$ be a model of $p(x_1,\dots,x_n)$.
If $\ell = \astore(x_1)$ and $\ell' \in \locs{\aheap}$ then there exists a path from $\ell$ 
to $\ell'$ in $\aheap$.
\end{proposition}
\begin{proof}
The proof is by induction on $\card{\dom{\aheap}}$.
We have $p(x_1,\dots,x_n) \unfoldto{\asid} \exists \vec{x}. \aform$ for some quantifier-free
formula $\aform$, and there exists a store $\astore'$ coinciding with $\astore$ on all variables not occurring in $\vec{x}$, such that 
$(\astore',\aheap) \models \aform$.
If $\ell = \ell'$ then the result is immediate (a path of length $1$ always exists from any location $\ell$ to $\ell$). Assume that $\ell \not = \ell'$.
By the progress condition, $\aform$ is of the form $x_1 \mapsto (y_1,\dots,y_\rank) * p_1(z_1^1,\dots,z_{\ar{p_1}}^1) * \cdots * p_m(z_1^m,\dots,z_{\ar{p_m}}^m) * \atform$, for some \tformula 
$\atform$.
Thus there exist disjoint heaps $\aheap_1,\dots,\aheap_m$ such that
$\aheap = \{ (\astore(x_1),\astore'(y_1),\dots,\astore'(y_\rank)) \} \dunion \aheap_1 \dunion \dots \dunion \aheap_m$ and $(\astore,\aheap_i) \models  p_i(z_1^i,\dots,z_{\ar{p_i}}^i)$ for $i = 1,\dots,m$.
If $\astore(y_i) = \ell'$, for some $i = 1,\dots,\rank$ we have 
$\ell' \in \aheap(\astore(x_1)) = \aheap(\ell)$ and the proof is completed.
Otherwise, $\ell'$ cannot occur in $\{ \astore(x_1),\astore'(y_1),\dots,\astore'(y_\rank)\}$, hence $\ell'$ necessary occurs in $\locs{\aheap_i}$ for some $i = 1,\dots,m$, since $\ell' \in \locs{\aheap}$  by hypothesis.
By the induction hypothesis, there exists a path in $\aheap_i$ from $\astore'(z_1^i)$ to $\ell'$,
and by the connectivity condition $\astore'(z_1^i)\in (\astore'(y_1),\dots,\astore'(y_\rank)) = \aheap(\astore(x_1)) = \aheap(\ell)$.
Thus there exists a path in $\aheap$ from $\ell$ to $\ell'$.
\end{proof}

\begin{lemma}
\label{lem:svalid}
Let $\aform_1 * \aform_2 \vdashr \aseq$ be a valid prenex \purelyspatial sequent, 
where $\aform_i\not = \emp$ for $i = 1,2$, $\aform_1$ and $\aform_2$ are quantifier-free and disjunction-free. If the rule  \hsep\ is not applicable on 
$\aform_1 * \aform_2 \vdashr \aseq$ and $\aseq$ is in prenex form then $\aform_1 * \aform_2  \vdashr \aseq$ is \svalid.
\end{lemma}
\begin{proof}
Assume that there exists a structure $(\astore,\aheap_1 \dunion \aheap_2)$ such that $(\astore,\aheap_i) \modelsr \aform_i$ for $i = 1,2$
and $\forall x\in \fv{\aform_i}, \astore(x) \in \dom{\aheap_i} \implies x\in \alloc{\aform_i}$ ($\dagger$).
Since $\aform_1 * \aform_2 \vdashr \aseq$ is valid, and $\aseq$ is in prenex form, $\aseq$ contains a formula $\aformB = \exists\vec{x}. (\aformB_1 * \cdots * \aformB_n)$ such that
$(\astore,\aheap_1 \dunion \aheap_2) \modelsr \aformB$ and
each formula $\aformB_i$ ($i \in \interv{1}{n}$) is either a points-to atom or a \Watom. 
Thus there exist a store $\astore'$, coinciding with $\astore$ on all variables not occurring in $\vec{x}$ and disjoint heaps $\aheap_1',\dots,\aheap_n'$  such that
$(\astore',\aheap_i') \modelsr \aformB_i$ for $i = 1,\dots,n$
and
$\aheap_1 * \aheap_2 = \aheap_1' * \dots * \aheap_n'$. By $\alpha$-renaming, we assume that $\fv{\aform_1 * \aform_2} \cap  \vec{x} = \emptyset$, so that $\astore$ and $\astore'$ coincide on $\fv{\aform_1 * \aform_2}$.
Note that by hypothesis
$\aform_1$ and $\aform_2$ are separating conjunctions of atoms.

Assume that one of the heap $\aheap_i'$ for $i = 1,\dots,n$ is such that 
$\aheap_i' \not \subseteq \aheap_1$
and
$\aheap_i' \not \subseteq \aheap_2$.
This entails that $\aformB_i$ cannot be a points-to atom because
$\card{\dom{\aheap_i'}} = 1$ in this case, hence that $\aformB_i$ is a \Watom.
By Proposition \ref{prop:root_unsat}, 
$\astore'(\roots{\anatom_i}) \subseteq \dom{\aheap_i'} \subseteq \dom{\aheap_1} \cup\dom{\aheap_2}$.
We assume by symmetry that $\astore'(\roots{\anatom_i}) \subseteq \dom{\aheap_1}$, the other case is similar.
Let $\ell \in \dom{\aheap_i'}  \setminus \dom{\aheap_1}$.
By Proposition \ref{prop:path}, 
there exists 
a sequence of locations $\ell_1,\dots,\ell_m$ with $\{ \ell_1 \} = \astore'(\roots{\anatom_i})$, $\ell_m = \ell$
and $\ell_{i+1} \in \aheap_i'(\ell_i)$ for $i = 1,\dots,m$. 
We may assume, by considering the location $\ell$ associated with the minimal 
sequence ending outside of $\dom{\aheap_1}$, that $\ell_i \in \dom{\aheap_1}$ for all $i < m$.
This entails that $\ell \in \aheap_1(\ell_{m-1})$ so that $\ell \in \locs{\aheap_1} \setminus \dom{\aheap_1}$, and 
by Lemma \ref{lem:fr} we deduce that $\ell = \astore(x)$ 
for some $x\in \fv{\aform_1}$.
Furthermore, since $\astore(x) \in \dom{\aheap_2}$ we must have $x\in \alloc{\aform_2}$ by ($\dagger$).
Since the rule \hsep\ is not applicable, the variable $x$ must occur in $\roots{\aformB_1 * \dots * \aformB_n}$.
This entails that 
there exists $j$ such that $\astore(x) \in \dom{\aheap_j'}$ by Proposition \ref{prop:root_unsat}.
But $x$ cannot be the \mroot of $\aformB_i$, because $\aheap_1$ and $\aheap_2$ are disjoint, $\astore(x)\in \dom{\aheap_2}$ and $\astore'(\roots{\anatom_i} \subseteq \dom{\aheap_1}$, hence $i\neq j$. This contradicts
the fact that the $\aheap_1',\dots,\aheap_n'$ are disjoint.

Therefore, every heap $\aheap_i'$ is a subheap of either $\aheap_1$ or $\aheap_2$.
By regrouping the formulas $\aformB_i$ such that $\aheap_i'\subseteq \aheap_j$ in a formula $\aformB_j'$ (for $j = 1,2$), 
we get a decomposition of $\exists \vec{x}. (\aformB_1 * \dots * \aformB_n)$ of the form:
$\exists \vec{x}. (\aformB_1' * \aformB_2')$, where
$(\astore',\aheap_j) \modelsr \aformB_j'$ for $j = 1,2$.
Thus $\aform_1 * \aform_2  \vdashr \aseq$ is \svalid.
\end{proof}

\begin{lemma}
\label{lem:ED}
Let $\aform_1 * \aform_2 \vdashr \exists \vec{y}.\exists x. \aformC, \aseq$ be a \svalid, {\purelyspatial} and \qprenex sequent, 
with $\aform_i \not = \emp$ for $i = 1,2$.
There exists an application of the rule \ED\ with conclusion
$\aform_1 * \aform_2 \vdashr \exists \vec{y}. \exists x. \aformC, \aseq$, such that the premise is \svalid and \qprenex.
\end{lemma}
\begin{proof}
We apply the rule \ED\ where $\{ \aformC_1,\dots,\aformC_m \}$
is the set of all formulas satisfying the conditions of the rule,  the formulas $\aform$ and $\aform'$ in the application condition of \ED\ are prenex formulas and $\atform = \emp$. 
Recall from the rule definition that $\mnset{x_1, \ldots, x_n} = \fv{\aform}\cap \fv{\aform'}$. It is straightforward to verify that the premise is \qprenex. Note that both $\aformB$ and $\aformB'$ are {\purelyspatial} by hypothesis. 
Let $(\astore,\aheap_1 \dunion \aheap_2)$ be a structure such that for all $i = 1,2$,
 $(\astore,\aheap_i) \modelsr \aform_i$ and $\forall x\in \fv{\aform_i}, \astore(x) \in \dom{\aform_i} \implies x\in \alloc{\aform_i}$. 
 We have $\aheap_i \not = \emptyset$, because
 $\aform_i\not = \emp$.
 We show that the right-hand side of the premise contains a formula satisfying the conditions of Definition \ref{def:svalid}. 
 Since $\aform_1 * \aform_2 \vdashr \exists \vec{y}.\exists x. \aformC, \aseq$  is \svalid 
 there exist a formula of the form $\exists \vec{x}. (\aformB_1 * \aformB_2)$ in $\exists \vec{y}. \exists x. \aformC, \aseq$
and a store $\astore'$ coinciding with $\astore$ on every variable not occurring in $\vec{x}$
such that 
$(\astore',\aheap_i) \modelsr \aformB_i$ for $i = 1,2$.
Since the considered sequent is \qprenex, $\aformB_1$ and $\aformB_2$ are in prenex form. 
If $\exists \vec{x}. (\aformB_1 * \aformB_2)$ occurs in $\aseq$ then the proof is completed.
Otherwise we have $\vec{y}.x = \vec{x}$ and $\aformC = \aformB_1 * \aformB_2$.
If $\astore'(x) = \astore(x_j)$, for some $j = 1,\dots,n$, then
$(\astore',\aheap_i) \modelsr \repl{\aformB_i}{x}{x_j}$, and
$\repl{\aformB_1}{x}{x_j} * \repl{\aformB_2}{x}{x_j} = \repl{\aformC}{x}{x_j}$ occurs on the right-hand side of the premise, by definition of the rule \ED.
Thus, the result also holds in this case.
Now assume that $\astore'(x) \not = \astore(x_j)$, for all $j = 1,\dots,n$. 
By Lemma \ref{lem:fr}, 
since $(\astore,\aheap_i) \models \aform_i$, we have $\locs{\aheap_i} \setminus \dom{\aheap_i} \subseteq \astore(\fv{\aform_i})$, for $i = 1,2$. Since $\astore'(x) \not = \astore(x_j)$, for all $j = 1,\dots,n$, $\astore'(x) \not \in \astore(\fv{\aform_i})$ hence $\astore'(x)\not \in \locs{\aheap_i} \setminus \dom{\aheap_i}$.
If $\astore'(x) \in \locs{\aheap_1} \cap \locs{\aheap_2}$, then we would have 
$\astore'(x) \in \dom{\aheap_1} \cap \dom{\aheap_2}$, contradicting the fact that 
$\aheap_1$ and $\aheap_2$ are disjoint. This entails that $\astore'(x) \not \in \locs{\aheap_1} \cap \locs{\aheap_2}$.
 We assume, by symmetry, that $\astore'(x) \not \in \locs{\aheap_2}$. Then by Lemma \ref{lem:outsideheap}, we deduce that
$(\astore'',\aheap_2) \modelsr \aformB_2$, for every 
store $\astore''$ coinciding with $\astore'$ on all variables distinct from $x$.
We also have $(\astore'',\aheap_1) \modelsr \exists x. \aformB_1$, since
 $(\astore',\aheap_1) \modelsr \aformB_1$.
 By definition of the rule, the right-hand side of the premise contains a formula $\exists \vec{y}. ((\exists x. \aformB_1) * \repl{\aformB_2}{x}{x'})$, hence the conditions of Definition \ref{def:svalid} are fulfilled 
 (note that $\aformB_2$ cannot be $\emp$, since $\aheap_2 \not = \emptyset$).
\end{proof}

  
\begin{lemma}
\label{lem:hdec}
Let $\aform_1 * \aform_2 \vdashr \aseq$ be a {\svalid} and \purelyspatial sequent.
If  rules \ED\ and \HC\ do not apply on $\aform_1 * \aform_2 \vdashr \aseq$ 
then rule \hdec\ applies on $\aform_1 * \aform_2 \vdashr \aseq$, with a 
valid  premise.
\end{lemma}
\begin{proof}
Let $\aseqB$ be the subsequence of formulas in $\aseq$ that are not separated conjunctions, 
so that $\aseq$ is of the form $\aformB_1^1 * \aformB_1^2,\dots,\aformB_n^1 * \aformB_n^2, \aseqB$.
By definition of the notion of \svalidity, if $(\astore,\aheap_1 \dunion \aheap_2) \modelsr \aform_1 * \aform_2$ and $\forall x\, (\astore(x) \in \dom{\aheap_i} \implies x \in \alloc{\aform_i})$, then $\aseq$ contains a formula of the form $\exists \vec{x}. (\aformB_1 * \aformB_2)$ such that
$(\astore',\aheap_1) \models \aformB_1$ 
and
$(\astore',\aheap_2) \models \aformB_2$, where $\astore'$ coincides with $\astore$ on all 
variables not occurring in $\vec{x}$.
By Lemma \ref{lem:ED}, none of the formulas in $\aseq$ may be existentially quantified, since otherwise rule \ED\ applies. Thus $\vec{x}$ is empty,  
$\exists \vec{x}. (\aformB_1 * \aformB_2) = \aformB_1 * \aformB_2$ does not occur in $\aseqB$ and there exists $i =1,\dots,n$ such that
$\aformB_1 = \aformB_i^1$ and 
$\aformB_2 = \aformB_i^2$. 

Let $I_1,\dots,I_m$ denote the inclusion-minimal subsets of $\interv{1}{n}$ such that for all $j =1, \ldots, m$, $\aform_1 \vdashr \bigvee_{c\in I_j} \aformB_c^1$ is valid. Similarly, let
$J_1,\dots,J_l$ denote the inclusion-minimal subsets of $\interv{1}{n}$ such that for all $j = 1, \ldots, l$, 
$\aform_2 \vdashr \bigvee_{c\in J_j} \aformB_c^2$ is valid.
If the rule \hdec\ applies with such sets, then, by construction all the premises are valid, hence the proof is completed.
Otherwise, by the application condition of the rule, there exists a set $X$ such that
$X \not \supseteq I_i$ for all $i = 1,\dots,n$
and
$\interv{1}{n} \setminus X \not \supseteq  J_j$ for all $j = 1,\dots,l$.
This entails that $\aform_1 \vdashr \bigvee_{c\in X} \aformB_c^1$ is not valid, since otherwise $X$ would contain one of the $I_i$ which are inclusion-minimal by construction. Similarly,
$\aform_2 \vdashr \bigvee_{c\in \interv{1}{n} \setminus X} \aformB_c^2$ is not valid because otherwise $\interv{1}{n} \setminus  X$ would contain one of the $J_j$.
Therefore, there exist injective stores $\astore_i$ and heaps $\aheap_i$ for $i = 1,2$, such that
$(\astore_i,\aheap_i) \modelsr \aform_i$,
$(\astore_1,\aheap_1) \not \modelsr \bigvee_{c\in X} \aformB_c^1$
and
$(\astore_2,\aheap_2) \not \modelsr \bigvee_{c\in \interv{1}{n} \setminus X} \aformB_c^2$.
Let $\astore$ be an injective store such that $\Loc \setminus \img{\astore}$ is infinite, and let $\asetloc_1,\asetloc_2$ be disjoint infinite subsets of $\Loc\setminus \img{\astore}$. 
By applying Lemma \ref{lem:rename_counter_example} twice on $\aheap_1,\aheap_2$ with $\asetloc_1$ and $\asetloc_2$ respectively, we obtain
two heaps
$\aheap_1'$ and $\aheap_2'$ such that:
\begin{compactitem}
	\item $\dom{\aheap_i'} \subseteq  L_i \cup \astore(\alloc{\aform_i})$ for $i = 1,2$, 
	\item $(\astore,\aheap_i') \modelsr \aform_i$ for $i = 1, 2$,
	\item $(\astore,\aheap_1') \not \modelsr \bigvee_{c\in X} \aformB_c^1$,
	\item $(\astore,\aheap_2') \not \modelsr \bigvee_{c\in \interv{1}{n} \setminus X} \aformB_c^2$.
\end{compactitem}
In particular, since $L_1$ and $L_2$ are disjoint, we have $\dom{\aheap_1'} \cap \dom{\aheap_2'} \subseteq \astore(\alloc{\aform_1}) \cap \astore(\alloc{\aform_2})$.
Assume that $\dom{\aheap_1'} \cap \dom{\aheap_2'}$ is nonempty, and contains a location $\ell$.
Then for $i=1,2$, there exists a variable $x_i \in \fv{\aform_i}$ such that $\ell = \astore(x_i)$. Since $\astore$ is injective, necessarily $x_1 = x_2$, and $\alloc{\aform_1} \cap \alloc{\aform_2} \not = \emptyset$; this entails that
the rule \HC\ is applicable, which contradicts the hypothesis of the lemma.

We deduce that $\aheap_1'$ and $\aheap_2'$ are disjoint. Then we have
$(\astore,\aheap_1' \dunion \aheap_2') \modelsr \aform_1 * \aform_2$
and
for all $c \in \interv{1}{n}$, either $(\astore,\aheap_1') \not \modelsr \aformB_c^1$ or
$(\astore,\aheap_2') \not \modelsr \aformB_c^2$. 
Furthermore, if $\astore(x)\in \dom{\aheap_i'}$, then 
$\astore(x)\in \astore(\alloc{\aform_i})$, thus
$x\in \alloc{\aform_i}$ because $\astore$ is injective.
This contradicts that fact that $\aform_1 * \aform_2 \vdashr \aseq$ is \svalid.

\end{proof}


 \newcommand{\vinit}{N_{\mathit init}}
  \newcommand{\vsid}{N_{\asid}}
 \newcommand{\vex}{N_{\exists}}
 \newcommand{\vmax}{N}

  \newcommand{\maxarity}{A}
  \newcommand{\maxlength}{P}
  \newcommand{\maxvar}{V}
  
  \newcommand{\nbPreds}{\card{\preds}}

%
%

\newcommand{\vared}{variable-redundant\xspace}
\newcommand{\rootred}{root-redundant\xspace}


\begin{definition}
	Let $\aform \vdashr \aformB, \aseq$ be a sequent.
	The formula $\aformB$ is {\em \rootred} if there exists a variable $x$ such that either $x \in \rootsr{\aformB}$ 
	and $x \not \in \alloc{\aform}$, or $x \in \rootsl{\aformB}$ and $x \not \in \fv{\aform}$.
	This formula is {\em \vared} if there exist two 
	injective substitutions $\sigma,\theta$ such that 
	$\aformB\sigma \in \aseq\theta$ and $(\dom{\sigma} \cup \dom{\theta}) \cap \fv{\aform} = \emptyset$.
\end{definition}

\begin{proposition}
\label{prop:rootsl}
If $(\astore,\aheap) \models \aform$ then $\astore(\rootsl{\aform}) \subseteq \locs{\aheap}$.
\end{proposition}
\begin{proof}
The proof is by induction on the satisfiability relation. We only handle the case where 
$\aform = \MWs{\atail}{\anatom}{\vec{u}}{\theta}$, the other cases are straightforward.
By definition, there exists a formula $\exists \vec{x}. (\atail' * \aformB)$, a substitution $\sigma$ with $\dom{\theta} \subseteq \vec{x}$ and a store $\astore'$
coinciding with $\astore \circ \sigma$ on all variables not occurring in $\vec{x}$ such that
$\anatom \unfoldto{\asid}^+ \exists \vec{x}. (\atail' * \aformB)$, 
$\atail'\sigma = \atail$
and $(\astore',\aheap)\models \aformB\sigma\theta$.
Let $x\in \rootsl{\aform}$. By definition, $\atail$ contains a predicate atom of the form $p(x,\vec{y})$, thus
$\atail'$ contains a predicate atom of the form $p(x',\vec{y}')$, with $x'\sigma= x$.
This predicate atom must be introduced in the derivation 
$\anatom \unfoldto{\asid}^+ \exists \vec{x}. (\atail' * \aformB)$ by  unfolding some predicate atom 
$q(u,\vec{v})$. By the progress and connectivity condition this entails that 
$\aformB$ contains a points-to atom of the form $u \mapsto (\dots,x',\dots)$.
Since $(\astore',\aheap)\models \aformB\sigma$ 
this entails that $\astore'(x'\sigma) \in \locs{\aheap}$, thus $\astore(x)\in \locs{\aheap}$.
\end{proof}

\begin{proposition}
\label{prop:subst_store}
If  $(\astore \circ \sigma,\aheap) \modelsr \aform$ then
$(\astore,\aheap) \modelsr \aform\sigma$.
\end{proposition}
\begin{proof}
The proof is by induction on the satisfiability relation. We only handle the case where 
$\aform = \MWs{\atail}{\anatom}{\vec{u}}{\theta}$. 
By definition, there exists a formula $\exists \vec{x}. (\atail' * \aformB)$, a substitution $\theta$ with $\dom{\theta} \subseteq \vec{x}$ and a store $\astore'$
coinciding with $\astore \circ \sigma$ on all variables not occurring in $\vec{x}$ such that
$\anatom \unfoldto{\asid}^+ \exists \vec{x}. (\atail' * \aformB)$, 
$\atail'\theta = \atail$
and $(\astore',\aheap)\models \aformB\theta\eta$.
By $\alpha$-renaming, we assume that $\vec{x} \cap (\dom{\sigma} \cup \img{\sigma}) = \emptyset$.
Let $\hat{\astore}$ be the store such that $\hat{\astore}(x) = \astore(x)$ if $x\not \in \vec{x}$ and 
$\hat{\astore}(x) = \astore'(x)$ if $x\in \vec{x}$.
By construction we have $\astore' = \hat{\astore} \circ \sigma$.
By the induction hypothesis, we deduce that 
and $(\hat{\astore},\aheap)\models \aformB\theta\sigma$.
Furthermore, $\atail'\theta\sigma = \atail\sigma$, and 
$\anatom\sigma \unfoldto{\asid}^+ \exists \vec{x}. (\atail' * \aformB)\sigma = 
\exists \vec{x}. (\atail'\sigma * \aformB\sigma)$.
 Thus  $(\astore,\aheap) \modelsr \aform\sigma$.
\end{proof}

\begin{lemma}
\label{lem:rootred}
Let $\aform \vdashr \aformB, \aseq$ be a valid \purelyspatial disjunction-free sequent.
If $\aformB$ is \rootunsat, \rootred or \vared,  then $\aform \vdashr \aseq$ is valid.
\end{lemma}
\begin{proof}
We consider the three cases separately.
\begin{compactitem}
\item{
If $\aformB$ is \rootunsat, then $\aformB$ is unsatisfiable, thus 
every model of $\aformB,\aseq$ is also a model of $\aseq$. Therefore $\aform \vdashr \aseq$ if and only if $\aform \vdashr \aformB, \aseq$.}
\item{
Assume that $\aformB$ is \rootred and let $(\astore,\aheap)$ be an injective model of $\aform$. By  
Lemma \ref{lem:alloc_counter_model}, applied on any sequent
$\aform \vdashr \aform'$ such that $\aform'$ is unsatisfiable, there exists an injective 
mapping $\amap$ 
and a model $(\astore',\aheap')$ of $\aform$ 
such that 
$\astore = \amap \circ \astore'$, $\aheap = \amap(\aheap')$
and $\astore'(x) \in \dom{\aheap'} \implies x \in \alloc{\aform}$.
Let $\astore''$ be a store coinciding with $\astore'$ on all the variables in $\fv{\aform}$ and such that
$\astore''(y) \not \in \locs{\aheap'}$, 
for all $y \notin \fv{\aform}$. 
Since $\astore''$ and $\astore'$ coincide on $\fv{\aform}$, necessarily $(\astore'',\aheap') \models
\aform$. Also, if $\astore''(x) \in \dom{\aheap'} \subseteq \locs{\aheap'}$, then necessarily $x\in \fv{\aform}$, so that $\astore''(x) = \astore'(x)$ and $x\in \alloc{\aform}$. Therefore,  $\astore''(x) \in \dom{\aheap'} \implies x \in \alloc{\aform}$.
We distinguish two cases.
\begin{compactitem}
\item{Assume that $(\astore'',\aheap') \not \models \aseq$.
Since  $\aform \vdashr \aformB, \aseq$ is valid we have 
$(\astore'',\aheap') \models \aformB,\aseq$ and necessarily, 
 $(\astore'',\aheap') \models \aformB$.
If there exists $x \in \rootsr{\aformB}$ such that $x \not \in \alloc{\aform}$, then 
by Proposition \ref{prop:root_unsat} we have $\astore''(x) \in \dom{\aheap'}$, which contradicts the above implication.
Otherwise, since $\aformB$ is \rootred, there exists a variable $x\in \rootsl{\aformB} \setminus \fv{\aform}$.
Since $x\in \rootsl{\aformB}$, 
we have $\astore''(x) \in \locs{\aheap'}$ by Proposition \ref{prop:rootsl}, 
which contradicts the definition of $\astore''$ because $x\not \in \fv{\aform}$.}
\item{Assume that $(\astore'',\aheap') \models  \aseq$. 
Let $\amap'$ be the function mapping all locations of the form $\astore''(y)$ to $\astore'(y)$ and leaving all 
other locations unchanged.
By definition, we have $\astore = \amap \circ \amap' \circ \astore'$ and 
for all $y \in \vars$, if $\astore''(y) \in \locs{\aheap'}$ then we must have $y\in \fv{\aform}$, so that $\astore''(y) = \astore'(y)$. We deduce that $\amap'$ is the identity on every location in $\locs{\aheap'}$ and that $\amap(\amap'(\aheap')) = \aheap$.
By Proposition \ref{prop:heap_morphism},  $(\amap \circ \amap' \circ \astore',\amap(\amap'(\aheap))) \models  \aseq$, 
i.e.,
 $(\astore,\aheap) \models  \aseq$, contradicting our assumption.}
 \end{compactitem}}
 \item{Assume that $\aformB$ is \vared and let $(\astore,\aheap)$ be a model of $\aform$, where $\astore$ is injective. We show that
 $(\astore,\aheap) \modelsr \aseq$.
   Let $\astore'$ be an injective store coinciding with $\astore$ on every variable in $\fv{\aform}$ and such that
 $\astore'(x) \not \in \locs{\aheap}$ for every $x \not \in \fv{\aform}$ ($\ddagger$).
 It is clear that $(\astore',\aheap) \modelsr \aform$, thus
 $(\astore',\aheap) \modelsr \aformB, \aseq$, since $\aform \vdashr \aformB, \aseq$  is valid by hypothesis.
 If 
 $(\astore',\aheap) \modelsr \aseq$ then 
by Lemma \ref{lem:outsideheap}  we deduce that
 $(\astore,\aheap) \modelsr \aseq$, and the proof is completed.
 Otherwise, $(\astore',\aheap) \modelsr \aformB$.
 By hypothesis, there exist 
 injective substitutions $\sigma$ and $\theta$ such that 
 $\aformB\sigma \in \aseq\theta$ and $(\dom{\sigma} \cup \dom{\theta}) \cap \fv{\aform} = \emptyset$.
By ($\ddagger$), we have 
$(\astore' \circ \sigma)(x) = \astore'(x)$ for all $x$ such that $\astore'(x) \in \locs{\aheap}$, since in this case $x$ cannot be in $\dom{\sigma}$.
By Lemma \ref{lem:outsideheap}, we deduce that 
  $(\astore' \circ \sigma,\aheap) \modelsr \aformB$, thus 
    $(\astore',\aheap) \modelsr \aformB\sigma$  by Proposition \ref{prop:subst_store}. 
  Since 
  $\aformB\sigma \in \aseq\theta$, we deduce  that 
  $(\astore',\aheap) \modelsr \aseq\theta$, hence (as $\theta$ is injective)
  $(\astore' \circ \theta,\aheap) \modelsr \aseq$.
  If $x$ is a variable such that $(\astore' \circ \theta)(x) \in \locs{\aheap}$ then by definition of $\astore'$, we have $\theta(x) \in \fv{\aform}$. 
  Since $\dom{\theta} \cap \fv{\aform} = \emptyset$, we deduce that $\theta(x) = x$, and $(\astore' \circ \theta)(x) = \astore'(x) = \astore(x)$.
    By Lemma \ref{lem:outsideheap}, we conclude that $(\astore,\aheap) \modelsr \aseq$.
    }
 
 \end{compactitem}
\end{proof}

\newcommand{\asize}{N}
 
\begin{theorem}(Termination)
\label{theo:comp_term}
For all valid, \purelyspatial, prenex, disjoint-free, \alloccompatible and \swfree sequents  $\aform \vdashr \aseq$, there exists a rational proof tree with end-sequent $\aform \vdashr \aseq$. Furthermore,  the size of the proof tree is at most $\bigO(2^{2^{c.\asize^3}})$, where $c$ is a constant and $\asize = \widt{\aform \vdashr \aseq}$.
\end{theorem}
\begin{proof}
    Let $\maxarity = \max \{ \rank+1, \ar{p} \mid p \in \preds \}$ 
and  $\maxlength = \max \{ \len{p} \mid p \in \preds \}$.
We assume, w.l.o.g., that all the predicate in $\preds$ occur in $\asid$, so that 
$\maxarity = \bigO(\asize)$
and
$\maxlength = \bigO(\asize)$. 
The proof tree is constructed in a similar way as in the proof of Theorem \ref{theo:comp}, except 
that  rule \Unf\ will be applied  only when the left-hand side of the sequent contains a unique spatial atom. 
We show that there is some  rule that is applicable to every valid sequent, in such a way that 
all the premises are valid.
All the inference rules that are axioms are applied with the highest priority, whenever possible.
Afterwards, the rule \Wk\ is applied to remove from the right-hand side of the sequents
all the formulas that are \rootunsat, \rootred or \vared; Lemma \ref{lem:rootred} guarantees that the validity the sequent is preserved.
Rules \Sk\ and {\hsep} are then applied as much as possible.
If the process terminates, then we eventually obtain a sequent  of the form $\aformB \vdashr \aseqB$, such that $\aformB$ is quantifier-free (by irreducibility w.r.t.\ \Sk) and 
$\alloc{\aformB} \subseteq \roots{\aformB'}$, for every $\aformB'$ in $\aseqB$ (by irreducibility w.r.t.\ $\hsep$). 

We now distinguish several cases, depending on the form of $\aformB$.
First if $\aformB$ is 
a points-to atom, then by Theorem \ref{theo:comp}, there exists a rule application yielding valid premises, and the applied rule cannot be \Unf, since by hypothesis $\aformB$ contains no predicate atom.
 Otherwise, if $\aformB$ is a single predicate atom then we apply the rule \Unf\footnote{Note that in the proof of Theorem \ref{theo:comp}
	rule {\Unf} is applied when $\aformB$ \emph{contains} a predicate atom. However, this strategy is not applicable here because it
	may produce an infinite proof tree.} 
	and validity is preserved thanks to Lemma \ref{lem:inv}. 
Otherwise, $\aformB$ must be a separating conjunction. By Lemma \ref{lem:svalid}, $\aformB \vdashr \aseqB$ is \svalid.
Rule {\ED} is then applied as much as possible.
Note that, by Lemma \ref{lem:ED}, 
if the prefix of a formula in $\aseqB$ contains an existential variable then there exists an application of $\ED$ such that the obtained sequents are 
\svalid and \qprenex.  
If \ED\ is not applicable then by Lemma \ref{lem:hdec}, rule \hdec\ necessarily applies. 

We now prove that the constructed proof tree is rational.
To this purpose, we analyze the form of the sequents occurring in it.
Consider any 
sequent $\aformC \vdashr \aformC_1,\dots,\aformC_n$ occurring in the proof tree and assume that none of the formulas $\aformC_1,\dots,\aformC_n$ is \rootunsat, \rootred 
or \vared. 
We prove the following invariant. 
\begin{invariant}
\label{inv}
For every $i = 1,\dots,n$, the following properties hold:
\begin{compactenum}
\item{If an existential variable is the \mroot of a \Watom in $\aformC_i$, then
it is the \mroot of an atom occurring in the initial sequent $\aform \vdashr \aseq$.
Moreover, existential variables cannot be {\aroot}s.
\label{inv:3}}

\item{Both $\rootsr{\aformC_i}$ and $\rootsl{\aformC_i}$ are sets (i.e., contain at most one occurrence of each variable), and 
	$\roots{\aformC_i} \subseteq \fv{\aformC}$.\label{inv:1} 
}
\item{If $x\in \rootsl{\aformC_i}$ and $x\not \in \rootsr{\aformC_i}$ then $x\not \in \alloc{\aformC}$. \label{inv:2}}

\item{The number of variables occurring in $\aformC$ is at most $\max(\vinit,\vsid)$, where $\vinit = \card{\fv{\aform} \cup \fv{\aseq}}$ and $\vsid$ is the maximal number of 
	 free or bound variables occurring in a  rule in $\asid$. \label{inv:4}}
\end{compactenum}
\end{invariant}
\begin{proof}
We assume, w.l.o.g., that $\aformC \vdashr \aformC_1,\dots,\aformC_n$ is the first sequent not satisfying these properties, along some (possibly infinite) path from the root.
\begin{compactenum}
	\item{It is straightforward to check, by inspection of the rules, that no rule can add atoms with existentially quantified roots in the premise: the only rule that can add new atoms to the right-hand side of a sequent 
is \hsep\ (by applying the function $\splitv{}{}$ defined in Section \ref{sect:split}),
and the roots of these atoms must be free variables. Then the proof follows from the fact that the initial sequent is \swfree. 
Note that no rule can rename the existential variables occurring in the sequents; the only variables that are renamed are those occurring in the inductive rules.
}

\item{
 The inclusion $\roots{\aformC_i} \subseteq \fv{\aformC}$ follows from the fact that the $\aformC_i$ are not \rootred.
If $\rootsr{\aformC_i}$ contains two occurrences of the same variable for some $i\in \interv{1}{n}$, then $\aformC_i$ is \rootunsat, contradicting our assumption.
Assume that a multiset $\rootsl{\aformC_i}$ contains two occurrences of the same variable $x$.
The initial sequent 
is \swfree, hence contains no \aroot, and the only rule that can add new {\aroot}s to a sequent
is {\hsep}. Thus assume  that $\aformC \vdashr \aformC_1,\dots,\aformC_n$ is obtained from 
a sequent of the form $\aformC \vdashr \aformC_1',\aseq'$, by applying rule \hsep\ on $\aformC_1'$ with  variable $x$. By the application condition of the rule, necessarily $x \in \alloc{\aformC}$. Furthermore, since by definition 
of the splitting operation, \hsep\ introduces exactly one \aroot in $\aformC_i$,
and  $x$ occurs twice in  $\rootsl{\aformC_i}$,
necessarily $x\in \rootsl{\aformC_1'}$. By Property \ref{inv:2} of the invariant, applied to the sequent $\aformC \vdashr \aformC_1',\aseq'$ 
which satisfies the invariant by hypothesis, since $x \in \alloc{\aformC}$, necessarily  
$x \in \rootsr{\aformC_1'}$.
But \hsep\ also introduces an atom with \mroot $x$, which entails that $x$ occurs twice in $\rootsr{\aformC_i}$, hence that $\aformC_i$ is \rootunsat, yielding a contradiction. 
}
\item{
	The only rules that can add affect the roots of the right-hand side of the sequent are \ED, \hsep,  and \hdec.
	\begin{compactitem}

		\item \ED\ may generate new roots 
		 by instantiating an existential variable with a free variable, however by Property \ref{inv:3}, existential variables
		cannot be {\aroot}s, hence Property \ref{inv:2} is preserved. 
		\item \hsep\ adds a new \aroot $x$ to the right-hand side of the sequent.
		However,
		for each such atom, \hsep\ also 
		adds an atom with \mroot $x$, hence the property is preserved.
		\item \hdec\ may remove {\mroot}s from the right-hand side of the sequent. If some premise of  \hdec\ does not fulfill Property \ref{inv:2} of the invariant, then, using the notations of the rule,
		there exists an index $i \in \interv{1}{n}$ such that  either 
		$x \in \rootsl{\aformB_i}$, 
		$x \not \in \rootsr{\aformB_i}$ 
		and $x \in \alloc{\aform}$, or  $x \in \rootsl{\aformB_i'}$, $x \not \in \rootsr{\aformB_i'}$ 
		and $x \in \alloc{\aform'}$. 
		We assume by symmetry that the former assertion holds.
		Since $\alloc{\aform} \subseteq \alloc{\aform * \aform'}$, we have 
		$x\in \alloc{\aform * \aform'}$, which entails that 
		$x \in \rootsr{\aformB_i'}$, since the conclusion satisfies Property \ref{inv:2}. 
		By irreducibility w.r.t.\ \HC,  
		$\alloc{\aform} \cap \alloc{\aform'} = \emptyset$, thus
		$x \not \in \alloc{\aform'}$.
		This entails that the formula $\aformB_i'$ is \rootred in all sequents with left-hand side $\aform'$, 
		 contradicting our assumption. 
	\end{compactitem}
 }

\item{Property \ref{inv:4} stems from the fact that no rule may add variables to the left-hand side of a sequent, 
 except for
\Unf. However, \Unf\ is always applied on a sequent with a left-hand side that is an atom, which entails that the number of variables occurring on the left-hand side after any application of the rule is bounded by $\vsid$.}
\end{compactenum}
\end{proof}

Properties \ref{inv:1}, \ref{inv:3} and \ref{inv:4} in Invariant \ref{inv} entail that the number of roots 
in every formula $\aformC_i$ is at most $2\cdot\max(\vinit,\vsid) + \vex$, where $\vex$ denotes the number of  existential variables in the end-sequent. Indeed, a free variable may occur at most twice as a root,  
once as an \aroot and once as a \mroot, and an existential variable may occur at most once as a \mroot.
Note that $\vinit = \bigO(\asize)$, $\vsid = \bigO(\asize)$ and $\vex = \bigO(\asize)$.
Thus the number of (free or existential) variables in $\aformC_i$ is bounded by $\maxarity \cdot (2\cdot\max(\vinit,\vsid) + \vex)$.
Hence we may assume (up to a renaming of variables) that
the total number of variables occurring in the considered sequent 
at most $\maxarity \cdot (2\cdot\max(\vinit,\vsid) + \vex)$
so that every such variable may be represented by a word of length $\ln( \maxarity \cdot (2\cdot\max(\vinit,\vsid) + \vex))$. 
Let $\aformC_1',\dots,\aformC_n'$ be formulas obtained from $\aformC_1,\dots,\aformC_n$
by replacing all the free variables not occurring in $\fv{\aformC}$ by some unique fixed variable $u$.
The size of every expression of the form $p(\vec{x})$ occurring in $\aformC_1',\dots,\aformC_n'$ (possibly within a \Watom) is bounded by $\maxlength + \maxarity \cdot \ln(\maxarity \cdot (2\cdot\max(\vinit,\vsid) + \vex)) = \bigO(\asize^2)$, thus 
the size of the formulas $\aformC_i'$ is bounded by 
$\bigO(\asize^3)$.
By Proposition \ref{prop:size_set}, 
the size of the set $\{ \aformC_1',\dots,\aformC_n' \}$ 
is therefore at most $\bigO(2^{d\cdot\asize^3})$ 
for some constant $d$.
If the size of the sequence $\aformC_1,\dots,\aformC_n$ 
is greater than $\bigO(2^{d\cdot\asize^3})$, then  there must exist 
  distinct formulas $\aformC_i$ and $\aformC_j$ such that $\aformC_i' = \aformC_j'$, i.e., $\aformC_i$ and $\aformC_j$ 
 are identical up to the replacement of variables not occurring in $\fv{\aformC}$. But then 
$\aformC_i$ and $\aformC_j$ would be \vared, which contradicts our assumption.
Therefore, 
$\Sigma_{i=1}^n \size{\aformC_i} = \bigO(2^{d\cdot\asize^3})$.

Necessarily $\size{\aformC} = \bigO(\asize)$, because no rule can increase the size of the left-hand side of the sequents above $\asize$: indeed, the only rule that can add new atoms to the left-hand side is \Unf, and 
this rule is applied only on single atoms, which entails that the  left-hand side of the obtained sequent is of the same size as the right-hand side of one of the  rules in $\asid$.

We deduce that 
the size of the sequents occurring in the proof tree
is at most $\bigO(2^{d\cdot\asize^3})$ for some constant $d$, and 
by Proposition \ref{prop:size_card}, 
there are at most 
 $\bigO(2^{2^{c\cdot\asize^3}})$ distinct sequents (for some constant $c$, up to a renaming of variables).
This entails that the constructed proof tree is rational.
\end{proof}

\section{Discussion}


Due to the high expressive power of \pcSIDs, the conditions ensuring termination (Theorems \ref{cor:finite} and \ref{theo:comp_term}) are necessarily restrictive. 
In the light of the undecidability result in \cite{EP22b}, one cannot hope for more. 
Theorem \ref{theo:comp} shows that, even if the conditions are not satisfied,  the calculus is still useful 
as a semi-decision procedure to detect non-validity.
The complexity of the procedure for \purelyspatial formulas cannot be improved significantly 
since entailment checking is \twoexptime-hard \cite{DBLP:conf/lpar/EchenimIP20}.
The high complexity of entailment testing is unsatisfactory in practice. We thus plan to investigate fragments of inductive rules for which the devised calculus yields an efficient decision procedure.
As emphasized by the lower-bound result in \cite{DBLP:conf/lpar/EchenimIP20}, which  relies on very simple data structures and by the fact that language inclusion is already \exptime-complete for tree automata \cite{DBLP:journals/ipl/Seidl94}, there is no hope that this can be achieved by restricting only the shape of 
the structures: it is also necessary to strongly restrict the class of inductive rules (forbidding  for instance overlapping rules).
 We will also try to identify classes of entailments for which the procedure terminates for nonempty theories, under reasonable 
 conditions on the theory. 
 Another problem that can be of  practical interest is to extract {\countermodel}s of 
 non-valid (or irreducible) entailments.
 Finally, we plan to extend the proof procedure in order to solve bi-abduction problems,  
 a generalized form of abduction which 
 plays a central r\^ole for the efficiency and scalability of program analysis algorithms \cite{DBLP:journals/jacm/CalcagnoDOY11}. 
  
  \subsubsection*{Acknowledgments.}
  
  This work has been partially funded by the 
the French National Research Agency ({\tt ANR-21-CE48-0011}).
  The authors wish to thank Radu Iosif for his comments on the paper and for fruitful discussions.


\begin{thebibliography}{10}

\bibitem{berdine-calcagno-ohearn04}
J.~Berdine, C.~Calcagno, and P.~W. O'Hearn.
\newblock A decidable fragment of separation logic.
\newblock In {\em Proc. of FSTTCS'04}, volume 3328 of {\em LNCS}. Springer,
  2004.

\bibitem{DBLP:conf/cav/BerdineCI11}
J.~Berdine, B.~Cook, and S.~Ishtiaq.
\newblock Slayer: Memory safety for systems-level code.
\newblock In G.~G. andShaz Qadeer, editor, {\em Computer Aided Verification -
  23rd International Conference, {CAV} 2011, Snowbird, UT, USA, July 14-20,
  2011. Proceedings}, volume 6806 of {\em LNCS}, pages 178--183. Springer,
  2011.

\bibitem{BrotherstonSimpson11}
J.~Brotherston and A.~Simpson.
\newblock Sequent calculi for induction and infinite descent.
\newblock {\em Journal of Logic and Computation}, 21(6):1177--1216, December
  2011.

\bibitem{DBLP:conf/nfm/CalcagnoDDGHLOP15}
C.~Calcagno, D.~Distefano, J.~Dubreil, D.~Gabi, P.~Hooimeijer, M.~Luca, P.~W.
  O'Hearn, I.~Papakonstantinou, J.~Purbrick, and D.~Rodriguez.
\newblock Moving fast with software verification.
\newblock In K.~Havelund, G.~J. Holzmann, and R.~Joshi, editors, {\em {NASA}
  Formal Methods - 7th International Symposium, {NFM} 2015, Pasadena, CA, USA,
  April 27-29, 2015, Proceedings}, volume 9058 of {\em LNCS}, pages 3--11.
  Springer, 2015.

\bibitem{DBLP:journals/jacm/CalcagnoDOY11}
C.~Calcagno, D.~Distefano, P.~W. O'Hearn, and H.~Yang.
\newblock Compositional shape analysis by means of bi-abduction.
\newblock {\em J. {ACM}}, 58(6):26:1--26:66, 2011.

\bibitem{CalcagnoYangOHearn01}
C.~Calcagno, H.~Yang, and P.~W. O’hearn.
\newblock Computability and complexity results for a spatial assertion language
  for data structures.
\newblock In {\em FST TCS 2001, Proceedings}, pages 108--119. Springer, 2001.

\bibitem{cook-haase-ouaknine-parkinson-worell11}
B.~Cook, C.~Haase, J.~Ouaknine, M.~J. Parkinson, and J.~Worrell.
\newblock Tractable reasoning in a fragment of separation logic.
\newblock In {\em Proc. of CONCUR'11}, volume 6901 of {\em LNCS}. Springer,
  2011.

\bibitem{DemriGalmicheWendlingMery14}
S.~Demri, D.~Galmiche, D.~Larchey-Wendling, and D.~M{\'e}ry.
\newblock Separation logic with one quantified variable.
\newblock In {\em {CSR}'14}, volume 8476 of {\em LNCS}, pages 125--138.
  Springer, 2014.

\bibitem{DBLP:journals/eceasst/DoddsP08}
M.~Dodds and D.~Plump.
\newblock From hyperedge replacement to separation logic and back.
\newblock {\em Electron. Commun. Eur. Assoc. Softw. Sci. Technol.}, 16, 2008.

\bibitem{DBLP:conf/cav/DudkaPV11}
K.~Dudka, P.~Peringer, and T.~Vojnar.
\newblock Predator: {A} practical tool for checking manipulation of dynamic
  data structures using separation logic.
\newblock In G.~Gopalakrishnan and S.~Qadeer, editors, {\em Computer Aided
  Verification - 23rd International Conference, {CAV} 2011, Snowbird, UT, USA,
  July 14-20, 2011. Proceedings}, volume 6806 of {\em LNCS}, pages 372--378.
  Springer, 2011.

\bibitem{DBLP:conf/lpar/EchenimIP20}
M.~Echenim, R.~Iosif, and N.~Peltier.
\newblock Entailment checking in separation logic with inductive definitions is
  2-exptime hard.
\newblock In {\em {LPAR} 2020: 23rd International Conference on Logic for
  Programming, Artificial Intelligence and Reasoning, Alicante, Spain, May
  22-27, 2020}, volume~73 of {\em EPiC Series in Computing}, pages 191--211.
  EasyChair, 2020.

\bibitem{EIP21a}
M.~Echenim, R.~Iosif, and N.~Peltier.
\newblock Decidable entailments in separation logic with inductive definitions:
  Beyond establishment.
\newblock In {\em {CSL} 2021: 29th International Conference on Computer Science
  Logic}, EPiC Series in Computing. EasyChair, 2021.

\bibitem{EP22b}
M.~Echenim and N.~Peltier.
\newblock Two results on separation logic with theory reasoning.
\newblock In {\em ASL 2022 (Workshop on Advancing Separation Logic)}, 2022.
\newblock \url{https://arxiv.org/abs/2206.09389}.

\bibitem{DBLP:journals/fmsd/EneaLSV17}
C.~Enea, O.~Leng{\'{a}}l, M.~Sighireanu, and T.~Vojnar.
\newblock Compositional entailment checking for a fragment of separation logic.
\newblock {\em Formal Methods Syst. Des.}, 51(3):575--607, 2017.

\bibitem{spen}
C.~Enea, M.~Sighireanu, and Z.~Wu.
\newblock On automated lemma generation for separation logic with inductive
  definitions.
\newblock In {\em ATVA 2015, Proceedings}, pages 80--96, 2015.

\bibitem{DBLP:journals/logcom/GalmicheM21}
D.~Galmiche and D.~M{\'{e}}ry.
\newblock Labelled cyclic proofs for separation logic.
\newblock {\em J. Log. Comput.}, 31(3):892--922, 2021.

\bibitem{IosifRogalewiczSimacek13}
R.~Iosif, A.~Rogalewicz, and J.~Simacek.
\newblock The tree width of separation logic with recursive definitions.
\newblock In {\em Proc. of CADE-24}, volume 7898 of {\em LNCS}, 2013.

\bibitem{DBLP:conf/atva/IosifRV14}
R.~Iosif, A.~Rogalewicz, and T.~Vojnar.
\newblock Deciding entailments in inductive separation logic with tree
  automata.
\newblock In F.~Cassez and J.~Raskin, editors, {\em {ATVA} 2014, Proceedings},
  volume 8837 of {\em LNCS}, pages 201--218. Springer, 2014.

\bibitem{IshtiaqOHearn01}
S.~S. Ishtiaq and P.~W. O'Hearn.
\newblock Bi as an assertion language for mutable data structures.
\newblock In {\em ACM SIGPLAN Notices}, volume~36, pages 14--26, 2001.

\bibitem{DBLP:conf/gg/JansenGN14}
C.~Jansen, F.~G{\"{o}}be, and T.~Noll.
\newblock Generating inductive predicates for symbolic execution of
  pointer-manipulating programs.
\newblock In H.~Giese and B.~K{\"{o}}nig, editors, {\em {ICGT} 2014}, volume
  8571 of {\em LNCS}, pages 65--80. Springer, 2014.

\bibitem{DBLP:conf/vmcai/Le21}
Q.~L. Le.
\newblock Compositional satisfiability solving in separation logic.
\newblock In F.~Henglein, S.~Shoham, and Y.~Vizel, editors, {\em Verification,
  Model Checking, and Abstract Interpretation - 22nd International Conference,
  {VMCAI} 2021, Copenhagen, Denmark, January 17-19, 2021, Proceedings}, volume
  12597 of {\em Lecture Notes in Computer Science}, pages 578--602. Springer,
  2021.

\bibitem{NTKY2018}
K.~Nakazawa, M.~Tatsuta, D.~Kimura, and M.~Yamamura.
\newblock {Cyclic Theorem Prover for Separation Logic by Magic Wand}.
\newblock In {\em ADSL 18 (First Workshop on Automated Deduction for Separation
  Logics)}, July 2018.
\newblock Oxford, United Kingdom.

\bibitem{DBLP:conf/csl/OHearnRY01}
P.~W. O'Hearn, J.~C. Reynolds, and H.~Yang.
\newblock Local reasoning about programs that alter data structures.
\newblock In L.~Fribourg, editor, {\em Computer Science Logic, 15th
  International Workshop, {CSL} 2001. 10th Annual Conference of the EACSL,
  Paris, France, September 10-13, 2001, Proceedings}, volume 2142 of {\em
  LNCS}, pages 1--19. Springer, 2001.

\bibitem{PMZ20}
J.~Pagel, C.~Matheja, and F.~Zuleger.
\newblock Complete entailment checking for separation logic with inductive
  definitions, 2020.

\bibitem{PZ20}
J.~Pagel and F.~Zuleger.
\newblock Beyond symbolic heaps: Deciding separation logic with inductive
  definitions.
\newblock In {\em LPAR-23}, volume~73 of {\em EPiC Series in Computing}, pages
  390--408. EasyChair, 2020.

\bibitem{DBLP:conf/aplas/PerezR13}
J.~A.~N. P{\'{e}}rez and A.~Rybalchenko.
\newblock Separation logic modulo theories.
\newblock In C.~Shan, editor, {\em Programming Languages and Systems - 11th
  Asian Symposium, {APLAS} 2013, Melbourne, VIC, Australia, December 9-11,
  2013. Proceedings}, volume 8301 of {\em LNCS}, pages 90--106. Springer, 2013.

\bibitem{DBLP:conf/cav/PiskacWZ13}
R.~Piskac, T.~Wies, and D.~Zufferey.
\newblock Automating separation logic using {SMT}.
\newblock In N.~Sharygina and H.~Veith, editors, {\em Computer Aided
  Verification - 25th International Conference, {CAV} 2013, Saint Petersburg,
  Russia, July 13-19, 2013. Proceedings}, volume 8044 of {\em LNCS}, pages
  773--789. Springer, 2013.

\bibitem{DBLP:conf/pldi/Qiu0SM13}
X.~Qiu, P.~Garg, A.~Stefanescu, and P.~Madhusudan.
\newblock Natural proofs for structure, data, and separation.
\newblock In H.~Boehm and C.~Flanagan, editors, {\em {ACM} {SIGPLAN} {PLDI}
  '13}, pages 231--242. {ACM}, 2013.

\bibitem{Reynolds02}
J.~Reynolds.
\newblock {Separation Logic: A Logic for Shared Mutable Data Structures}.
\newblock In {\em Proc. of LICS'02}, 2002.

\bibitem{DBLP:journals/ipl/Seidl94}
H.~Seidl.
\newblock Haskell overloading is dexptime-complete.
\newblock {\em Inf. Process. Lett.}, 52(2):57--60, 1994.

\bibitem{DBLP:conf/aplas/TatsutaNK19}
M.~Tatsuta, K.~Nakazawa, and D.~Kimura.
\newblock Completeness of cyclic proofs for symbolic heaps with inductive
  definitions.
\newblock In A.~W. Lin, editor, {\em Programming Languages and Systems - 17th
  Asian Symposium, {APLAS} 2019, Nusa Dua, Bali, Indonesia, December 1-4, 2019,
  Proceedings}, volume 11893 of {\em LNCS}, pages 367--387. Springer, 2019.

\bibitem{DBLP:conf/cade/XuCW17}
Z.~Xu, T.~Chen, and Z.~Wu.
\newblock Satisfiability of compositional separation logic with tree predicates
  and data constraints.
\newblock In L.~de~Moura, editor, {\em {CADE} 26}, volume 10395 of {\em LNCS},
  pages 509--527. Springer, 2017.

\end{thebibliography}
\end{document}